\documentclass{patmorin}
\usepackage{pat}
\usepackage[T1]{fontenc}
\usepackage[utf8]{inputenc}
\usepackage[inline]{enumitem}
\usepackage[normalem]{ulem}
\usepackage{nicefrac}%

\usepackage{mathtools}  %
\usepackage[longnamesfirst,numbers,sort&compress]{natbib}

\crefformat{equation}{#2(#1)#3}
\let\eqref\cref
\crefformat{subsection}{Subsection #2#1#3}
\crefformat{subsubsection}{Subsubsection #2#1#3} 

\renewcommand{\ge}{\geqslant}
\renewcommand{\le}{\leqslant}
\renewcommand{\geq}{\geqslant}
\renewcommand{\leq}{\leqslant}

\definecolor{brightmaroon}{rgb}{0.76, 0.13, 0.28}
\definecolor{linkblue}{rgb}{0, 0.337, 0.227}
\newcommand{\defin}[1]{\emph{\textcolor{brightmaroon}{#1}}}
\makeatletter
\def\mathcolor#1#{\@mathcolor{#1}}
\def\@mathcolor#1#2#3{%
  \protect\leavevmode
  \begingroup
    \color#1{#2}#3%
  \endgroup
}
\makeatother
\newcommand{\mathdefin}[1]{\mathcolor{brightmaroon}{#1}}

\newcommand{\torso}[2]{{#1}\langle{#2}\rangle}
\DeclareMathOperator{\height}{height}
\DeclareMathOperator{\depth}{depth}

\DeclareMathOperator{\lca}{lca}
\DeclareMathOperator{\lgg}{\hat{lg}}
\DeclareMathOperator{\bin}{bin}

\newcommand{\calG}{\mathcal{G}}
\newcommand{\Oh}{O}

\DeclarePairedDelimiter{\set}{\{}{\}}

\DeclarePairedDelimiter{\ceil}{\lceil}{\rceil}

\title{\MakeUppercase{\boldmath  Adjacency labelling for proper minor-closed graph classes}}

 \author{
 Vida Dujmovi{\'c}\,\footnote{School of Computer Science and Electrical Engineering, University of Ottawa, Ottawa, Canada (\texttt{vida.dujmovic@uottawa.ca}). Research supported by NSERC and a University of Ottawa Research Chair.}
 \qquad
 Cyril Gavoille\footnote{LaBRI, University of Bordeaux, France (\texttt{gavoille@labri.fr}). This work is supported by the French ANR Projets ENEDISC (ANR-24-CE48-7768-01) and TEMPOGRAL (ANR-22-CE48-0001).}
 \qquad
  Gwena\"el Joret\footnote{D\'epartement d'Informatique, Universit\'e libre de Bruxelles, Belgium ({\tt gwenael.joret@ulb.be}). G.\ Joret is supported by the Belgian National Fund for Scientific Research (FNRS) and by the Australian Research Council.}\\[1ex]
Piotr Micek\footnote{Department of Theoretical Computer Science, Faculty of Mathematics and Computer Science, Jagiellonian University, Kraków, Poland (\texttt{piotr.micek@uj.edu.pl}). Research supported by the National Science Center of Poland under grant UMO-2023/05/Y/ST6/00079 within the WEAVE-UNISONO program.}
 \qquad
 Pat Morin\footnote{School of Computer Science, Carleton University, Ottawa, Canada (\texttt{morin@scs.carleton.ca}). Research supported by NSERC.}
 \qquad
 David~R.~Wood\footnote{School of Mathematics, Monash University, Melbourne, Australia (\texttt{david.wood@monash.edu}). Research supported by the Australian Research Council and by NSERC.}
 }

\date{}

\begin{document}
\maketitle

\begin{abstract}
  We show that every proper minor-closed class of graphs admits a $(1+o(1))\log_2 n$-bit adjacency labelling scheme. Equivalently, for every proper minor-closed class $\mathcal{G}$ and every positive integer $n$ there exists an $n^{1+o(1)}$-vertex graph $U$ such that every $n$-vertex graph in $\mathcal{G}$ is isomorphic to an induced subgraph of $U$.   Both results are optimal up to the lower order term. They generalize the corresponding results for planar graphs and apex-minor-free classes (Dujmovi\'c et al., J.~ACM 2021) to all proper minor-closed classes, answering the open question raised in that paper and anticipated earlier by Bonamy, Gavoille, and Pilipczuk (SODA 2020).

\end{abstract}

\bigskip
\section{Introduction}
Let $\calG$ be a class of graphs and let $f:\N\to\N$ be a function.
We say that $\calG$ admits an \defin{$f(n)$-bit adjacency labelling scheme} if there exists a function $A:(\set{0,1}^*)^2\to\set{0,1}$ such that for all positive integers $n$, for every $n$-vertex graph $G\in\calG$, there exists a function $\ell: V(G)\to\set{0,1}^*$ such that
$|\ell(v)|\leq f(n)$ for each vertex $v$ in $G$, and such that
for every two vertices $u$, $v$ in $G$,
\[
  A(\ell(u),\ell(v))=\begin{cases}
    0 & \textrm{if $uv\notin E(G)$}, \\
    1 & \textrm{if $uv\in E(G)$}.
  \end{cases}
\]

A graph $H$ is a \defin{minor} of a graph $G$ if a graph isomorphic to $H$ can
be obtained from a subgraph of $G$ by contracting edges. A class of graphs
$\mathcal{G}$ is \defin{minor-closed} if every minor of every graph in
$\mathcal{G}$ is also in $\mathcal{G}$, and it is \defin{proper} if it is not the class of all graphs. Some examples of proper proper minor-closed include planar graphs, graphs of treewidth at most $k$, graphs embeddable in a fixed
surface, linklessly embeddable graphs, knotlessly embeddable graphs and, for every fixed graph $X$, the class of graphs with no $X$-minor. Note that every proper minor-closed class is contained in the class of $K_t$-minor-free graphs for some fixed $t$.

In this paper we prove the following result. (All logarithms are in base $2$.)

\begin{thm}
  \label{main_result}
  Every proper minor-closed class of graphs admits
  a $(1+o(1))\log n$-bit adjacency labelling scheme.
\end{thm}

Note that all the dependence on the fixed proper minor-closed class in~\cref{main_result} is in the $o(\log n)$ term. Also, \cref{main_result} is optimal up to the $o(\log n)$ term,
which is $\Oh\left((\log n)^{3/4}\right)$.
The proof of~\cref{main_result} is constructive. For every fixed proper minor-closed class $\calG$, there is a polynomial-time algorithm that takes a graph $G\in\calG$ as input and constructs the labelling $\ell:V(G)\to\{0,1\}^*$.

A consequence\footnote{In fact, the viewpoints of adjacency labelling schemes and induced-universal graphs are essentially equivalent. A minor detail is that the labelling schemes need to be injective for the equivalence to hold but this is the case for those developed in this paper. See~\cite[Section~2.1]{spinrad:efficient} for more details on the connection between the two notions.} of \cref{main_result} is the existence of induced-universal graphs with a near-linear number of vertices for $n$-vertex graphs belonging to a fixed proper minor-closed class of graphs.

\begin{cor}\label{main_result_induced_universal}
  For every proper minor-closed class $\calG$, for every positive integer $n$,
  there exists a graph $U$ with $n^{1+o(1)}$ vertices such that every $n$-vertex graph in $\calG$ is isomorphic to an induced subgraph of $U$.
\end{cor}

\subsection{State of the Art}
\label{ssec:state-of-art}

Adjacency labelling schemes were introduced in the late 1980s by \citet{kannan.naor.ea:implicit} and independently in the PhD thesis of \citet{muller:local}. Since this initial work, labelling schemes have been developed for many other queries besides adjacency, including distances \citep{GAVOILLE200485}, ancestry and nearest common ancestors in
trees \citep{AKM01,alstrup.rauhe:improved,DBLP:journals/mst/AlstrupGKR04}, reachability in planar
digraphs \citep{DBLP:journals/jacm/Thorup04,DBLP:conf/focs/HolmRT15}, and flows and connectivity \citep{DBLP:journals/siamcomp/KatzKKP04}. A review of this vast literature is beyond the scope of this paper. Here we review results most relevant to the current work, focusing mainly on trees, bounded treewidth graphs, and planar graphs.

Adjacency labelling schemes have also been studied for other general families of graph classes. At the dense end of the spectrum, \citet{DBLP:conf/stoc/BonamyEGS21} constructed
asymptotically optimal adjacency labelling schemes for every hereditary class
containing $2^{\Omega(n^2)}$ $n$-vertex graphs. In particular, they proved that comparability graphs admit a labelling with labels of length $n/4+o(n)$, which is best possible up to the lower-order term. Proper minor-closed classes lie at the opposite, sparse end of the spectrum, where the situation is markedly different, as we now discuss.

We start with a simple example to better grasp the notion of labelling schemes. 
Given an $n$-vertex tree $T$, fix an arbitrary vertex to be the root of $T$, and assign each vertex of $T$ a unique identifier, i.e., a number in $\set{0,\ldots,n-1}$.
For each non-root vertex $v$ of $T$,
let the label of $v$ be
the pair consisting of its identifier and the identifier of its parent.
For the root vertex of $T$, let its label be
the pair consisting of two copies of its identifier.
Clearly, given two labels of vertices in $T$,
a test of adjacency is simply to verify whether an identifier of one vertex is equal to a stored identifier of the parent of the other vertex.
This constitutes a $\lceil\log(n^2)\rceil = \lceil2\log{n}\rceil$-bit adjacency labelling scheme for trees.
There is a fascinating series of results improving on this simple scheme for
trees. In 1990, a result of \citet{chung:universal} about induced-universal
graphs implies that $n$-vertex trees admit an adjacency labelling scheme with
labels of length $\log n + \log\log n + \Oh(1)$. In 2002, \citet{AR02} devised
an improved scheme with labels of length $\log n + \Oh(\log^* n)$. Finally, in
2015, \citet{alstrup.dahlgaard.ea:optimal} gave a $(\log n + \Oh(1))$-bit
adjacency labelling scheme for trees, which is optimal up to the $\Oh(1)$ term. Indeed, consider an $n$-vertex graph with no two vertices having the same neighborhood (e.g., a path for $n\geq4$). In any adjacency labelling scheme, each vertex must be assigned a distinct label and therefore one label must be of length at least  $\log (n+1)-1$.

In 2007, \citet{gavoille.labourel:shorter} presented
a $(1+o(1))\log n$-bit adjacency labelling scheme for graphs of bounded treewidth.
This is particularly relevant to the topic of this paper because of the following famous application of the graph minor structure theorem:
\citet{devos2004} showed that every graph
in a proper minor-closed class can be edge $2$-coloured so that each monochromatic subgraph
has bounded treewidth.
Therefore, given a proper minor-closed class of graphs $\calG$,
for each $G$ in $\calG$, fix such a $2$-colouring of the edges of $G$, say with red and blue, and label each vertex of $G$ by the concatenation of the two labels given by the adjacency labelling schemes for the red subgraph of $G$ and for the blue subgraph of $G$. Thus, all labels are of length $(2+o(1))\log n$.
Given two labels, to test adjacency simply test whether there is a blue edge or whether there is a red edge.
This gives a $(2+o(1))\log n$-bit adjacency labelling scheme for a fixed proper minor-closed class of graphs, which remained the best-known scheme up to the present work.

Planar graphs are a prominent example of a proper minor-closed class of graphs and there is a long history of research
on their adjacency labelling schemes.
Since planar graphs are $5$-degenerate,
they admit a simple $\lceil6\log n\rceil$-bit adjacency labelling scheme, which was already observed by \citet{muller:local}.
\citet{kannan.naor.ea:implicit} use a similar approach that makes use of the fact that planar graphs have arboricity~3 (so their edges can be partitioned into three forests \citep{nash-williams:edge-disjoint}) to devise an adjacency labelling scheme for planar graphs whose labels have length at most $\lceil4\log n\rceil$. This length was reduced to $3\log{n} + \Oh(\log^*{n})$ by the improved scheme of \citet{AR02} for arboricity-$k$ graphs. 
Of course, the later result of \citet{gavoille.labourel:shorter} also applies to planar graphs, and gives a scheme with labels of length $(2+o(1))\log n$.
The most recent breakthrough came with the application of the product structure theorem for planar graphs by \citet{dujmovic.joret.ea:planar}.
In 2020, \citet{BGP22} recognized that the product structure is particularly useful for adjacency labelling, breaking the $(2+o(1))\log n$  barrier by introducing a scheme with labels of length $(\nicefrac{4}{3}+o(1))\log n$.
Finally, in 2020 \citet{dujmovic.esperet.ea:adjacency} used the product structure theorem to give a $(1+o(1))\log n$-bit adjacency labelling scheme for planar graphs, which is optimal up to the lower order term. In \citet{dujmovic.esperet.ea:adjacency}, the lower order term is in $\Oh(\sqrt{\log n\log\log n})$.  \citet{gawrychowski.janczewski:simpler} later introduced a simplification that reduces the lower order term to  $\Oh(\sqrt{\log n})$ and allows for the adjacency testing function to be implemented in constant time in a realistic model of computation.

The adjacency labelling scheme of \citet{dujmovic.esperet.ea:adjacency} works for  graphs embeddable in any fixed surface, and more generally for any minor-closed class excluding a fixed apex graph as a minor. Here a graph $X$ is \defin{apex} if $X$ can be made planar by the removal of at most one vertex. Since $K_5$ is apex, this implies a $(1+o(1))\log n$-bit adjacency labelling scheme for $K_5$-minor-free graphs. Such apex-minor-free classes are the limit of this approach, since a minor-closed class $\mathcal{G}$ has `product structure' if and only if some apex graph is not in $\mathcal{G}$. Indeed, the existence of a $(1+o(1))\log n$-bit adjacency labelling scheme for $K_t$-minor-free graphs is stated as an open problem by \citet[Section~6, Problem~2]{dujmovic.esperet.ea:adjacency}. \citet{BGP22} likewise anticipated that such labelling schemes might extend to all proper minor-closed classes via the graph minor structure theorem of Robertson and Seymour, and identified the difficulty: ``we do not know how to combine labeling schemes along tree decompositions''.
 Prior to the present work, the best result for $K_t$-minor-free graphs, $t > 5$, was the $(2+o(1))\log n$-bit adjacency labelling scheme of \citet{gavoille.labourel:shorter}, who also conjectured that labels of length $\log n + \Oh(1)$ suffice for
every proper minor-closed class of graphs.

Finally we discuss the relationship between our result and the Implicit Graph Conjecture. A hereditary graph class is \defin{small} if it has $2^{O(n\log n)}$ labelled $n$-vertex graphs for all $n$. For example, every proper minor-closed graph class is small~\citep{NSTW06}. The Implicit Graph Conjecture~\citep{kannan.naor.ea:implicit} asserted that every small graph class admits an $O(\log n)$-bit adjacency labelling scheme. Proper minor-closed classes satisfy this conjecture, and indeed they admit labels of only $(1+o(1))\log n$ bits by \cref{main_result}. The conjecture in general was disproved by \citet{HH22}, who in fact showed there exist small graph classes for which every adjacency labelling scheme needs $\Omega(n^{1/2-\epsilon})$ bits. See \citep{Bonnet-etal-ICALP24,Bonnet-etal-SJC24,Alon24} for recent extensions. In light of the counterexamples of \cite{HH22}, it remains an interesting problem to determine which hereditary classes admit $O(\log n)$-bit, or even $(1+o(1))\log n$-bit, adjacency labelling schemes. Our result shows that excluding a fixed minor suffices for the strongest conclusion.

\subsection{Setup and Proof Overview}
\label{KeyTools}

In this paper, all graphs are finite, simple, and undirected.  For a graph $G$, \defin{$V(G)$} denotes the vertex-set of $G$, \defin{$E(G)$} denotes the edge-set of $G$.  A \defin{clique} in a graph is a non-empty set of pairwise adjacent vertices.   Let $G$ be a graph and $X\subseteq V(G)$. By $G[X]$ we denote the subgraph of $G$ induced on $X$. For other standard graph theory terms, we refer the reader to the textbook by \citet{diestel:graph}.

A \defin{tree-decomposition} of a graph $G$ is a pair $(T,(B_x \mid x \in V(T)))$, where $T$ is a tree and $B_x \subseteq V(G)$ for every $x \in V(T)$, with the following properties:
\begin{enumerate*}[label=(\alph*)]
  \item for each vertex $u$ in $G$, the subgraph of $T$ induced by $\{x \in V(T) \mid u \in B_x\}$ is non-empty and connected; and
  \item for each edge $uv$ in $G$, there exists $x\in V(T)$ such that $u,v\in B_x$.
\end{enumerate*}
We call the sets $B_x$ the \defin{bags} of the tree-decomposition and, for each edge $xy$ of $T$, $B_x\cap B_y$ is an \defin{adhesion} of the tree-decomposition.  When $T$ is rooted and $x$ is the parent of $y$, $B_x\cap B_y$ is the \defin{parent adhesion} of $y$. The \defin{parent adhesion} of the root of $T$ is the empty set. The \defin{home} of a vertex $v$ of $G$ is the unique $x\in V(T)$ such that $v\in B_x$ and $v$ is not in the parent-adhesion of $x$. The \defin{adhesion-width} of a tree-decomposition is the maximum size of an adhesion. For a tree-decomposition $\mathcal{T}:=(T,(B_x \mid x\in V(T)))$ of a graph $G$, for each  $x\in V(T)$, the \defin{torso} of $x$, denoted by \defin{$\torso{G}{\mathcal{T},B_x}$}, is the graph with vertex-set $B_x$ in which two distinct vertices $u$ and $v$ are adjacent if $uv \in E(G)$, or there exists $y \in N_T(x)$  such that $u,v \in B_x \cap B_y$.

The \defin{strong product} of graphs~$A$ and~$B$, denoted by~${A \boxtimes B}$, is the graph with vertex-set ${V(A) \times V(B)}$, where distinct vertices ${(v,x),(w,y) \in V(A) \times V(B)}$ are adjacent if
${v=w}$ and ${xy \in E(B)}$, or
${x=y}$ and ${vw \in E(A)}$, or
${vw \in E(A)}$ and~${xy \in E(B)}$.
For an integer $k\geq 0$, let \defin{$\mathcal{R}_{k}$} be the class of graphs isomorphic to a subgraph of $H\boxtimes P$ for some graph $H$ with treewidth at most $k$ and for some path $P$.  The main result of \citet{dujmovic.esperet.ea:adjacency} is a $(1+o(1))\log n$-bit adjacency labelling scheme for the class $\mathcal{R}_k$, for any fixed $k$.

For any graph class $\mathcal{G}$, let \defin{$\mathcal{G}^{+a}$} be the class of graphs $G$ such that $G-X$ is in $\mathcal{G}$ for some $X\subseteq V(G)$ with $|X|\leq a$.  It is a simple exercise to show that any $(1+o(1))\log n$-bit adjacency labelling scheme for a graph class $\mathcal{G}$ implies the existence of a $(1+o(1))\log n$-bit adjacency labelling scheme for $\mathcal{G}^{+a}$ (see \cref{sec:adding_apexes} for a more general result).  In particular, for any fixed $k,a\ge 0$, there exists a $(1+o(1))\log n$-bit adjacency labelling scheme for $\mathcal{R}_k^{+a}$. 

Our starting point is the following variant of the Graph Minor Structure Theorem:

\begin{thm}[\protect Graph Minor Product Structure Theorem~\citep{dujmovic.joret.ea:planar}]
  \label{gmpst}
  For every proper minor-closed class $\mathcal{G}$ there exist integers $k,a\geq 0$ such that every graph in $\mathcal{G}$ has a tree-decomposition in which every torso is in $\mathcal{R}_k^{+a}$.
\end{thm}

We now have enough background to sketch the proof of our main result.  We begin by describing natural approaches to the problem, where these approaches have issues, and then explain how we overcome these issues. By the above-mentioned results, it suffices to consider a graph class $\mathcal{G}$ with the following property:
For every $G\in\mathcal{G}$, there exists a tree-decomposition $\mathcal{T}:=(T,(B_x\mid x\in V(T)))$ of $G$ such that each torso of $\mathcal{T}$ comes from a monotone graph class $\mathcal{G}'$ that has a $((1+o(1))\log n)$-bit adjacency labelling scheme.  Let $A'$ be the adjacency tester for the labelling scheme on $\mathcal{G}'$. 

Let $G$ be an $n$-vertex graph in $\mathcal{G}$ and let $\mathcal{T}:=(T,(B_x\mid x\in V(T)))$ be a tree-decomposition of $G$ such that each torso of $\mathcal{T}$ is in the class $\mathcal{G}'$.  Each edge of $G$ is present in $G[B_x]$ for at least one $x\in V(T)$. Thus, a natural first approach is to compute an adjacency labelling of $G[B_x]$ for $A'$, for each $x\in V(T)$.  The problem with this approach is that a vertex $v$ of $G$ can appear in many bags of $\mathcal{T}$, so it can have many labels. Simply concatenating these labels would result in labels of length much greater than $\log n$.

\paragraph{Assigning homes:} For each $x\in V(T)$, let $A_x$ be the parent adhesion of $x$ in $\mathcal{T}$ and let $B^-_x:=B_x\setminus A_x$.  (Thus $x$ is the home of each $v\in B^-_x$.) Then $\{B^-_x\mid x\in V(T)\}$ is a partition of $V(G)$.  A natural second attempt is to compute an adjacency labelling $\ell_x$ of $G[B^-_x]$ for $A'$, for each $x\in V(T)$.  This leaves two closely-related issues:
\begin{enumerate}[nosep,nolistsep]
  \item For two vertices $v\in B^-_x$ and $w\in B^-_y$ with $x\neq y$, the adjacency test $A'(\ell_x(v),\ell_y(w))$ does not produce a useful result, since $\ell_x$ and $\ell_y$ are labellings of different graphs. To avoid misusing $A'$ this way, we need a way of determining that $x\neq y$ from the labels of $v$ and $w$.
  \item In the case where $x\neq y$, we need an alternative way of testing if $v$ and $w$ are adjacent, since $vw$ is not present in $\bigcup_{x\in V(T)}G[B^-_x]$ even when $vw\in E(G)$.
\end{enumerate}

These two issues make precise the difficulty that was also identified by \citet{BGP22}, namely
that of combining labellings along a tree-decomposition. The tools we develop
next are designed to overcome it.

\paragraph{2-Part labels:}
To deal with the first issue, we can assign each vertex $v$ a 2-part label $\langle x(v),\ell_x(v) \rangle$, where $x(v)$ is a unique identifier for the node $x$ of $T$ such that $v\in B^-_x$.  Doing this naïvely could easily lead to a label length of $2\log n + o(\log n)$.  This can be avoided by using a variable-length code so that $|x(v)|=\log n-\log|B^-_{x(v)}|+\Oh(1)$.  This way, $|x(v)|+|\ell_x(v)|=|x(v)|+\log|B^-_{x(v)}|+o(\log n)=\log n + o(\log n)$.  This solves the first of the two problems above, but still leaves us with the second problem.  There is still no way of testing adjacency between vertices with different homes.

\paragraph{Mixed labelling schemes:} An important contribution of the current work is to introduce a generalization of adjacency labelling schemes, called \defin{weighted mixed labelling schemes}, that attempts to solve both of these problems simultaneously. The weights play an important but technical role. For the simplicity of this overview we omit the role of the weights for now, speaking only of \emph{mixed labelling schemes}, and return to it at the end of the overview when stating our main technical result. The formal definition appears at the beginning of \cref{weighted_mixed_section}. Informally, a mixed labelling $\mu$ is defined for a pair of graphs $Q^+$ and $Q$, where $Q$ is a spanning subgraph of $Q^+$, and $\mu$ defines labels for both vertices and cliques of $Q^+$. The vertex labels assigned by $\mu$ are used with an adjacency tester $A$ for the spanning subgraph $Q$.  That is, for any two vertices $v$ and $w$ of $Q$, $A(\mu(v),\mu(w)) = 1$ if and only if $vw\in E(Q)$.  In addition to the vertex and clique labels, a mixed labelling assigns, for each clique $K$ of $Q^+$ and each vertex $u\in K$, a short \defin{local identifier} $\kappa(K,u)$ for $u$ in $K$. These clique, vertex, and local identifiers are used with an \defin{identity tester} $I$.  For any clique $K$ of $Q^+$, any vertex $u\in K$, and any vertex $v$ of $Q$, $I(\mu(K),\kappa(K,u),\mu(v))=1$ if and only if $v=u$.

Loosely speaking, a class admits an \emph{efficient} mixed labelling scheme if it admits a $(1+o(1))\log n$-bit mixed labelling scheme. This notion of efficiency is defined formally in the introduction to \cref{weighted_mixed_section}. \Cref{weighted_product_proof} shows that the adjacency labelling scheme for graphs in $\mathcal{R}_k$ given by \citet{dujmovic.esperet.ea:adjacency} and its improvements by \citet{gawrychowski.janczewski:simpler} can be extended to give an efficient mixed labelling scheme for $\mathcal{R}_k$.  \Cref{add_apexes} shows that the addition of apex vertices to obtain the class $\mathcal{R}^{+a}_k$ can be accommodated by any mixed labelling scheme.  The resulting mixed labelling scheme has clique and vertex labels of length $\log n + \Oh(\sqrt{\log n})$ (see \cref{gk_labelling_precise}).

Thus, for every proper minor-closed graph family $\mathcal{G}$, every graph $G$ in $\mathcal{G}$ has a tree-decomposition $\mathcal{T}$ of bounded adhesion-width whose torsos come from a class that has an efficient mixed labelling scheme.  This is the only property of proper minor-closed graph classes that we use.

\paragraph{A solution for short tree-decompositions:}
We now describe a first attempt to use mixed labelling schemes to simultaneously address the questions of determining if two vertices of $G$ have the same home and testing adjacency between two vertices of $G$ with different homes.  We begin by constructing a mixed labelling $\mu_x$ for the pair of graphs $\torso{G}{\mathcal{T},B_x}-A_x$ and $G[B_x]-A_x$, for each $x\in V(T)$.  Then, the label $\mu(v)$ for a vertex $v$ in $G$ has the form
\[
  \mu(v):=\langle \mu_{x_0}(K_1),\ldots,\mu_{x_{p-1}}(K_p),\mu_{x_p}(v),\alpha(v)\rangle
\]
where $x_0,\ldots,x_p$ is the path from the root $x_0$ of $T$ to the home $x=x_p$ of $v$, and $K_i=A_{x_i}\setminus A_{x_{i-1}}$ for each $i\in\{1,\ldots,p\}$ is a clique in $\torso{G}{\mathcal{T},B_{x_{i-1}}}-A_{x_{i-1}}$.  Here, $\alpha(v)$ is a form of adjacency list that stores, for each $u\in N_G(v)\cap A_{x_p}$, the index $i$ of the clique $K_i$ that contains $u$ and the local identifier $\kappa_{x_{i-1}}(K_i,u)$.

Now, consider a vertex $v$ whose home is $x$ and a vertex $w$ whose home is $y$.  By examining the sequence of clique labels in $\mu(v)$ and $\mu(w)$, an adjacency tester can determine which of the following four cases applies:
\begin{itemize}[nosep,nolistsep]
  \item $x=y$: In this case, the tester can use $\mu_x(v)$ and $\mu_y(w)=\mu_x(w)$ to determine if $vw\in E(G[B_x]-A_x)$.
  \item $x$ is an ancestor of $y$:  In this case, the path $x_0,\ldots,x_q$ from the root of $T$ to $y$ contains a node $x_p=x$.  The only way in which $v$ and $w$ can be adjacent is if $v\in K_{p+1}$.  If $v$ and $w$ are adjacent, then $v\in N_G(w)\cap A_{x_q}$ and therefore the pair $(p+1,\kappa_{x_p}(K_{p+1}, v))$ is contained in $\alpha(w)$. The adjacency tester checks this by searching through the entries in $\alpha(w)$ looking for an entry $(p+1,\kappa_{x_p}(K_{p+1}, v'))$ such that $I(\mu_{x_p}(K_{p+1}),\kappa_{x_p}(K_{p+1},v'),\mu(v))=1$.
  \item $y$ is an ancestor of $x$: This is symmetric to the previous case.
  \item $x$ and $y$ are not in an ancestor-descendant relationship.  In this case $A_x$ separates $v$ and $w$, so $v$ and $w$ are not adjacent.
\end{itemize}

With a careful application of variable-length codes, we can ensure that the total length of the label $\mu(v)$ for a vertex $v$ whose home has depth\footnote{The \defin{depth} of a node $x$ in a rooted tree is the number of edges on the path from $x$ to the root of the tree. The \defin{height} of the tree is the maximum depth of a node of the tree.} $p$ is at most $\log n + p\cdot g_3(n)+ o(\log n)$, where $g_3(n)$, with  $1\leq g_3(n)\in o(\log n)$, is some overhead that is incurred with each step $x_{i-1}x_i$ in the path $x_0,\ldots,x_p$.  In particular, there is an additive overhead of $g_3(n)$ in the length of $\mu_{x_{i-1}}(K_i)$ for each $i\in\{1,\ldots,p\}$.

Unfortunately, mixed labelling schemes still do not completely solve the problem, because the length $p$ of the path $x_0,\ldots,x_p$ can only be upper bounded by the height of $T$, which may be $\Omega(n)$. To prove \cref{main_result} would require that $\height(T)\cdot g_3(n)=o(\log n)$.  Even when $g_3(n)=1$, this would require a tree-decomposition of height $o(\log n)$, which no variant of \cref{gmpst} can guarantee.

Nevertheless, this strategy gives a $((1+o(1))\log n)$-bit adjacency labelling scheme for the class of graphs whose $n$-vertex members have a tree-decomposition of bounded adhesion-width and height %
$o(\log n/g_3(n))$ whose torsos come from a graph class $\mathcal{G}$ that has an efficient mixed labelling scheme.  (The function $g_3(n)$ is determined by the mixed labelling scheme for $\mathcal{G}$.) This result is \cref{small_height} in \cref{sec:short}.

\paragraph{Short partition into skinny subtrees:}
To deal with the case where the tree $T$ in the tree-decomposition $\mathcal{T}$ does not have small height, we partition the vertices of $T$ into a collection of subtrees such that each root to leaf path of $T$ intersects at most $1+\log_b n$ of these subtrees, for a carefully chosen value of $b$.  This is done in such a way that each subtree in the decomposition, rooted at its node closest to the root of $T$, is \defin{$b$-skinny}: it has at most $b$ nodes
at depth $i$ (measured from its own root), for each integer $i\ge 0$. This partition is described in \cref{b_skinny_decomp} in \cref{lowpw}. By treating each subtree as a single giant bag, this gives a tree-decomposition $\mathcal{T}'$ of $G$ whose height is at most $\log_b n$, and in which each torso of $\mathcal{T}'$  has a $b$-skinny tree-decomposition whose bags are bags of $\mathcal{T}$. In particular, the torsos of this $b$-skinny tree-decomposition belong to the original class admitting an efficient mixed labelling scheme.  Note that at this point we do not yet know how to label the (typically large) bags of $\mathcal{T}'$ themselves. That is the purpose of the next step.

\paragraph{Skinny tree-decompositions:} The last ingredient needed to complete the proof is a method of handling graphs
that have $b$-skinny tree-decompositions whose torsos belong to a graph class that
admits an efficient mixed labelling scheme.  To use our solution for short
tree-decompositions, we must show that such graphs have an efficient mixed
labelling scheme. To do this, we group the nodes of the tree-decomposition $T$
into layers $L_1,\ldots,L_p$, where $L_i$ contains all nodes of depth $i-1$.  For
each $i\in\{1,\ldots,p\}$, a mixed labelling $\mu_i$ is computed for the pair of
graphs $\bigcup_{x\in L_i}(\torso{G}{\mathcal{T},B_x}-A_x)$ and
$\bigcup_{x\in L_i}(G[B_x]-A_x)$. We then reuse ideas of \citet{gavoille.labourel:shorter} and \citet{dujmovic.esperet.ea:adjacency} to handle adjacency testing between vertices in different homes and extend these to provide the remaining functionality (clique labels, local identifiers, and identity testing) required for a mixed labelling scheme.  The resulting mixed labelling scheme has labels of length $\log n + O(\log b) + o(\log n)$. 

\paragraph{Closing the loop:}
Putting the preceding pieces together with $g_3(n)=O(\sqrt{\log n})$ and setting $b=2^{(\log n)^{3/4}}$ shows
that any proper minor-closed graph class has an $f(n)$-bit adjacency labelling scheme where
\[
  f(n) := \log n + (\log_b n)\cdot O(\sqrt{\log n}) + O(\log b) = \log n + O((\log n)^{3/4}) \enspace .
\]
Without much extra work, the adjacency-labelling solution for short tree-decompositions can be extended to provide an efficient mixed labelling scheme.  Since this only relies on the underlying tree-decomposition having bounded adhesion-width and having torsos in a class $\mathcal{G}$ that has an efficient mixed labelling scheme, our main result is thus described concisely by the following theorem:

\begin{thm}\label{PlainRealMainTechnical}
  Let $\mathcal{G}$ be a hereditary graph class closed under taking disjoint union and admitting an efficient weighted mixed labelling scheme. Then, for any fixed integer $k\geq 0$, the class of graphs that have a tree-decomposition of adhesion-width at most $k$ whose torsos belong to $\mathcal{G}$ also admits an efficient weighted mixed labelling scheme.
\end{thm}

We can now say what the weights are for. For the recursion behind this theorem, plain mixed labelling schemes are not enough. We need a refinement in which each vertex $v$ carries a positive weight $\omega(v)$ and vertex labels satisfy the Kraft-type bound $|\mu(v)|\le\log\omega(V(G))-\log\omega(v)+o(\log n)$. Intuitively, weights let heavier vertices receive shorter labels; when recursing on a tree-decomposition, inflating the weight of a vertex according to the total weight hanging below it makes the label lengths telescope along root-to-leaf paths. This notion of a \defin{weighted mixed labelling scheme} (and the accompanying notion of \emph{efficiency}) is defined formally in \cref{weighted_mixed_section}.

A quantitative version of \cref{PlainRealMainTechnical} appears as \cref{RealMainTechnical} in \cref{sec:finalizing}.

\subsection{Paper Organization}

The remainder of the paper is organized as follows:  

\Cref{building_blocks} reviews some basic background material and shows how to
partition any tree-decomposition into $b$-skinny subtrees so that contracting
each subtree yields a tree-decomposition of height at most $\log_b n$, as outlined above.    \Cref{weighted_mixed_section} formally defines weighted mixed labelling schemes and proves \cref{PlainRealMainTechnical}.
We conclude in \Cref{conclusion} by explaining how the constructive nature of \cref{main_result} leads to a
polynomial-time labelling algorithm.

\section{Building blocks}\label{building_blocks}

Let \defin{$\N$} be the set of all nonnegative integers, and let \defin{$\N^+$} be the set of all positive integers.  For $p\in\N^+$, let $\mathdefin{[p]} := \set{1,\dots,p}$. 

A \defin{graph class} is a set of graphs closed under isomorphism. A graph class  $\mathcal{G}$ is \defin{hereditary} if for every $G$ in $\mathcal{G}$ every induced subgraph of $G$ is also in $\mathcal{G}$. A  graph class $\mathcal{G}$ is \defin{monotone} if for every $G\in\mathcal{G}$, every subgraph of $G$ is in $\mathcal{G}$. 
 A \defin{disjoint union} of two graphs $G$ and $H$ is a graph obtained from two vertex-disjoint copies of $G$ and $H$ with no edge in between. A class of graphs $\mathcal{C}$ is \defin{closed under taking disjoint union} if for all $G,H$ in $\mathcal{C}$, the disjoint union of $G$ and $H$ is also in $\mathcal{C}$.

If $G$ is a graph, $P$ is a path and $(P,(B_x\mid x \in V(P)))$ is a tree-decomposition of $G$, then $(P,(B_x\mid x \in V(P)))$ is called a \defin{path-decomposition} of $G$, which we usually denote as $(B_1,\ldots,B_{|V(P)|})$.

\subsection{Trees and Forests}

For a rooted tree $T$ with root $r\in V(T)$ and a node $x$ of $T$, define \defin{$P_T(x)$} to be the path in $T$ from $x$ to $r$.    A \defin{forest} $F$ is a graph (possibly disconnected) each of whose components is a tree.  If each component of a forest $F$ is a rooted tree, then $F$ is a \defin{rooted} forest. For a node $x$ in a rooted forest $F$, let $\mathdefin{P_F(x)}:=P_T(x)$ where $T$ is the component of $F$ that contains $x$.

Let $F$ be a rooted forest.  For each $x\in V(F)$, the \defin{$F$-depth} of $x$, denoted \defin{$\depth_F(x)$}, is the number of edges in $P_T(x)$. The \defin{height} of $F$, denoted $\mathdefin{\height(F)}:=\max_{x\in V(F)}\depth_F(x)$, is the maximum $F$-depth of a vertex in $F$. For each $y\in V(F)$ and each $x\in V(P_F(y))$ (including $y$),  the vertex $x$ is an \defin{$F$-ancestor} of $y$ and $y$ is a \defin{$F$-descendant} of $x$.  For each $x\in V(F)$, \defin{$F_x$} denotes the subtree of $F$ induced by all $F$-descendants of $x$. For each edge $xy$ of $F$ where $x$ is an $F$-ancestor of $y$, we say that $x$ is the \defin{$F$-parent} of $y$ and $y$ is an \defin{$F$-child} of $x$. We say that a set $S\subseteq V(F)$ is an \defin{$F$-antichain}, if there are no distinct $x,y\in S$ such that $x$ is a $F$-ancestor of $y$.  (When there is no danger of ambiguity, we may drop the leading $F$- from $F$-depth, $F$-ancestor, $F$-descendant, $F$-parent, and $F$-child.) We treat each subgraph $F'\subseteq F$, as a rooted forest in which each component $T$ of $F'$ is rooted at the node of $T$ having minimum $F$-depth.  For each $i\in\N^+$, the \defin{$i$-th layer} of $F$ is $\mathdefin{L_i(F)}:=\{x\in V(F):\depth_F(x)=i-1\}$. Thus, $L_1(F)$ is the set of roots of components of $F$.

For a rooted tree $T$ and any set $S\subseteq V(T)$, the \defin{lowest common $T$-ancestor} of $S$, denoted \defin{$\lca_T(S)$}, is the node in $\bigcap_{x\in S} V(P_T(x))$ of maximum $T$-depth.

\subsection{Tree-Decompositions and Forest-Decompositions}

It will be convenient to work with forest-decompositions, which are the natural generalization of a tree-decomposition that allow the graph indexing the bags to be a forest.  Specifically, $(F,(B_x\mid x\in V(F))$ is a \defin{forest-decomposition} of a graph $G$ if $F$ is a forest, for each edge $vw$ of $G$ there exists $x\in V(F)$ with $v,w\in B_x$, and for each vertex $v$ of $G$, $F[\{x\in V(F)\mid v\in B_x\}]$ is connected.  A forest-decomposition is \defin{rooted} if it is indexed by a rooted forest.  A rooted forest-decomposition $(F,(B_x\mid x\in V(F))$ of a graph $G$ is \defin{tidy} if
\begin{enumerate*}[label=(\arabic*)]
  \item $B_x\neq\emptyset$ for each $x\in V(F)$;
  \item $B_y\not\subseteq B_x$ for each $y\in V(F)$ with an $F$-parent $x$; and
  \item $B_z \cap B_y \not\subseteq B_y\cap B_x$ for each $z\in V(F)$ with $F$-parent $y$ and $F$-grandparent $x$.
\end{enumerate*}
Any rooted forest-decomposition $\mathcal{F}:=(T,(B_x\mid x\in V(T)))$ can be made into a rooted tidy forest-decomposition by
\begin{enumerate*}[label=(\arabic*)]
  \item removing every node $x$ of $T$ with $B_x=\emptyset$,
  \item removing every edge $xy$ of $T$ with $B_x\cap B_y=\emptyset$, and
  \item repeatedly replacing an edge $yz$ with the edge $xz$ if $B_z\cap B_y=B_x\cap B_y$ and $x$ is the parent of $y$.
\end{enumerate*} This transformation has the properties expressed in the following observation:
\begin{obs}\label{tidy_forest}
  Let $G$ be a graph and let $(F,(B_x\mid x\in V(F)))$ be a forest-decomposition of $G$. Then there exists a rooted forest $F'$ with $V(F')\subseteq V(F)$, $\height(F')\le\height(F)$, such that $\mathcal{F}':=(F',(B_x\mid x\in V(F')))$ is a rooted tidy forest-decomposition of $G$ and every torso of $\mathcal{F}'$ is a torso of $\mathcal{F}$.
\end{obs}
We make use of the fact that rooted tidy forest-decompositions remain tidy when we remove the root bags and their contents.  More precisely, let $\mathcal{F}:=(F,(B_x\mid x\in V(F)))$ be a rooted tidy forest-decomposition of a graph $G$, let $R$ be the set of roots in $F$, let $B_R:=\bigcup_{r\in R}B_r$, and let $C_x:=B_x\setminus B_R$ for each $x\in V(F)$.  Then $(F-R,(C_x\mid x\in V(F-R)))$ is a rooted tidy rooted forest-decomposition of $G-B_R$.

\subsection{Weight Functions}

For a non-empty set $S$, a \defin{weight function} over $S$ is a function $\omega:S\to\R^+$. For $X\subseteq S$, we use the shorthand $\mathdefin{\omega(X)}:=\sum_{x\in X}\omega(x)$.  When $S:=V(G)$ is the vertex-set of a graph $G$ and $H$ is a subgraph of $G$, we use the shorthand $\mathdefin{\omega(H)}:=\omega(V(H))$.

Let $\omega:S\to\R^+$ be a weight function.  For each $x\in S$, the \defin{(Kraft) ideal codeword length} for $x$ is $\log\omega(S)-\log\omega(x)$.
It will be convenient to place some upper bound on the Kraft ideal codeword length or, equivalently, on the ratio of total weight, $\omega(S)$ to minimum weight $\min_{x\in S}(\omega(x))$.  The following observation gives a way to do this without increasing the ideal codeword length of each $x\in S$ by more than a constant:

\begin{obs}\label{obs:nice}
  For every weight function $\omega:S\to\R^+$, there exists a weight function $\omega'(S)\to\N^+$ such that $\log\omega'(S)-\log\omega'(x)\le\min\{\log|S|,\;\log\omega(S)-\log\omega(x)\}+2$, for each $x\in S$.
\end{obs}

\begin{proof}
  By dividing every weight by $\min_{x\in S}\omega(x)$ we may assume, without loss of generality, that $\omega(x)\ge 1$ for each $x\in S$, since this does not change the value of the right-hand side in the inequality.
  For each $x\in S$, define $\omega'(x):=\ceil{\max\{\omega(S)/|S|,\; \omega(x)\}}$.  Then $\omega'(S)<\sum_{x\in S}(\omega(S)/|S|+\omega(x)+1)= 2\omega(S)+|S|\le 3\omega(S)$.  For each $x\in S$, the inequality $\omega'(x)\ge \omega(x)$ implies that $\log\omega'(S)-\log\omega'(x)\le\log \omega(S)-\log\omega(x)+\log 3$ and the inequality $\omega'(x)\ge \omega(S)/|S|$ implies that $\log\omega'(S)-\log\omega'(x)\le \log\omega(S)-\log(\omega(S)/|S|) + \log 3 = \log|S| + \log 3$.
\end{proof}

\subsection{Shannon and Shannon--Fano--Elias Codes}

A \defin{code} for a set $S$ is an injective function $\lambda:S\to\{0,1\}^*$ that maps each element in $S$ to a distinct binary string.  A code is \defin{prefix-free} if $\lambda(x)$ is not a prefix of $\lambda(y)$ for all distinct $x,y\in S$.

A rooted tree $T$ is a \defin{binary tree} if each node of $T$ has at most two children among which at most one is the \defin{left child} and at most one is the \defin{right child}.
A binary tree $T$ is \defin{full} if for every node $x$ of $T$,
either $x$ has two children or $x$ has no children.
A full binary tree $T$ is \defin{complete} if all the leaves of $T$ have the same depth in $T$.
Let $T$ be a binary tree. There is a natural ordering, called \defin{left-to-right ordering}, of the leaves of $T$, such that for every node $x$, if $x$ has two children then all the leaf descendants of the left child of $x$ precede all the leaf descendants of the right child of $x$.

The following result is a variant of Alphabetic Binary Trees. Similar results (with the constant $3$ replaced by $2$) are well-known and there are several proofs \cite{knuth:art_iii,mehlhorn:nearly,gilbert.moore:variable_length}. We provide a proof of the weaker statement here because it illustrates a simple but powerful technique---\defin{simulating weight by multiplicity}---that we will use again later.

\begin{lem}\label{shannon_fano_elias_tree}
  Let $S$ be a linearly ordered set and let $\omega:S\to\N^+$. Then there exists a binary tree $T$ such that $S$ is the set of leaves of $T$,
  the linear order on $S$ agrees with the left-to-right-order of the leaves in $T$, and for each $x\in S$
  \[
    \depth_{T}(x)\le \log\omega(S) - \log\omega(x) + 3 \enspace.
  \]
\end{lem}

\begin{proof}
  Let $h:=\lceil\log\omega(S)\rceil$.
  Let $Q$ be the complete binary tree of height $h$, so with $2^h$ leaves.
  For each $x\in S$, define the $\omega(x)$-element set
  \[  S_{\omega}(x):=\set{(x,1),\,(x,2),\,\ldots\,,\,(x,\omega(x))}\]
  and let $S_\omega:=\bigcup_{x\in S}S_{\omega}(x)$.  Then $|S_\omega|=\omega(S)$.  Treat $S_{\omega}$ as a linearly ordered set by using lexicographic order and rename the leftmost $\omega(S)$ leaves of $Q$, in order, with elements of $S_{\omega}$.

  Now for each $x\in S$, we are going to fix a node $r(x)$ in $Q$ such that
  (1) all the leaves of $Q_{r(x)}$ lie in $S_{\omega}(x)$; and
  (2) $Q_{r(x)}$ contains at least $\omega(x)/4$ leaves.

  Let $x\in S$.
  If $\omega(x)=1$ then define $r(x):=(x,1)$, which is a leaf of $Q$.
  Clearly, (1) and (2) hold.
  Now suppose that $\omega(x)\geq2$.
  Consider the lowest common $Q$-ancestor $y$ of $S_{\omega}(x)$.
  Thus, $y$ is an internal node of $Q$ and
  both children of $y$ contain in their subtree a leaf in $S_{\omega}(x)$.
  Since $S_{\omega}(x)$ is a set of consecutive leaves in $Q$,
  all leaves of the left child $y'$ of $y$ in $Q$ that are in $S_{\omega}(x)$ are rightmost in $Q_{y'}$ and
  all leaves of the right child $y''$ of $y$ in $Q$ that are in $S_{\omega}(x)$ are leftmost in $Q_{y''}$.
  Let $z$ be the child of $y$ in $Q$ such that $Q_z$ has at least $\omega(x)/2$ leaves in $S_{\omega}(x)$ and let $S_z$ be the set of all such leaves.
  Thus, $|S_z|\geq \omega(x)/2$.
  Since the argument is completely symmetric in these cases, assume that $z$ is the left child of $y$.
  Consider a path in $Q_z$ starting from $z$ and always taking an edge to the right child of the current node until we finish at a leaf.
  Let $v$ be the first vertex of that path such that $Q_v$ has all leaves in $S_z$.
  If $v=z$ then we set $r(x):=v$. Clearly (1) and (2) hold.
  Suppose that $v\neq z$.
  Let $u$ be the parent of $v$.
  Since the set of leaves of $Q_u$ form an interval of rightmost leaves of $Q_z$ and since $Q_u$ contains a leaf not in $S_z$, we conclude that $Q_u$ contains all nodes in $S_z$ as leaves.
  Therefore, $Q_v$ contains at least $|S_z|/2$ vertices of $S_z$ as leaves.
  We set $r(x):=v$ and again (1) and (2) hold.

  Note that the set $\set{r(x) \mid x\in S}$ is a $Q$-antichain.
  Indeed, if $x,y\in S$, $x\neq y$ and $r(x)$ is a $Q$-ancestor of  $r(y)$, then the leaves of $Q_{r(y)}$ are contained in the leaves of $Q_{r(x)}$ which by (1) implies that $Q_{r(x)}$ contains all nodes in $S_{\omega}(y)$ as leaves.
  However, again by (1) all the leaves of $Q_{r(x)}$ lie in $S_{\omega}(x)$ which is disjoint from $S_{\omega}(y)$, a contradiction.

  Let $x\in S$. Since $Q_{r(x)}$ contains at least $\omega(x)/4$ leaves, $\height(T_{r(x)}) \geq \log(\omega(x)/4) = \log\omega(x) - 2$. Finally,
  \begin{align*}
    \depth_Q(r(x))
     & \le \height(Q)-\height(Q_{r(x)})                    \\
     & \le \lceil\log\omega(S)\rceil - (\log\omega(x) - 2) \\
     & \le \log\omega(S) - \log\omega(x) +3 \enspace.
  \end{align*}

  Let $T:=\bigcup_{x\in S} P_Q(r(x))$.
  Since $\set{r(x) \mid x\in S}$ is a $Q$-antichain,
  we conclude that $\set{r(x)\mid x\in S}$ is the set of all leaves of $T$.
  Identifying each leaf $r_x$ of $T$ with $x$ for each $x\in S$ gives the desired tree.
\end{proof}

In coding theory, \cref{shannon_fano_elias_tree} is the basis of the Shannon-Fano-Elias coding scheme, as we now explain.  By assigning each edge $e$ of a binary tree $T$ a $0$ or $1$ depending on whether $e$ joins a parent to its left or right child, we obtain an encoding of each root-to-leaf path in $T$ as a binary string. Doing this with the tree from \cref{shannon_fano_elias_tree}, we obtain the following.

\begin{cor}\label{shannon_fano_elias_code}
  Let $S$ be a linearly ordered set and let $\omega:S\to\N^+$.  Then there exists a code $\rho:S\to \{0,1\}^\star$ for $S$ such that:
  \begin{enumerate}[nosep,nolistsep,label=(\alph*)]
    \item $|\rho(x)|\le \log\omega(S) - \log\omega(x)+3$ for each $x\in S$, and
    \item $\rho(x)$ is lexicographically less than $\rho(y)$ for each $x,y\in S$ with $x< y$ in $S$.
  \end{enumerate}
\end{cor}

\subsection{On Multipart Labels}
\label{multipart}

In many cases, a label (a bitstring) will have a variable number of parts (also bitstrings), of varying lengths.  We write this as $s:=\langle s_1,\ldots,s_p\rangle$, where each of $s_1,\ldots,s_p$ is a (possibly empty) bitstring.  The concatenation of $s_1,\ldots,s_p$, has length
\[
  \textstyle |s|=|s_1,\ldots,s_p|=\sum_{i=1}^p|s_p| \enspace .
\]
This ignores the important issue that, in order for a decoder to extract the individual parts $s_1,\ldots,s_p$ it must know the number of parts, $p$, and (at least $p-1$ of) the lengths $|s_1|,\ldots,|s_p|$.  We handle this by encoding the value of $p$ and the lengths $|s_1|,\ldots,|s_{p}|$ and prepending these to $s$. This is most easily done using Elias' code $\gamma:\N^+\to\{0,1\}^*$ for positive integers, in which each $x\in\N^+$ is encoded by a self-delimiting binary string $\gamma(x)$ with $|\gamma(x)|\le 2\lfloor \log x  \rfloor  + 1$ \cite{elias:universal}.\footnote{For each $x\in\N^+$, $\gamma(x)$ consists of $\lfloor\log x\rfloor$ 0 bits, followed by the binary representation of $x$ (which begins with a $1$ bit and has length $\lfloor\log x\rfloor+1$).} In our application $s_i$ may have zero length and $p$ may even be equal to zero, so that the codeword for each $x\in\N$ is $\gamma(x+1)$ and has length $2\floor{\log (x+1)}+1$.
In this way, we obtain a self-delimiting encoding of $s_1,\ldots,s_p$ as
\[  \mathdefin{\langle s_1,\ldots,s_p\rangle}:=\gamma(p+1),\gamma(|s_1|+1),\ldots,\gamma(|s_p|+1),s_1,\ldots,s_p \enspace ,
\]
which has length
\begin{equation}
  |\langle s_1,\ldots,s_p\rangle| = \textstyle\sum_{i=1}^p|s_i|+\Oh(\lgg p) + \Oh(\sum_{i=1}^p\lgg|s_i|) \enspace , \label{encoding_cost}
\end{equation}
where $\mathdefin{\lgg x} := \log\max\{2,x\}$ (just so that $\lgg(x)$ is defined and $\lgg(x)\ge 1$ for each $x\in\N$).    The bits of $\langle s_1,\ldots,s_p\rangle$ that come from $s_1,\ldots,s_p$ are called the \defin{payload} bits and the remaining bits are called the \defin{bookkeeping} bits of $\langle s_1,\ldots,s_p\rangle$.
In our setting it will always be the case that $p\in o(\log n)$ and we will use \cref{obs:nice} to ensure that $|s_i|\in \Oh(\log n)$ for each $i\in[p]$.
Thus, the number of bookkeeping bits in $\langle s_1,\ldots,s_p\rangle$  is $\Oh(p\log\log n)$.

In some cases, we include one or more non-negative integers as part of a multipart label.  When we include $d\in\N$ as part of a label this way, we are actually including the $\ceil{\log(d+1)}$-bit binary representation \defin{$\bin(d)$} of $d$.  For the sake of readability, we write this as $\langle s_1,d,s_2,\ldots\rangle$ rather than $\langle s_1,\bin(d),s_2,\ldots\rangle$. \label{bin_d}

\subsection{Separators with Low Pathwidth and Low Path-Adhesion-Width}\label{lowpw}

Recall that for $b\in\R^+$, a rooted tree $T$ is $b$-skinny if $|L_i(T)|\le b$ for each $i\in\N^+$, and a rooted tree-decomposition $(T,(B_x\mid x\in V(T)))$ of a graph $G$ is $b$-skinny if $T$ is $b$-skinny.  The following lemma shows how to construct a $(1/b)$-separator using the bags in a $b$-skinny subtree of a tree-decomposition.

\begin{lem}\label{low_pathwidth_separator}
  Let $b>1$, let $T$ be a rooted tree, let $\omega:V(T)\to\R$, and let
  \[
    X:=\set{x\in V(T)\mid\omega(T_x)>\omega(T)/b} \enspace.
  \]
  Then:
  \begin{enumerate}[nosep,nolistsep,label=(\alph*)]
    \item\label{connected} $T[X]$ is connected and contains the root of $T$;
    \item\label{skinny} $|X\cap L_i(T)|< b$, for each $i\in\N^+$; and
    \item\label{light} $\omega(C)\le \omega(T)/b$ for each component $C$ of $T-X$.
  \end{enumerate}
\end{lem}

\begin{proof}
  Let $r$ be the root of $T$.  Since $b>1$, we have that $\omega(T_r)=\omega(T)>\omega(T)/b$. %
  For $x\in X\setminus\{r\}$, let $y$ be the $T$-parent of $x$, then
  $\omega(T_y)\ge \omega(y)+\omega(T_x)>\omega(T_x)>\omega(T)/b$.
  Therefore $r\in X$ and, for each $x\in X\setminus\{r\}$, $X$ contains the $T$-parent of $x$. Thus $T[X]$ is connected.
  This proves~\ref{connected}.

  For each $i\in\N^+$, the subtrees in $\set{T_x \mid x\in X\cap L_i(T)}$ are pairwise vertex-disjoint, so
  \[
    \omega(T) \ge \sum_{x\in X\cap L_i(T)} \omega(T_x)> \sum_{x\in X\cap L_i(T)}\omega(T)/b = |X\cap L_i(T)|\cdot\omega(T)/b.
  \]
  Rewriting this inequality gives $|X\cap L_i(T)|<b$ for each $i\in\N^+$.
  This proves~\ref{skinny}.

  By~\ref{connected}, each component $C$ of $T-X$ is a subtree of $T$ that is rooted at some node $y$ whose $T$-parent is in $X$.  Since  $y\not\in X$, $\omega(C)=\omega(T_y)\le \omega(T)/b$.
  This proves~\ref{light}.
\end{proof}

\begin{cor}\label{b_skinny_decomp}
  For every $b>1$, every $n\in\N^+$, every $n$-vertex graph $G$ and every tree-decomposition $\mathcal{T}:=(T,(B_x\mid x\in V(T)))$ of $G$, there exists a tree-decomposition $\mathcal{Q}:=(Q,(D_x\mid x\in V(Q)))$ of $G$ such that
  \begin{enumerate}[nosep,nolistsep,label=\rm(\alph*)]
    \item\label{prop_adhesion}
          each adhesion of $\mathcal{Q}$ is an adhesion of $\mathcal{T}$;
    \item\label{prop_skinny_torso} for each $q\in V(Q)$, $\torso{G}{\mathcal{Q},D_q}$ has a $b$-skinny tree-decomposition $(T(q), (B_y\mid y\in V(T(q))))$ where $T(q)$ is a subtree of $T$; and
    \item\label{prop_geometric} $\height(Q)\le \log_b n$.
  \end{enumerate}
\end{cor}

\begin{proof}
  Fix $b>1$.    Let $G$ be an $n$-vertex graph, let $\mathcal{T}:=(T,(B_x\mid x\in V(T)))$ be a rooted tree-decomposition of $G$, and let $r$ be the root of $T$.  We prove a slightly stronger statement that involves a set $A\subseteq B_r$.  The proof is by induction on $|V(G-A)|$.   To make the induction work we strengthen \ref{prop_geometric} to:
  \begin{enumerate}[nosep,nolistsep,label=\rm(\alph*')]
    \setcounter{enumi}{2}
    \item\label{prop_geometric_ii} $\height(Q) \le \max\{0, \log_b(|V(G-A)|)\}$
  \end{enumerate}

  Define the function
  \[
    \omega(y) := \begin{cases}
      |B_r\setminus A|   & \text{if $y=r$}                                         \\
      |B_y\setminus B_x| & \text{if $y\in V(T)\setminus\{r\}$ has $T$-parent $x$.}
    \end{cases}
  \]
  Observe that $\sum_{y\in V(T)}\omega(y)=|V(G-A)|$. Let $N:=|V(G-A)|$.

  First we consider the base case. If $N\le b$, let $Q$ be a tree that has only one node, $s$. Let $X=\{r\}$ if $N=0$, and $X=\{x\in V(T)\mid \omega(T_x)>0\}$ otherwise. Then $T[X]$ is a subtree of $T$ containing $r$. For any $i\in\N^+$, the subtrees $T_x$ for $x\in X\cap L_i(T)$ are disjoint, so $|X\cap L_i(T)| \le \sum_{x\in X\cap L_i(T)}\omega(T_x) \le \omega(T) = N \le b$. Thus, $T[X]$ is $b$-skinny. Let $D_s:=\bigcup_{x\in X} B_x$ and let $\mathcal{Q}:=(Q,(D_q\mid q\in V(Q)))$. For each $y\in N_T(X)$, $\omega(T_y)=0$, so $V(G_y-A_y)=\emptyset$, which implies $B_z\subseteq A_y\subseteq B_x$ for all $z\in V(T_y)$, where $x\in X$ is the $T$-parent of $y$. Hence $D_s=\bigcup_{y\in V(T)} B_y$, meaning $\mathcal{Q}$ is a tree-decomposition of $G$ with no adhesions, so~\ref{prop_adhesion} holds vacuously. Since $\torso{G}{\mathcal{Q},D_s}=G$, we take $T(s)=T[X]$ and then $(T(s),(B_y\mid y\in V(T(s))))$ is a $b$-skinny tree-decomposition of $\torso{G}{\mathcal{Q},D_s}$, so~\ref{prop_skinny_torso} holds. Since $V(Q)=\set{s}$ and $\depth_Q(s)=0$, condition \ref{prop_geometric_ii} holds because $\height(Q)=0 \le \max\{0, \log_b(N)\}$.

  Now assume $N> b$. Let $X:=\set{x\in V(T)\mid  \omega(T_x)>N/b}$. By \cref{low_pathwidth_separator}, $T[X]$ is a $b$-skinny subtree of $T$ that includes $r$.  The root of the tree $Q$ will be a node $s$ and $D_s:=\bigcup_{x\in X} B_x$.

  For each $y\in N_T(X)$, let $G_y:=G[\bigcup_{z\in V(T_y)} B_z]$ and let $A_y := B_y \cap B_x$ be the parent adhesion of $y$ in $\mathcal{T}$, where $x$ is the $T$-parent of $y$.  Then $\mathcal{T}_y:=(T_y, (B_z\mid z\in V(T_y)))$ is a tree-decomposition of $G_y$ and every adhesion of $\mathcal{T}_y$ is an adhesion of $\mathcal{T}$. Observe that $|V(G_y-A_y)| = \omega(T_y) \le N/b$. Apply the inductive hypothesis to $G_y$ with tree-decomposition $\mathcal{T}_y$ and special set $A_y$ to obtain a tree-decomposition $\mathcal{Q}_y=(Q_y,(D_q \mid q\in V(Q_y)))$ of $G_y$.  By induction, each bag of $\mathcal{Q}_y$ has a tree-decomposition formed by a (skinny) subtree of $T$.  Therefore, some bag $D_{r_y}$ of $\mathcal{Q}_y$ contains all the vertices of $B_y$.  Make $r_y$ a child of $s$ in the tree $Q$.

  We now show that $\mathcal{Q}$ is a tree-decomposition of $G$ that satisfies the requirements of the lemma.

  \begin{itemize}

    \item[\ref{prop_adhesion}:] For each $y\in N_T(X)$, the adhesion between $s$ and $r_y$ is $A_y$, which is an adhesion of $\mathcal{T}$. Furthermore, each adhesion of $\mathcal{T}_y$ is an adhesion of $\mathcal{T}$, so each adhesion of $\mathcal{Q}_y$ is an adhesion of $\mathcal{T}$.  Thus every adhesion of $\mathcal{Q}$ is an adhesion of $\mathcal{T}$.

    \item[\ref{prop_skinny_torso}:] $\mathcal{T}_X:=(T[X],(B_y\mid y\in V(T[X])))$ is a $b$-skinny tree-decomposition of $\torso{G}{\mathcal{Q},D_s}$.  For each $y\in N_T(X)$, $\torso{G}{\mathcal{Q}_y,D_{r_y}}$ has a tree-decomposition whose tree is a $b$-skinny subtree of $T_y$. The tree-decomposition of $\torso{G_y}{\mathcal{Q}_y,D_{r_y}}$ contains $B_y$, therefore this tree-decomposition is also a tree-decomposition of $\torso{G}{\mathcal{Q},D_{r_y}}$.  By induction, every other bag of $\mathcal{Q}_y$ has a tree-decomposition whose tree is a $b$-skinny subtree of $T$. Thus, each torso of $\mathcal{Q}$ has a $b$-skinny tree-decomposition whose tree is a subtree of $T$.

    \item[\ref{prop_geometric_ii}:] For each $y\in L_2(Q)$, the inductive hypothesis implies that $\height(Q_y)\le\max\{0,\linebreak \log_b(|V(G_y-A_y)|)\}$. Since $|V(G_y-A_y)| = \omega(T_y) \le N/b = |V(G-A)|/b$,  \[\height(Q_y)\le\max\{0, \log_b(|V(G-A)|/b)\} = \max\{0, \log_b|V(G-A)| - 1\}.\] 
    Thus, $\height(Q)=1+\max\{\height(Q_y):y\in L_2(Q)\}\le \max\{1, \log_b|V(G-A)|\} =\linebreak \log_b|V(G-A)|$ (since $|V(G-A)| > b \ge 1$). \qedhere
  \end{itemize}
\end{proof}

\section{Weighted Mixed Labelling Schemes -- Proof of \texorpdfstring{\cref{PlainRealMainTechnical}}{Theorem \ref*{PlainRealMainTechnical}}}
\label{weighted_mixed_section}
 
This section defines weighted mixed labelling schemes and proves our main technical result, namely \cref{PlainRealMainTechnical}. Let $\mathcal{G}$ be a class of graphs.
Let $g_i:\N\to\R^+$ for each $i\in[3]$.
The class  $\mathcal{G}$ admits a \defin{weighted $(g_1,g_2,g_3)$ mixed labelling scheme} of $\mathcal{G}$ if there exists a pair $(A,I)$ with
\begin{itemize}[nosep,nolistsep]
  \item $A:(\{0,1\}^*)^2\to\{0,1\}$ (the \defin{adjacency tester}); and
  \item $I:(\{0,1\}^*)^3\to\{0,1\}$ (the \defin{identity tester})
\end{itemize}
such that for every $n\in\N^+$, every $n$-vertex graph $G^+$ in $\mathcal{G}$, every spanning subgraph $G$ of $G^+$, and every weight function $\omega:V(G)\to \R^+$, there exists a function $\mu: V(G)\cup \{K\mid \text{$K$ is a clique in $G^+$}\} \to \set{0,1}^*$ and a function $\kappa:\{(K, u)\mid \text{$K$ is a clique in $G^+$}, u\in K\}\to \set{0,1}^*$ such that:
\begin{enumerate}[nosep,nolistsep,label=(\roman*)]
  \item For all pairs of distinct cliques $K$ and $L$ in $G^+$, $\mu(K)\neq\mu(L)$;
  \item For all $u,v\in V(G)$,
        \[
          \textrm{$A(\mu(u),\mu(v))=1$ if and only if $uv\in E(G)$;}
        \]
  \item For every clique $K$ in $G^+$, $u\in K$, $v\in V(G)$,
        \[
          \textrm{$I(\mu(K),\kappa(K, u),\mu(v))=1$ if and only if $u=v$;}
        \]
  \item for every $v\in V(G)$,
        \[
          |\mu(v)| \leq \log \omega(G) - \log\omega(v) + g_1(n);
        \]
  \item for every clique $K$ in $G^+$,
        \[
          |\mu(K)|\le \log \omega(G) - \log\min_{v\in K}(\omega(v)) + g_3(n);
        \]
  \item  for every clique $K$ in $G^+$ and $u\in K$,
        \[
          |\kappa(K, u)| \leq g_2(n).
        \]
\end{enumerate}
A pair $(\mu,\kappa)$ that satisfies the first two conditions is a \defin{mixed labelling} of $(G^+,G)$ for $(A,I)$.  A pair $(\mu,\kappa)$ that satisfies all these conditions is a \defin{weighted $(g_1,g_2,g_3)$ mixed labelling} of $(G^+,G,\omega)$ for $(A,I)$.

Here is a little intuition about these definitions. Function $g_1(n)$ measures how much the length of a vertex label exceeds the ideal value $\log \omega(G) - \log\omega(v)$ that appears in Kraft's Inequality \cite{Kraft1949}. Function $g_2(n)$ is the length of the longest local identifier. Function $g_3(n)$ measures how much the length of a clique label exceeds the Kraft-like quantity $\log\omega(G)-\log\min_{v\in K}(\omega(v))$.

The existence of a weighted $(g_1(n),\cdot,\cdot)$ mixed labelling scheme for a graph class $\mathcal{G}$ immediately implies the existence of a $(\log n+g_1(n))$-bit adjacency labelling scheme for $\mathcal{G}$. Thus, we are interested in the case where $g_1(n)\in o(\log n)$. We say that a graph class $\mathcal{G}$ admits an \defin{efficient} weighted mixed labelling scheme if $\mathcal{G}$ admits a weighted $(g_1,g_2,g_3)$ mixed labelling scheme with $g_i(n)\in o(\log n)$ for each $i\in [3]$. So any graph class that admits an efficient weighted mixed labelling scheme admits a $(1+o(1))\log n$-bit adjacency labelling scheme. The following lemma shows that any adjacency labelling scheme obtained this way is injective, which justifies \cref{main_result_induced_universal}.

\begin{lem}
  Let $(\mu,\kappa)$ be a mixed labelling of some $(G^+,G,\omega)$ for some $(A,I)$. Then the restriction $\mu_{|V(G)}$ of $\mu$ to $V(G)$ is injective.
\end{lem}

\begin{proof}
  Suppose, for the sake of contradiction, that there exists distinct $v,w\in V(G)$ such that $\mu(v)=\mu(w)$. Since  $K:=\{v\}$ is a clique in $G^+$, this immediately leads to the contradiction
  $1=I(\mu(K),\kappa(K,v),\mu(v))=I(K,\kappa(K,v),\mu(w))=0$.
\end{proof}

It will actually be convenient to consider a slightly weaker notion of weighted $(g_1,g_2,g_3)$ mixed labelling schemes.

A weight function $\omega:S\to \N^+$ is \defin{nice} if
for every $x\in S$,
\[
  \log \omega(S) - \log \omega(x) \leq \log |S| + 2\enspace.
\]
(We emphasize that weights are required to be positive integers here, which will be important later on.)
We then say that a class of graphs $\mathcal{G}$ admits a \defin{weak weighted $(g_1,g_2,g_3)$ mixed labelling scheme} if $\mathcal{G}$ satisfies the definition given above for weighted $(g_1,g_2,g_3)$ mixed labelling schemes, with the exception that only nice weight functions $\omega$ are considered.
The two notions are essentially equivalent, as shown by the following lemma. This will allow us to only consider nice weight functions in the proofs.

\begin{lem}
  \label{lem:nice_weights}
  Let $\mathcal{G}$ be a class of graphs.
  Let $g_i:\N\to\R$ be functions for each $i\in[3]$.
  If $\mathcal{G}$ admits a weak weighted $(g_1,g_2,g_3)$ mixed labelling scheme, then $\mathcal{G}$ admits a weighted $(g_1+2,g_2,g_3+2)$ mixed labelling scheme.
\end{lem}
\begin{proof}
  Suppose that $\mathcal{G}$ admits a weak weighted $(g_1,g_2,g_3)$ mixed labelling scheme.
  Let $A$ and $I$ be the corresponding adjacency and identity tester, respectively.
  We will use the same pair $(A, I)$ for our weighted $(g_1+2,g_2,g_3+2)$ mixed labelling scheme.
  Let $n\in \N^+$, let $G^+$ be an $n$-vertex graph in $\mathcal{G}$, let $G$ be a spanning subgraph of $G^+$, and let $\omega:V(G)\to \R^+$ be an arbitrary weight function.
  Let $\omega'$ be the function obtained by applying \cref{obs:nice} to $\omega$.
  Thus, for every $v\in V(G)$,
  \[
    \log\omega'(V(G))-\log\omega'(v)\le\min\{\log n,\;\log\omega(G)-\log\omega(v)\}+2 \enspace.
  \]
  In particular, $\omega'$ is nice.
  Hence, by our assumption, there is a mixed labelling $\mu$ of $(G^+,G,\omega')$ for $(A,I)$ and a local identifier $\kappa$ of $(G^+,G,\omega')$ for $(A,I)$.
  Hence, for every $v\in V(G)$,
  \[
    |\mu(v)| \leq \log \omega'(G) - \log\omega'(v) + g_1(n) \leq \log \omega(G) - \log\omega(v) + g_1(n) + 2\enspace,
  \]
  and  for every clique $K$ in $G^+$,
  \[
    |\mu(K)|\le \log \omega'(G) - \log\min_{v\in K}(\omega'(v)) + g_3(n)
    \leq \log \omega(G) - \log\min_{v\in K}(\omega(v)) + g_3(n) + 2\enspace.
  \]
  It follows that $\mu$ and  $\kappa$ are the desired  mixed labelling and local identifier, respectively,  of $(G^+,G,\omega)$ for $(A,I)$.
\end{proof}

\subsection{Adding Apex Vertices}
\label{sec:adding_apexes}

The next lemma shows how a mixed labelling scheme can accommodate the addition of apex vertices. Recall that $\mathcal{G}^{+a}$ is the class of graphs $G$ such that $G-X\in\mathcal{G}$ for some $X\subseteq V(G)$ with $|X|\leq a$.

\begin{lem}\label{add_apexes}
  Let $\mathcal{G}$ be a class of graphs that admits a weighted $(g_1,g_2,g_3)$ mixed labelling scheme, where $g_i:\N\to\R$ is a non-decreasing function  with $g_i(n)\in \Oh(\log n)$ for each $i\in[3]$.
  Let $a\in\N^+$. Then $\mathcal{G}^{+a}$ admits a weighted $(g'_1,g'_2,g'_3)$ mixed labelling scheme, where
  \begin{align*}
    g'_1(n)
     & = g_1(n)+\Oh(a)+\Oh(\log\log n) \enspace,  \\
    g'_2(n)
     & = g_2(n)+\Oh(\log a) \enspace,             \\
    g'_3(n)
     & = g_3(n)+\Oh(a)+\Oh(\log\log n) \enspace .
  \end{align*}
\end{lem}

\begin{proof}
  By~\cref{lem:nice_weights} it is enough to show that $\mathcal{G}^{+a}$ admits a weak weighted $(g_1',g_2',g_3')$ mixed labelling scheme. We will describe an adjacency tester $A'$ and an identity tester $I'$ for which we can construct a mixed labelling $(\mu,\kappa)$ for any $(G^+,G,\omega)$ where $G^+\in \mathcal{G}^{+a}$, $G$ is a spanning subgraph of $G^+$, and $\omega:V(G)\to\N^+$ is a nice weight function.
  Since it is not possible to describe $A'$ and $I'$ without knowing the contents of the labels, we first describe how to compute $\mu$ and $\kappa$ for a particular $(G^+,G,\omega)$.

  Let $G^+$ be an $n$-vertex graph in $\mathcal{G}^{+a}$.  Thus, there is a subset $X\subseteq V(G^+)$ of at most $a$ vertices such that $G^+-X\in\mathcal{G}$.  Let $G$ be a spanning subgraph of $G^+$, and let $\omega:V(G)\to\N^+$ be a nice weight function.
  Thus, for each $v\in V(G)$,
  \begin{equation}
    \log\omega(G)-\log\omega(v)\le\log n + 2\enspace.
  \end{equation}

  Let $n\in \N^+$,
  let $G^+$ be an $n$-vertex graph in $\mathcal{G}^{+a}$,    let $G$ be a spanning subgraph of $G^+$,    and
  let $\omega:V(G)\to\N^+$ be a nice weight function.

  \paragraph{\boldmath The subgraph label $\mu^-$:}
  Let $(A,I)$ be the pair of adjacency and identity testers for a weighted $(g_1,g_2,g_3)$ mixed labelling scheme of $\mathcal{G}$.
  Let $(\mu^-,\kappa^-)$ be a weighted $(g_1,g_2,g_3)$ mixed labelling of $(G^+-X, G-X,\omega_{|V(G)})$ for $(A,I)$.
  Then for each $v\in V(G-X)$,
  \begin{align}
    \begin{split}
      |\mu^-(v)|
       & \le \log\omega(G-X) -\log\omega(v)+g_1(|V(G-X)|) \\
       & \le \log\omega(G)-\log\omega(v)+g_1(n) \enspace,
    \end{split}
    \label{eq:mu-}
  \end{align}
  and for each clique $K$ in $G-X$,
  \begin{align}
    \label{eq:split}
    \begin{split}
      |\mu^-(K)|
       & \le \log\omega(G-X)-\log\min_{v\in K}(\omega(v))+g_3(|V(G-X)|)  \\
       & \le \log\omega(G)-\log\min_{v\in K}(\omega(v))+g_3(n) \enspace.
    \end{split}
  \end{align}
  Extend the domain of $\mu^-$ to include the vertices in $X$ by defining $\mu^-(v):=\varepsilon$ (the empty string), for each $v\in X$.

  \paragraph{\boldmath The apex identifier $\sigma(u)$:}
  Fix a linear ordering $u_1,\ldots,u_{b}$ of the vertices in $X$.  By definition, $b\le a$.
  For each $u\in V(G)$ define
  \[
    \sigma(u) :=
    \begin{cases}
      \textrm{$\lceil\log(b+1)\rceil$-bit binary representation of $i$} & \textrm{if $u\in X$ and $u=u_i$,} \\
      \textrm{$\lceil\log(b+1)\rceil$ zero bits}                        & \textrm{if $u\in V(G)-X$.}
    \end{cases}
  \]

  \paragraph{\boldmath The vertex label $\mu(v)$:} For each $v\in V(G)$, let $c(v)$ be a $b$-bit string whose $i$-th bit indicates if $vu_i\in E(G)$.
  For each $v\in V(G)$, define
  \[
    \mu(v):=\langle \sigma(v),c(v),\mu^-(v) \rangle \enspace .
  \]
  For each $v\in V(G)$,
  \begin{align*}
    |\sigma(v)|+b+|\mu^-(v)|
     & \le \lceil\log(b+1)\rceil + b +\log\omega(G-X) - \log\omega(v)+g_1(n-|X|)           & \textrm{by~\eqref{eq:mu-}}      \\
     & \le\log\omega(G) - \log\omega(v)+g_1(n)+2a                                                                            \\
     & \le \log n + g_1(n)+2a+2
     & \textrm{since $\omega$ is nice}                                                                                       \\
     & \in\Oh(a+\log n)\enspace,
    \intertext{and therefore}
    |\mu(v)|
     & \leq |\sigma(v)|+a+|\mu^-(v)| +\mathcal{O}(\lgg |\sigma(v)|+\lgg a+\lgg |\mu^-(v)|) & \text{by \eqref{encoding_cost}} \\
     & \leq \log\omega(G) - \log\omega(v)+g_1(n) + \mathcal{O}(a+\log\log n) \enspace .
  \end{align*}

  \paragraph{Adjacency testing:}
  We now describe the adjacency tester $A'$.
  Given the labels $\mu(v)=\langle \sigma(v),c(v),\mu^-(v)\rangle$ and $\mu(w)=\langle \sigma(w),c(w),\mu^-(w)\rangle$ for two vertices $v$ and $w$ of $G$,
  $A'$ works as follows.
  First by checking the values of $\sigma(v)$ and $\sigma(w)$ we check if $v$ and/or $w$ lie in $X$.
  If at least one of them, say $v$, is in $X$, then $v=u_i$ for some $i\in[b]$ and $\sigma(v)$ encodes the value of $i$.  In this case,
  we set $A'(\mu(v),\mu(w))$ to be the $i$-th bit of $c(w)$.
  If neither $v$ nor $w$ is in $X$, then $A'(\mu(v),\mu(w))=A(\mu^-(v),\mu^-(w))$.

  \paragraph{\boldmath The clique label $\mu(K)$:}
  For each clique $K$ in $G^+$, let $c(K)$ be a $b$-bit string whose $i$-th bit indicates if $u_i\in K$.
  For each clique $K$ in $G^+$, define
  \[
    \mu(K):=
    \begin{cases}
      \langle c(K) \rangle           & \text{if $K\subseteq X$,} \\
      \langle c(K),\mu^-(K-X)\rangle & \text{otherwise.}
    \end{cases}
  \]
  To see that $\mu$ is injective when restricted to cliques, consider two distinct cliques $K$ and $L$ in $G^+$. If $K\cap X \neq L\cap X$, then $c(K)\neq c(L)$, so $\mu(K)\neq \mu(L)$ because their first parts differ. Otherwise, $K\setminus X \neq L\setminus X$. If one is empty, say $K\subseteq X$, then $\mu(K)$ has one part while $\mu(L)$ has two parts, so $\mu(K)\neq\mu(L)$. If both are non-empty, then $\mu^-(K-X)\neq \mu^-(L-X)$ since $\mu^-$ is injective on cliques of $G^+-X$, which implies $\mu(K)\neq \mu(L)$.

  If $K\not\subseteq X$, then
  \begin{align*}
    b+|\mu^-(K-X)|
     & \le \log\omega(G)-\log\min_{v\in K}(\omega(v)) + g_3(n-|X|)+a
     & \textrm{by~\eqref{eq:split}}                                                                  \\
     & \le \log n + 2 + g_3(n)+a
     & \textrm{since $\omega$ is nice}                                                               \\
     & \in\Oh(a) +\Oh(\log n) \enspace,
    \intertext{and therefore}
    |\mu(K)|
     & \le |\mu^-(K-X)|+b + \Oh(\lgg|\mu^-(K-X)|+\lgg b)
     & \text{by \eqref{encoding_cost}}                                                               \\
     & \leq \log\omega(G)-\log\min_{v\in K}(\omega(v)) + g_3(n) + \mathcal{O}(a+\log\log n)\enspace.
  \end{align*}
  If $K\subseteq X$, then $|\mu(K)| \le b + \Oh(\lgg b)$, which easily satisfies the same bound.

  \paragraph{\boldmath The local identifier $\kappa(K,u)$:}
  For each clique $K$ in $G^+$ and each $u\in K$, define
  \[
    \kappa(K,u) :=
    \begin{cases}
      0,\kappa^-(K-X,u) & \text{if $u\not\in X$,} \\
      1,\sigma(u)       & \text{if $u\in X$.}
    \end{cases}
  \]
  Note that $\kappa(K,u)$ is a concatenation, so it does not incur the encoding cost from \eqref{encoding_cost}. Thus,
  \[
    |\kappa(K,u)|\le 1+ \max\{g_2(n), \lceil\log(a+1)\rceil\} \le g_2(n)+\Oh(\log a) \enspace.
  \]

  \paragraph{Identity testing:}
  We now describe the identity tester $I'$.
  Given $\mu(K)$, $\kappa(K,u)$, and $\mu(v)$ for some clique $K$ in $G^+$, $u\in K$, and $v\in V(G)$,
  we compute $I'(\mu(K),\kappa(K,u),\mu(v))$ as follows.
  First by examining $\sigma(v)$ and the first bit of $\kappa(K,u)$, the tester determines whether $v\in X$ and whether $u\in X$.
  If $u,v\in X$, then $I'(\mu(K),\kappa(K,u),\mu(v))=1$ if $\kappa(K,u)=1,\sigma(v)$, and $I'(\mu(K),\kappa(K,u),\mu(v))=0$ otherwise.
  If $u,v\not\in X$, then since $u\in K$ we have that $K\not\subseteq X$, which means $\mu(K)$ has a second part $\mu^-(K-X)$. The tester extracts this second part and sets
  $I'(\mu(K),\kappa(K,u),\mu(v)) = I(\mu^-(K-X),\kappa^-(K-X,u),\mu^-(v))$.
  Otherwise, one of $u$ or $v$ is in $X$ and the other is not, so $u\neq v$ and we set $I'(\mu(K),\kappa(K,u),\mu(v)) =0$.
  \qedhere
\end{proof}

\subsection{Disjoint Unions of Graphs}
\label{sec:disjoint}

\begin{lem}\label{disjoint_union_ii}
  Let $\mathcal{G}$ be a class of graphs that admits a weighted $(g_1,g_2,g_3)$ mixed labelling scheme, where $g_i:\N\to\R^+$ is a non-decreasing function  with $g_i(n)\in \Oh(\log n)$ for each $i\in[3]$. Let $\mathcal{G}'$ be the class of graphs that can be formed by taking the union of a finite number of pairwise vertex-disjoint graphs in $\mathcal{G}$.
  Then $\mathcal{G}'$ admits a weighted $(g'_1,g_2,g_3')$ mixed labelling scheme, where
  \begin{align*}
    g'_1(n)
     & = g_1(n)+\Oh(\log\log n)\enspace,   \\
    g'_3(n)
     & = g_3(n)+\Oh(\log\log n) \enspace .
  \end{align*}
\end{lem}

\begin{proof}
  By~\cref{lem:nice_weights} it is enough to show that $\mathcal{G}'$ admits a weak weighted $(g_1',g_2,g_3')$ mixed labelling scheme. We will describe an adjacency tester $A'$ and an identity tester $I'$ for which we can construct a mixed labelling $(\mu,\kappa)$ for any $(G^+,G,\omega)$ where $G^+\in\mathcal{G}'$, $G$ is a spanning subgraph of $G^+$, and $\omega:V(G)\to\N^+$ is a nice weight function.
  Since it is not possible to describe $A'$ and $I'$ without knowing the contents of the labels, we first describe how to compute $\mu$ and $\kappa$ for a particular $(G^+,G,\omega)$.

  Let $G^+$ be an $n$-vertex graph in $\mathcal{G}'$, let $G$ be a spanning subgraph of $G^+$, and let $\omega:V(G)\to\N^+$ be a nice weight function.
  Thus, for each $v\in V(G)$,
  \begin{equation}
    \label{eq:nice}
    \log\omega(G)-\log\omega(v)\le\log n + 2\enspace.
  \end{equation}
  By definition, $G^+:=\bigcup_{i=1}^m G^+_i$, where $G^+_1,\ldots,G^+_m$ are pairwise vertex-disjoint members of $\mathcal{G}$.
  Therefore, $G:=\bigcup_{i=1}^m G_i$, where $G_i:=G[V(G^+_i)]$ is a spanning subgraph of $G^+_i$ for each $i\in[m]$.

  \paragraph{\boldmath The subgraph label $\rho(x)$:}
  Let $\psi:[m]\to\N^+$ be the weight function defined by $\psi(i):=\omega(G_i)$ and observe that $\psi([m])=\omega(G)$ since $G_1,\ldots,G_m$ are pairwise vertex-disjoint.
  Apply \cref{shannon_fano_elias_code} to the set $[m]$ with weight function $\psi$ to obtain a prefix-free code $\rho:[m]\to\{0,1\}^*$ such that for each $i\in[m]$,
  \begin{equation}
    \label{eq:pi}
    |\rho(i)|\le\log \omega(G)-\log \omega(G_i)+3\enspace.
  \end{equation}
  For each $i\in[m]$ and each vertex $v$ of $G_i$, define $\rho(v):=\rho(i)$. For each $i\in[m]$ and each clique $K$ of $G^+_i$, define $\rho(K):=\rho(i)$.

  \paragraph{\boldmath The sub-label $\mu_i(x)$:}
  Let $(A,I)$ be the pair of adjacency and identity testers for a weighted $(g_1,g_2,g_3)$ mixed labelling scheme of $\mathcal{G}$.
  For each $i\in[m]$,
  let $(\mu_i,\kappa_i)$ be a mixed labelling of $(G_i^+, G_i,\omega_{|V(G_i)})$ for $(A,I)$.
  Thus, for each $i\in[m]$ and each $v\in V(G_i)$,
  \begin{align}
    \label{eq:mu}
    \begin{split}
      |\mu_i(v)|
       & \leq \log\omega(G_i) - \log\omega(v) + g_1(|V(G_i)|) \\
       & \leq \log\omega(G_i) - \log\omega(v) + g_1(n)
    \end{split}
    \intertext{and for each clique $K$ in $G^+_i$,}
    \begin{split}
      |\mu_i(K)|
       &
      \leq \log\omega(G_i) - \log\min_{v\in K}(\omega(v)) + g_3(|V(G_i)|) \\
       & \leq \log\omega(G_i) - \log\min_{v\in K}(\omega(v)) + g_3(n).
    \end{split}
    \label{eq:mu-K}
  \end{align}
  For each $i\in[m]$, and each vertex $v$ of $G_i$, define $\mu'(v):=\mu_i(v)$.
  For each $i\in[m]$, and each clique $K$ of $G^+_i$, define $\mu'(K):=\mu_i(K)$.

  \paragraph{\boldmath The vertex label $\mu(v)$:}
  Let $A$ and $I$ denote the adjacency and identity testers, respectively, for the weighted $(g_1,g_2,g_3)$ mixed labelling scheme of $\mathcal{G}$.
  For each $i\in[m]$ and each $v\in V(G_i)$, define
  \[
    \mu(v):=\langle \rho(v), \mu'(v) \rangle \enspace;
  \]
  thus
  \begin{align*}
    |\rho(v)| + |\mu_i(v)|
     & \le \log\omega(G) -\log\omega(v)+3+g_1(n)                        &  & \textrm{by~\eqref{eq:pi} and~\eqref{eq:mu},} \\
     & \le \log n + 5 + g_1(n) = \mathcal{O}(\log n)                    &  & \textrm{by~\eqref{eq:nice},}
    \intertext{and therefore}
    |\mu(v)|
     & \le |\rho(v)| + |\mu'(v)| + \Oh(\lgg |\rho(v)| + \lgg|\mu'(v)|)
     & \text{by \eqref{encoding_cost}}                                                                                    \\
     & \le \log\omega(G) -\log\omega(v)+g_1(n)+\Oh(\log\log n)\enspace.
  \end{align*}

  \paragraph{Adjacency testing:} We now define an adjacency tester $A'$.
  Given the labels $\mu(v)=\linebreak \langle\rho(v),\mu'(v)\rangle$ and $\mu(w)=\langle\rho(w),\mu'(w)\rangle$ for two vertices $v,w\in V(G)$, we compute $A'(\mu(v),\mu(w))$ as follows.
  First $A'$ verifies if $\rho(v)=\rho(w)$.
  If $\rho(v)\neq\rho(w)$, then $v$ and $w$ lie in different components of $G$, so they cannot be adjacent in $G$, and
  we set $A'(\mu(v),\mu(w)):=0$.
  If $\rho(v)=\rho(w)$, then $v$ and $w$ are both contained in the graph $G_i$, for some $i\in[m]$.  In this case
  $\mu(v)=\langle\rho(i),\mu'(v)\rangle$, and $\mu(w)=\langle\rho(i),\mu'(w)\rangle$.
  In this case, we simply set $A'(\mu(v),\mu(w)) := A(\mu'(v), \mu'(w))$.

  \paragraph{\boldmath The clique label $\mu(K)$:}
  For each $i\in[m]$ and each clique $K$ in $G^+_i$, define
  \[
    \mu(K):=\langle\rho(K), \mu'(K)\rangle \enspace ;
  \]
  To see that $\mu$ is injective when restricted to cliques, consider two distinct cliques $K$ and $L$ in $G^+$. Since $G^+$ is the disjoint union of $G^+_1,\ldots,G^+_m$, each clique is contained in exactly one component. If $K$ and $L$ are in different components $G^+_i$ and $G^+_j$, then $\rho(K)=\rho(i)\neq\rho(j)=\rho(L)$ since $\rho$ is a prefix-free code, so $\mu(K)\neq\mu(L)$. If they are in the same component $G^+_i$, then $\mu'(K)=\mu_i(K)\neq\mu_i(L)=\mu'(L)$ since $\mu_i$ is injective on cliques of $G^+_i$, which implies $\mu(K)\neq\mu(L)$.

  Thus,
  \begin{align*}
    |\rho(K)| + |\mu'(K)|
     & \le \log\omega(G) -\log\min_{v\in K}(\omega(v))+3+g_3(n)                         \\
     & \le \log n + 5 + g_3(n) = \mathcal{O}(\log n)                                    \\
    \intertext{and therefore}
    |\mu(K)|
     & \le |\rho(K)| + |\mu'(K)| + \Oh(\lgg |\rho(K)| + \lgg|\mu'(K)|)
     & \text{by \eqref{encoding_cost}}                                                  \\
     & \le \log\omega(G) -\log\min_{v\in K}(\omega(v))+g_3(n)+\Oh(\log\log n) \enspace.
  \end{align*}

  \paragraph{\boldmath The local identifier $\kappa(K,u)$:}
  Define $\kappa(K,u) := \kappa_i(K,u)$ for each $i\in[m]$, each clique $K$ in $G^+_i$, and each $u\in K$. In particular,
  \[
    |\kappa(K,u)| \leq g_2(n).
  \]

  \paragraph{Identity testing:}
  Define an identity tester $I'$ as follows.
  Given $\mu(K)$, $\kappa(K,u)$, and $\mu(v)$ for some clique $K$ in $G^+$, some $u\in K$, and some $v\in V(G)$, we compute $I'(\mu(K),\kappa(K,u),\mu(v))$ as follows.
  If $\rho(K)\neq\rho(v)$, then $K$ and $v$ lie in different components of $G$, so $v\not\in K$, and we set $I'(\mu(K),\kappa(K,u),\mu(v)):=0$.  If $\rho(K)=\rho(v)$, then $K$ and $v$ are contained in the same graph $G_i$ for some $i\in[m]$. In this case, 
  $\mu(K)=\langle\rho(i),\mu'(K)\rangle$,
  $\mu(v)=\langle\rho(i),\mu'(v)\rangle$, and
  $\kappa(K,u)=\kappa_i(K,u)$.
  In this case, we simply set $I'(\mu(K),\kappa(K,u),\mu(v)) := I(\mu'(K),\kappa(K,u),\mu'(v))$.
\end{proof}

\subsection{Graphs with Skinny Tree-Decompositions}
\label{sec:skinny}

\begin{lem}\label{skinny_labelling}
  Let $\mathcal{G}$ be a hereditary class of graphs closed under taking disjoint union.
  Let $k:\N\to\N^+$ and $b:\N\to\R^+$ be functions with $k(n)\in\Oh(\log n)$ and with $b(n)>1$ for all $n$.
  Let $\mathcal{G}'$ be the class of graphs consisting of, for each $n\in\N$, every $n$-vertex graph $G$ with a rooted $b(n)$-skinny tree-decomposition of adhesion-width at most $k(n)$ in which each torso is a member of $\mathcal{G}$.
  Let $g_i:\N\to\R$ be a non-decreasing function with $g_i(n)\in \Oh(\log n)$ for each $i\in[3]$.
  Suppose that $\mathcal{G}$ admits a weighted $(g_1,g_2,g_3)$ mixed labelling scheme.
  Then $\mathcal{G}'$ admits a weighted $(g_1',g_2',g'_3)$ mixed labelling scheme, where
  \begin{align*}
    g'_1(n)
     & = g_1(n) + \Oh(k(n)\log b(n) + k(n) \log\log n)\enspace, \\
    g'_2(n)
     & = g_2(n)+\Oh(\log b(n)+\log\log n)\enspace,              \\
    g'_3(n)
     & = g_3(n) + \Oh(k(n) + \log\log n)\enspace.
  \end{align*}
\end{lem}

\begin{proof}
  By~\cref{lem:nice_weights} it is enough to show that $\mathcal{G}'$ admits a weak weighted $(g_1',g_2',g_3')$ mixed labelling scheme. We will describe an adjacency tester $A'$ and an identity tester $I'$ for which we can construct a mixed labelling $(\mu,\kappa)$ for any $(G^+,G,\omega)$ where $G^+\in\mathcal{G}'$, $G$ is a spanning subgraph of $G^+$, and $\omega:V(G)\to\N^+$ is a nice weight function.
  Since it is not possible to describe $A'$ and $I'$ without knowing the contents of the labels, we first describe how to compute $\mu$ and $\kappa$ for a particular $(G^+,G,\omega)$.

  Let $G^+$ be an $n$-vertex graph in $\mathcal{G}'$, let $G$ be a spanning subgraph of $G^+$, and let $\omega:V(G)\to\N^+$ be a nice weight function.  By definition, $G^+$ has a $b$-skinny tree-decomposition $\mathcal{T}:=(T,(B_x\mid x\in V(T)))$ in which each torso belongs to $\mathcal{G}$ and each adhesion has size at most $k$, for integers $b =\lfloor b(n)\rfloor$ and $k = k(n)$. Note that $b,k\geq1$.

  Let $p:=\height(T)+1$.
  For each $i\in[p]$, let $B_i:=\bigcup_{x\in L_i(T)} B_x$.  Then $\mathcal{P}:=(B_1,\ldots,B_p)$ is a path-decomposition of $G$ in which each adhesion has size at most $bk$.
  Let $A_1:=\emptyset$ and, for each $i\in\{2,\ldots,p\}$, let $A_i:=B_i\cap B_{i-1}$.

  Let
  \begin{align*}
    G^\star
     & := \cup\set{\torso{G^+}{\mathcal{T},B_x}\mid x\in V(T)}\enspace, \\
    \intertext{and for each $i\in[p]$, let }
    G^\star_i
     & := G^\star[B_i]-A_i\enspace,                                     \\
    G_i
     & := G[B_i]-A_i\enspace.
  \end{align*}
  Note that for each $i\in[p]$ and each $x\in L_i(T)$, $\torso{G}{\mathcal{T},B_x}\in\mathcal{G}$
  (by definition), and also since $\mathcal{G}$ is hereditary, $G^\star[B_x]\in\mathcal{G}$.  Since $\mathcal{G}$ is closed under disjoint union, $G^\star_i\in\mathcal{G}$.

  \paragraph{\boldmath The layer label $\rho$:}
  Define the weight function $\psi:[p]\to\N^+$ by $\psi(i):=\omega(G_i)$.
  Since $\set{B_i\setminus A_i\mid i\in[p]}$ is a partition of $V(G)$, we have  that $\psi([p])=\sum_{i\in[p]}\psi(i)=\omega(G)$.
  Apply \cref{shannon_fano_elias_tree} and \cref{shannon_fano_elias_code} to the set $[p]$ with weight function $\psi$, and let $T_\psi$ and $\rho:[p]\to\{0,1\}^*$ be the resulting binary tree and code, respectively.
  For each $i\in[p]$ and each $v\in G_i$, define $\rho(v):=\rho(i)$, so
  \begin{equation}
    \label{eq:rho}
    |\rho(v)| \leq \log\psi([p])-\log\psi(i)+3=\log\omega(G)-\log\omega(G_i)+3 \enspace .
  \end{equation}
  Also
  \begin{equation}
    \label{eq:height-T-psi}
    \begin{split}
      \height(T_{\psi})
       & \leq \log\omega(G)-\log\omega(G_i)+3 \\
       & \in\Oh(\log n)\enspace.
    \end{split}
  \end{equation}

  Let $(A,I)$ be the pair of adjacency and identity testers for a weighted $(g_1,g_2,g_3)$ mixed labelling scheme of $\mathcal{G}$.
  For each $i\in[p]$ let $(\mu_i,\kappa_i)$ be a mixed labelling of $(G^\star_i, G_i,\omega_{|V(G_i)})$ for $(A,I)$.
  For each $i\in[p]$, each $v\in V(G_i)$, and each clique $K$ in $G^\star_i$,
  \begin{align}
    |\mu_i(v)|
     & \le \log\omega(G_i) - \log\omega(v) + g_1(n)
    \enspace,\label{eq:mu-i}                                                                   \\
    |\mu_i(K)|
     & \le \log\omega(G_i) - \log\min_{v\in K}(\omega(v)) + g_3(n)\enspace , \label{eq:mu-i-K}
    \intertext{and for each $u\in K$,}
    |\kappa_i(K,u)|
     & \le g_2(n) \enspace.
  \end{align}
  For each $i\in[p]$ and each $v\in B_i\setminus A_i$, let
  \[
    \mu'(v):=\mu_i(v).
  \]

  \paragraph{\boldmath The adhesion identifier $\beta(v)=(d(v),\varphi(v))$:}
  For each vertex $v$ in $G$, let
  \begin{align*}
    a(v)
     & := \min\set{i\in[p]\mid v\in B_i}\enspace, \\
    b(v)
     & := \max\set{i\in[p]\mid v\in B_i}\enspace, \\
    \lca(v)
     & :=\lca(T_{\psi},\set{a(v),b(v)})\enspace,  \\
    d(v)
     & := \depth_{T_{\psi}}(\lca(v))\enspace.
  \end{align*}

  Let $C:=\bigcup_{i\in[p]} A_i$  be the set of vertices of $G$ that participate in adhesions.
  Let $x$ be a node in $T_{\psi}$.
  Define $L_x := \set{v\in C \mid \lca(v)=x}$.
  We claim that $|L_x| \leq bk$.
  In order to prove this, consider a vertex $v\in C$ such that $\lca(v)=x$.
  Since $v\in C$, we have that $a(v) < b(v)$, so
  $\lca(v)=x$ must have two children in $T_{\psi}$.
  Consider the largest leaf $i$ in the subtree of the left child of $x$ in $T_{\psi}$ and
  the smallest leaf $j$ in the subtree of the right child of $x$ in $T_{\psi}$.
  Then, $j=i+1$, and  $a(v)\leq i < i+1 \leq b(v)$.
  Therefore, $v\in B_{i}\cap B_{i+1}$.
  We conclude that
  $|L_x| \leq |B_i\cap B_{i+1}| \leq bk$, as desired.

  For each node $x$ in $T_{\psi}$, define $\varphi_x: L_x\to[bk]$ to be an arbitrary injective function (which exists because $|L_x|\le bk$). Now,  define $\varphi: C\to [bk]$ so that
  $\varphi(v):=\varphi_{\lca(v)}(v)$ for each $v\in C$.

  For each $v\in V(G)$, let
  \[
    \beta(v) :=
    \begin{cases}\langle d(v),\varphi(v)\rangle & \textrm{if $v\in C$,}     \\
             \langle \varepsilon\rangle     & \textrm{if $v\not\in C$.}
    \end{cases}
  \]
  Clearly,
  \begin{align}
    |\bin(d(v))|+|\bin(\varphi(v))|
     & \leq \lceil\log(\height(T_{\psi})+1)\rceil+\lceil\log(bk+1)\rceil \notag                                       \\
     & \in\Oh(\log\log n +\log(bk))                                             & \textrm{by~\eqref{eq:height-T-psi}}
    \label{beta_v_size}                                                                                               \\
     & =\Oh(\log\log n + \log b),                                               & \textrm{since $k(n)\in\Oh(\log n)$}
  \end{align}
  and
  \begin{align*}
    |\beta(v)|
     & \leq |\bin(d(v))|+|\bin(\varphi(v))| + \Oh(1)+ \Oh(\lgg|\bin(d(v))|+\lgg|\bin(\varphi(v))|)
     & \text{by \eqref{encoding_cost}}                                                             \\
     & = \Oh(\log\log n +\log b)\enspace.
  \end{align*}
  Finally, for each $i\in[p]$ and each $w\in V(G_i)$, let
  \[
    \alpha(w):=\langle\beta(v)\mid v\in N_G(w)\cap A_i\rangle\enspace.
  \]
  Note that $|N_G(w)\cap A_i| \leq k$; indeed, if $w\in B_x$ with $x\in V(T)$ and $y$ is the parent of $x$ in $T$, then $N_G(w)\cap A_i \subseteq B_x\cap B_y$, and $|B_x\cap B_y|\leq k$ since $\mathcal{T}$ has adhesion-width at most $k$.
  Therefore,
  \begin{align}
    \begin{split}
      |\alpha(w)|
       & \le \textstyle\sum_{v\in N_G(w)\cap A_i}|\beta(v)| + \Oh(\lgg|N_G(w)\cap A_i|) + \Oh(\textstyle\sum_{v\in N_G(w)\cap A_i}\lgg|\beta(v)|) \\
       & \in\Oh(k\log\log n + k\log b)\enspace,                                                                                                   %
    \end{split}
    \label{alpha_len}
  \end{align}
  which follows by~\eqref{encoding_cost} and \eqref{beta_v_size}.

  \paragraph{\boldmath The vertex label $\mu(v)$:}
  For each $v\in V(G)$, define
  \[
    \mu(v):=
    \langle \rho(v)\,,\, \mu'(v)\,,\, \alpha(v)\,,\, \beta(v) \rangle \enspace .
  \]
  For each $v\in V(G)$,
  \begin{align}
    |\rho(v)|+|\mu'(v)|
     & \le \log\omega(G)-\log\omega(v)+g_1(n)+3 &  & \textrm{by~\eqref{eq:rho} and~\eqref{eq:mu-i}} \label{rho_mu_len} \\
     & \in \Oh(\log n) \enspace ,               &  & \textrm{since $\omega$ is nice and $g_1\in\Oh(\log n)$.}
    \label{rho_mu_len_x}
  \end{align}
  Therefore, for each $v\in V(G)$, \eqref{rho_mu_len},~\eqref{alpha_len},~and \eqref{beta_v_size} imply
  \begin{align*}
    |\rho(v)|+|\mu'(v)| + |\alpha(v)| + |\beta(v)|
     & \leq
    \log\omega(G)-\log\omega(v) + g_1(n) + \Oh(k\log\log n+k\log(bk)) \\
     & \in\Oh(\log n + k\log\log n+k\log b)\enspace,
  \end{align*}
  and
  \begin{align*}
    |\mu(v)|
     & \leq |\rho(v)|+|\mu'(v)| + |\alpha(v)| + |\beta(v)| + \Oh(\lgg|\rho(v)|+\lgg|\mu'(v)| + \lgg|\alpha(v)| + \lgg|\beta(v)|) \\
     & \leq \log\omega(G)-\log\omega(v) + g_1(n) + \Oh(k\log\log n+k\log b) \enspace .
  \end{align*}
  Thus, $\mu$ satisfies the condition on $g'_1(n)$.

  \paragraph{Adjacency Testing:}
  Now that the vertex labels are defined, we can describe the adjacency testing function $A'$.
  Given the labels $\mu(v)=\langle\rho(v),\mu'(v),\alpha(v),\beta(v)\rangle$ and $\mu(w)=\langle\rho(w),\mu'(w),\alpha(w),\beta(w)\rangle$ for two vertices $v$ and $w$ of $G$, the adjacency tester $A'$ works as follows. First, the tester compares $\rho(v)$ and $\rho(w)$:
  \begin{itemize}
    \item If $\rho(v)=\rho(w)$ then $v$ and $w$ are both vertices of $G_i$, for some $i\in[p]$, and $vw\in E(G)$ if and only if $vw\in E(G_i)$.  In this case, we set $A'(\mu(v),\mu(w))=A(\mu'(v),\mu'(w))=A(\mu_i(v),\mu_i(w))$.

    \item If $\rho(v)\neq\rho(w)$ then assume without loss of generality that $\rho(v)$ is lexicographically less than $\rho(w)$. In this case $v\in V(G_i)$ and $w\in V(G_j)$ for $i=a(v)<a(w)=j$. The edge $vw$ can only be present in $G$ if $v\in A_j$. First the tester checks whether $\beta(v)$ is not empty to see if $v\in C$. If $\beta(v)=\langle\varepsilon\rangle$, then $v\not\in C$, and  $v\not\in A_j$, so $A'(\mu(v),\mu(w))=0$.

          Otherwise, $\beta(v)=\langle d(v),\varphi(v)\rangle$.  Then the length-$d(v)$ prefix of $\rho(v)$ defines the path from the root of $T_\psi$ to $\lca(v)$.  By definition, $\lca(v)$ is a $T_\psi$-ancestor of all integers in $[a(v),b(v)]$. In this case, the tester checks if $\rho(v)$ and $\rho(w)$ have the same prefix of length $d(v)$.  If they do not, then $\lca(v)$ is not a $T_\psi$-ancestor of $j=a(w)$, so $j\not\in[a(v),b(v)]$.  Therefore, $v\not\in B_{j}$, which implies $v\not\in A_j$. Therefore $vw\not\in E(G)$, so $A'(\mu(v),\mu(w))=0$.

          Suppose that $\rho(v)$ and $\rho(w)$ have the same prefix of length $d(v)$.  Then $j\in [a(v),b(v)]$ and $v\in A_j$. In this case the tester inspects each item of $\alpha(w)$. Let $(d_0,\varphi_0)$ be such an item and say that $(d(u_0),\varphi(u_0))=(d_0,\varphi_0)$ for $u_0 \in N_G(w)\cap A_j$. Since $u_0 \in A_j$, we have that $j\in[a(u_0),b(u_0)]$. Therefore, both nodes $\lca(v)$ and $\lca(u_0)$ are $T_\psi$-ancestors of $j$. In other words, both $\lca(v)$ and $\lca(u_0)$ lie on the path from the root to $j$ in $T_{\psi}$. Therefore, $\lca(v)=\lca(u_0)$ if and only if $d(v)=d(u_0)$. Furthermore, $v=u_0$ if and only if $(d(v),\varphi(v))=(d(u_0),\varphi(u_0))$. Thus, if $(d_0,\varphi_0)=(d(v),\varphi(v))$ then the tester concludes that the corresponding neighbour $u_0=v$, so $A'(\mu(v),\mu(w))=1$. Finally, if no item $(d_0,\varphi_0)$ in $\alpha(w)$ satisfies $(d_0,\varphi_0)=(d(v),\varphi(v))$, then $v\not\in N_G(w)\cap A_j$ so $v$ and $w$ are not adjacent in $G$, and we set $A'(\mu(v),\mu(w))=0$.
  \end{itemize}

  \paragraph{\boldmath The clique labels $\mu(K)$ and local identifiers $\kappa(K,v)$:}
  For each clique $K$ in $G^\star$,
  let
  \[
    j(K):=\min\{j\in[p]:K\subseteq B_1\cup\cdots\cup B_j\}\enspace.
  \]
  Fix some clique $K$ in $G^\star$  and let $j=j(K)$.  By the definition of $j(K)$, we have that $B_j\cap K\setminus A_j=K\setminus A_j$ is non-empty.
  Therefore $K\setminus A_j$ is assigned a label $\mu_j(K\setminus A_j)$ in the mixed labelling of $(G^\star_j,G_j,\omega_{|V(G_j)})$ for $(A,I)$.
  Since $K\setminus A_j$ is non-empty, let $v$ be an arbitrary vertex in $K\setminus A_j$. Let $x=\lca(v)$, so $x\in L_j(T)$ and $v\in B_x\setminus B_y$, where $y$ is the parent of $x$ in $T$ (if $x$ is the root, let $B_y=\emptyset$). Since the sets $B_z\setminus B_{z'}$ are disjoint for all nodes $z\in L_j(T)$ with parent $z'$, the node $x$ is uniquely determined by $K\setminus A_j$. Because $K$ is a clique, $K\subseteq B_x$, and thus $K\cap A_j \subseteq B_x \cap B_y$. Since $\mathcal{T}$ has adhesion-width at most $k$, $|B_x \cap B_y|\le k$. Let $c(K)$ be a $k$-bit string whose $i$-th bit indicates whether the $i$-th vertex of $B_x\cap B_y$ (ordered canonically, for example by their $\beta$ values) belongs to $K$.
  Define the clique label
  \[
    \mu(K):=\langle\rho(j),\mu_j(K\setminus A_j), c(K)\rangle \enspace .
  \]
  To see that $\mu$ is injective when restricted to cliques, consider two distinct cliques $K$ and $L$ in $G^\star$. If $j(K)\neq j(L)$, their first parts differ since $\rho$ is a prefix-free code. If $j(K)=j(L)=j$, then either $K\setminus A_j\neq L\setminus A_j$ or $K\cap A_j\neq L\cap A_j$. In the former case, $\mu_j(K\setminus A_j)\neq \mu_j(L\setminus A_j)$ since $\mu_j$ is injective on cliques of $G^\star_j$. In the latter case, $K\setminus A_j=L\setminus A_j$ implies they belong to the same bag $B_x$, so $c(K)\neq c(L)$. Thus $\mu(K)\neq \mu(L)$.

  Then
  \begin{align}
    |\rho(j)|+|\mu_j(K\setminus A_j)| + |c(K)|
     & \le \log\omega(G)-\log\min_{v\in K\setminus A_j}(\omega(v))+g_3(n)+3+k
     & \textrm{by~\eqref{eq:mu-i-K}}                                                      \\
     & \le \log\omega(G)-\log\min_{v\in K}(\omega(v))+g_3(n)+k+3   \label{rho_mu_k_len_i} \\
     & \in \Oh(\log n)
    \label{rho_mu_k_len}
  \end{align}
  since $\omega$ is nice, $g_3(n)\in\Oh(\log n)$, and $k\in\Oh(\log n)$.  Thus
  \begin{align*}
    |\mu(K)|
     & \le |\rho(j)|+|\mu_j(K\setminus A_j)|+|c(K)| + \Oh(\lgg|\rho(j)|+\lgg|\mu_j(K\setminus A_j)|+\lgg|c(K)|)                                                                         \\
     & \le \log\omega(G)-\log\min_{v\in K}(\omega(v)) + g_3(n)+\Oh(k+\log\log n)                                & \text{by \eqref{rho_mu_k_len_i} and \eqref{rho_mu_k_len} } \enspace .
  \end{align*}
  Thus, $\mu$ satisfies the condition on $g'_3(n)$.

  Recall that the local identifier $\kappa_j(K\setminus A_j,u)$ is defined for each $u\in K\setminus A_j$.
  For each $u\in K$, define
  \[
    \kappa(K,u):=
    \begin{cases}
      \langle 0,\kappa_j(K\setminus A_j,u)\rangle & \text{if $u\in B_j\setminus A_j$} \\
      \langle 1,\beta(u)\rangle                   & \text{if $u\in A_j$.}
    \end{cases}
  \]
  In the first case, $|\kappa(K,u)|\le g_2(n)+\Oh(\log\log n)$.
  In the second case, $|\kappa(K,u)|\in \Oh(\log (b)+\log\log n)$.  Thus, $\kappa$ satisfies the requirements for $g'_2(n)$.

  \paragraph{Identity Testing:}
  We now describe the identity testing function $I'$.
  Suppose $I'$ is given the labels $\mu(K)$, $\kappa(K,u)$, and $\mu(v)=\langle \rho(v),\mu'(v),\alpha(v),\beta(v)\rangle$ for some clique $K$ in $G^\star$, some $u\in K$, and some $v\in V(G)$. The label $\mu(K)$ always has the form $\langle\rho(j),\mu_j(K\setminus A_j), c(K)\rangle$ where $j=j(K)$. The identity tester $I'$ simply ignores the third part of $\mu(K)$ and extracts $\rho(j)$ and $\mu_j(K\setminus A_j)$.
  The identity tester $I'$ first compares $\rho(j)$ and $\rho(v)=\rho(a(v))$.
  \begin{itemize}
    \item If $\rho(j)$ is lexicographically smaller than $\rho(a(v))$, then $j<a(v)$.  Since $K\subseteq\bigcup_{i=1}^j B_i$ and $v\not\in B_{j'}$ for any $j'<a(v)$, this implies that $v\not\in K$.  Therefore $v\neq u$, since $u\in K$. In this case $I'(\mu(K),\kappa(K,u),\mu(v)):=0$.

    \item If $\rho(j)=\rho(a(v))$ then $K\setminus A_j$ is a clique in $G^\star_j$ and $v\in B_j\setminus A_j$.
          In this case, $v$ can only be equal to $u$ if $u\in K\setminus A_j$.
          The first part $b$ of $\kappa(K,u)$ is a single bit that indicates if $u\in K\setminus A_j$.  If $b=1$ then $u\not\in K\setminus A_j$, so $I'(\mu(K),\kappa(K,u),\mu(v)):=0$.  If $b=0$ then $u\in K\setminus A_j$ and $\kappa(K,u)=\langle 0,\kappa_j(K\setminus A_j,u)\rangle$.
          In this case, $\mu'(v)=\mu_j(v)$ and we set
          $I'(\mu(K),\kappa(K,u),\mu(v)):=I(\mu_j(K\setminus A_j),\kappa_j(K\setminus A_j,u),\mu_j(v))$.

    \item If $\rho(j)$ is lexicographically larger than $\rho(a(v))$ then $K\setminus A_j$ is a clique in $G^\star_j$ and $v\in B_i\setminus A_i$ for some $i=a(v) < j$.  The tester examines the first part $b$ of $\kappa(K,u)$ to determine if $u\in K\setminus A_j$.  If $b=0$ then  $u\in B_j\setminus A_j$ and $v\in B_i\setminus A_i$, so $u\neq v$ and $I'(\mu(K),\kappa(K,u),\mu(v)):=0$.  If $b=1$ then $u\in A_j$ and $\kappa(K,u)=\langle 1,\beta(u)\rangle$. Since $u\in A_j$, $a(u) < j\le b(u)$.

          In particular, $u\in A_j$, so $\lca(u)$ is a $T_\psi$-ancestor of $j$.   Therefore, the first $d(u)$ bits of $\rho(j)$ are equal to the first $d(u)$ bits of $\rho(u)$.  The first $d(u)$ bits of $\rho(u)$ describe the path from the root of $T_\psi$ to $\lca(u)$.  The tester then compares the first $d(u)$ bits of $\rho(j)$ to the first $d(u)$ bits of $\rho(v)$.  If these are not equal,  then $\lca(u)\neq\lca(v)$, so $u\neq v$ and $I'(\mu(K),\kappa(K,u),\mu(v)):=0$.

          Since $u\in A_j$, we have that $u\in C$.  The tester examines $\beta(v)$ to determine if $v\in C$.  If $\beta(v)=\langle\varepsilon\rangle$ then $v\not\in C$, so $u\neq v$ and $I'(\mu(K),\kappa(K,u),\mu(v)):=0$. Otherwise $v\in C$ and $\beta(v)=\langle d(v),\varphi(v)\rangle$. Since we have already verified that the first $d(u)$ bits of $\rho(j)$ equal the first $d(u)$ bits of $\rho(v)$, we have that $\lca(u)=\lca(v)$ if and only if $d(u)=d(v)$. Thus, $v=u$ if and only if $(d(u),\varphi(u))=(d(v),\varphi(v))$.
          We set $I'(\mu(K),\kappa(K,u),\mu(v))=1$ if $(d(u),\varphi(u))=(d(v),\varphi(v))$ and $I'(\mu(K),\kappa(K,u),\mu(v))=0$ otherwise.\qedhere
  \end{itemize}
\end{proof}

\subsection{Graphs with Short Tree-Decompositions}
\label{sec:short}

\begin{lem}\label{small_height}
  Let $\mathcal{G}$ be a hereditary class of graphs closed under taking disjoint union. Let $k:\N\to\N^+$,  $h:\N\to\R$ be functions with $k(n),h(n)\in \Oh(\log n)$ and with $h(n)>0$ for all $n$.   Let $\mathcal{G}'$ be the class of graphs consisting of, for each $n\in\N$, every $n$-vertex graph $G$ with a rooted tree-decomposition $\mathcal{T}:=(T,(B_x \mid x\in V(T)))$ such that
  \begin{enumerate}[nosep,nolistsep,label=(\rm{\roman*)}]
    \item $\mathcal{T}$ has adhesion-width at most $k(n)$;
    \item $\height(T)\le h(n)$;
    \item for each node $x$ in $T$, the torso $\torso{G}{\mathcal{T},B_x}$ belongs to $\mathcal{G}$.
  \end{enumerate}
  Let $g_i:\N\to\R^+$ be a non-decreasing function  with $g_i\in\Oh(\log n)$ for each $i\in[3]$. Suppose that $\mathcal{G}$ admits a weighted $(g_1,g_2,g_3)$ mixed labelling scheme. Then $\mathcal{G}'$ admits a weighted $(g'_1,g'_2,g'_3)$ mixed labelling scheme, where
  \begin{align*}
    g'_1(n)
     & = g_1(n)
    + k(n)\cdot (g_2(n)+\Oh(\log\log n))
    + h(n)\cdot (g_3(n)+\Oh(\log\log n)) +\Oh(\log\log n)\enspace, \\
    g'_2(n)
     & = g_2(n) + \Oh(\log\log n)\enspace,                         \\
    g'_3(n)
     & = (h(n)+1)\cdot (g_3(n)+\Oh(k(n)+\log\log n))\enspace.
  \end{align*}
\end{lem}

\begin{proof}
  We will describe an adjacency tester $A'$ and an identity tester $I'$ for which we can construct a mixed labelling $(\mu,\kappa)$ for any $(G^+,G,\omega)$ where $G^+\in\mathcal{G}'$, $G$ is a spanning subgraph of $G^+$, and $\omega:V(G)\to\R^+$ is a weight function.
  Since it is not possible to describe $A'$ and $I'$ without knowing the contents of the labels, we first describe how to compute $\mu$ and $\kappa$ for a particular input tuple.

  The proof is by induction on a slightly more detailed statement. An \defin{instance} is a tuple $(n,h,k,G^+,G,\omega,\mathcal{F})$ with the following properties:
  $n$ is an integer,
  $h\le h(n)$ is a non-negative integer,
  $k=k(n)$ is a positive integer,
  $G^+\in\mathcal{G}'$ with $n\geq |V(G^+)|$,
  $G$ is a spanning subgraph of $G^+$,
  $\omega:V(G)\to\N^+$
  is a weight function with
  \begin{equation}
    \log\omega(G)-\log \omega(v)\le n+2,
    \label{eq:short-td-nice-weights}
  \end{equation}
  for each $v\in V(G)$, and
  $\mathcal{F}:=(F,(B_x \mid x\in V(F)))$ is a rooted tidy forest-decomposition of $G^+$ having adhesion-width at most $k=k(n)$ and with $F$ of height at most $h$.

  Let $(n,h,k,G^+,G,\omega,\mathcal{F})$ be an instance.
  For each $y\in V(F)$, let
  \[
    A_y:=\begin{cases}
      \emptyset   & \textrm{if $y$ is a root in $F$,}                                      \\
      B_y\cap B_x & \textrm{if $y$ is not a root in $F$ and $x$ is the $F$-parent of $y$.}
    \end{cases}
  \]
  We prove by an induction on $h$ that there is a $(g_1',g_2',g_3')$ mixed labelling $(\mu,\kappa)$ of $(G^+,G,\omega)$ for $(A',I')$ satisfying the following additional properties: for each $z\in V(F)$ and each $w\in B_z\setminus A_z$,
  \begin{align}
    |\mu(w)|
     & \le \log \omega(G) - \log \omega(w) + g_1(n)\notag                           \\
     & \quad {} + |A_z|\cdot (g_2(n) + \Oh(\log\log n))\label{strong_strong}        \\
     & \quad {} + h\cdot(g_3(n) + \Oh(\log\log n)) + \Oh(\log\log n)\enspace,\notag
  \end{align}
  and for each clique $K$ of $G^+$,
  \begin{equation}
    |\mu(K)| \le \log\omega(G) - \log\min_{v\in K}(\omega(v)) + (h+1)\cdot(g_3(n) + \Oh(k+\log\log n))\enspace.\label{strong_clique}
  \end{equation}
  Given an integer $n$, an $n$-vertex graph $G^+\in\mathcal{G}'$, a spanning subgraph $G$ of $G^+$, and a weight function $\omega:V(G)\to\R^+$, let $\mathcal{T}=(T,(B_x \mid x\in V(T)))$ be a rooted tree-decomposition of $G^+$ of adhesion-width at most $k(n)$ and with $\height(T)\leq h(n)$ and such that each torso of $\mathcal{T}$ is in $\mathcal{G}$. Let $\mathcal{F}_0:=(F,(B_x \mid x\in V(F)))$ be a rooted tidy forest-decomposition of $G^+$ (obtained by applying \cref{tidy_forest} to $\mathcal{T}$).  Let $(n_0,h_0,k_0,G^+,G,\omega_0,\mathcal{F}_0)$ be an instance where $n_0:=n=|V(G)|$, $h_0:=h(n)$, $k_0:=k(n)$, $\omega_0:V(G)\to\N^+$ is a nice weight function obtained by applying \cref{obs:nice} to $\omega$, $\mathcal{F}_0$ is a rooted tidy forest-decomposition obtained by applying \cref{tidy_forest} to $\mathcal{T}$.  (The definition of $\omega_0$ satisfies \eqref{eq:short-td-nice-weights}, by \cref{obs:nice}.)
  By our inductive statement we obtain labelling functions $(\mu,\kappa)$ for $(A',I')$.
  The labels lengths are bounded in terms of $\omega_0$ which is enough because $\log\omega_0(G) - \log\omega_0(v) \leq \log\omega(G)-\log\omega(v)+2$, which follows from~\cref{obs:nice}.
  This will complete the proof of the lemma.

  We now begin the inductive proof. Let $(n,h,k,G^+,G,\omega,\mathcal{F})$ be an instance, where $\mathcal{F}:=(F,(B_x\mid x\in V(F)))$.
  Let $(A,I)$ be the pair of adjacency and identity testers for a weighted $(g_1,g_2,g_3)$ mixed labelling scheme of $\mathcal{G}$.  Let
  \[
    G^\star:= \cup\set{\torso{G^+}{\mathcal{F},B_z}\mid z \in V(F)}\enspace.
  \]
  By definition, $\mathcal{F}$ is a forest-decomposition of $G^\star$.  By the assumptions of the lemma,
  \begin{align}
    G^\star[B_z] = \torso{G^+}{\mathcal{F},B_z}
     & \in \mathcal{G}\enspace,
    \label{eq:torso-in-the-class-G}
  \end{align}
  for each $z\in V(F)$.  Our labelling $\mu$ will assign a label $\mu(v)$ to each vertex $v$ of $G^\star$ and a label $\mu(K)$ to each clique $K$ of $G^\star$.  Since $G^\star$ is a supergraph of $G^+$, this ensures that $\mu$ assigns a label to each vertex and clique of $G^+$.

  \paragraph{\boldmath The root labelling $\mu_R$:}
  Let $R$ be the set of roots of $F$. Let
  \begin{align*}
    B_R
     & :=\textstyle\cup\set{B_r\mid r\in R}\enspace.
    \intertext{Since $\mathcal{F}$ is tidy, each bag of $\mathcal{F}$ is non-empty, so  $B_R\neq\emptyset$. Let}
    G^\star_R
     & :=G^\star[B_R]\enspace,                       \\
    G_R
     & :=G[B_R]\enspace.
  \end{align*}
  Observe that $G^\star_R\in\mathcal{G}$, since
  the sets in $\set{B_r \mid r\in R}$ are pairwise disjoint, $G^\star[B_r] \in \mathcal{G}$ (by~\eqref{eq:torso-in-the-class-G}) for each $r\in R$, and
  $\mathcal{G}$ is closed under disjoint union.  Let
  \begin{align*}
    \mathcal{K}_R
     & :=\set{B_y\cap B_R\mid y\in L_2(F)}\enspace.
    \intertext{Note that if $h=0$ then $\mathcal{K}_R$ is empty. Since $\mathcal{F}$ is tidy, each $K$ in $\mathcal{K}_R$ is nonempty. Note also that each $K$ in $\mathcal{K}_R$ is a clique in $G^\star_R$.
      For each $K$ in $\mathcal{K}_R$, let} F_K
     & :=\textstyle\cup \set{F_y\mid y\in L_2(F) \textrm{ and } B_y\cap B_R = K}\enspace, \\
    G_K
     & :=\textstyle\cup\set{G[B_z] \mid z\in V(F_K)}- B_R\enspace.
  \end{align*}
  Note that for each $K\in\mathcal{K}_R$,
  since $\mathcal{F}$ is tidy,
  $G_K$ has at least one vertex.
  Thus, the family of sets $\set{\set{B_R}}\cup\set{V(G_K)\mid K\in \mathcal{K}_R}$ is a partition of vertices of $G$.

  For each $v\in B_R$, let
  \begin{align*}
    \delta(v)
     & :=k\cdot\omega(v)+\textstyle\sum\set{\omega(G_K)\mid K\in\mathcal{K}_R,\ v\in K} \enspace .
  \end{align*}
  Since $k\geq1$, $\delta(v)\geq\omega(v)\geq 1$ for each $v\in B_R$.
  Clearly, for each $v\in B_R$,
  \begin{equation}
    \label{eq-gamma-v-omega-v}
    \delta(v) \geq k\cdot\omega(v)\enspace.
  \end{equation}
  Note also that for each $K\in\mathcal{K}_R$ and for each $v\in K$, we have that $\delta(v)\ge \omega(G_K)$, and therefore (recalling that $K$ is non-empty)
  \begin{equation}
    \min_{v\in K}\delta(v) \ge \omega(G_K) \enspace . \label{gamma_v_lower_bound}
  \end{equation}
  Since $B_R\neq\emptyset$, we have that $\delta(B_R)\geq1$.
  Since $\mathcal{F}$ has adhesion-width at most $k$, each clique in $\mathcal{K}_R$ has size at most $k$, so
  \begin{equation}
    \delta(B_R)=\sum_{v\in B_R}\delta(v)
    \le k\cdot\omega(B_R) + k\sum_{K\in\mathcal{K}_R}\omega(G_K)
    = k\cdot\omega(G)  \enspace . \label{gamma_br}
  \end{equation}

  The cliques of $G^\star_R$ in $\mathcal{K}_R$ play an important role in our labelling scheme.
  Each clique in $\mathcal{K}_R$ is defined by an adhesion $A_y=B_r\cap B_y$ for at least one $y\in L_2(F)$.  The label $\mu_R(K)$ for a clique $K\in\mathcal{K}_R$ will be used a part of the label $\mu(v)$ for each vertex $v$ in $G_R$ and as part of the label $\mu(K')$ for each clique $K'$ of $G^\star$ that contains vertices in $G_R$.

  Altogether, $G^\star_R\in\mathcal{G}$, $G_R$ is a spanning subgraph of $G^\star_R$, $\delta:V(G_R)\to\N^+$ is a weight function. Let $(\mu_R,\kappa_R)$ be a mixed labelling of $(G^\star_R, G_R, \delta)$ for $(A, I)$. For each $v\in B_R$,
  \begin{align}
    |\mu_R(v)|
     & \le \log \delta(B_R)-\log\delta(v)+g_1(|B_R|)\notag                                                                                                                              \\
     & \le \log (k\cdot\omega(G)) - \log (k\cdot\omega(v)) + g_1(n)             &  & \textrm{by~\eqref{gamma_br}, ~\eqref{eq-gamma-v-omega-v}, and since $g_1$ is non-decreasing}\notag \\
     & = \log\omega(G) - \log\omega(v) + g_1(n)\label{eq:mu-R-bounded-in-omega}                                                                                                         \\
     & \in\Oh(\log n)                                                           &  & \textrm{by~\eqref{eq:short-td-nice-weights} and since $g_1(n)\in\Oh(\log n)$.}
    \label{eq:mu-R-in-logn}
  \end{align}

  Let $K$ be a clique in $G^\star_R$. We make use of two different bounds on the length of $\mu_R(K)$.  If $K\in\mathcal{K}_R$:
  \begin{align}
    |\mu_R(K)|
     & \le\log \delta(B_R) - \log \min_{v\in K}(\delta(v)) + g_3(|B_R|) \notag                                                                                                                            \\
     & \le\log \omega(G) - \log \omega(G_K) + g_3(n) + \log k
     & \textrm{by~\eqref{gamma_br},~\eqref{gamma_v_lower_bound} and since $g_3$ is non-decreasing} \label{mu_RK_i}                                                                                        \\
     & \in\Oh(\log n)                                                                                              & \textrm{by~\eqref{eq:short-td-nice-weights} and since $g_3(n), k(n)\in\Oh(\log n).$}
  \end{align}
  For each $K\not\in\mathcal{K}_R$,
  \begin{align}
    |\mu_R(K)|
     & \le\log \delta(B_R) - \log\min_{v\in K}(\delta(v)) + g_3(|B_R|) \notag                                                                                                          \\
     & \le\log(k\cdot\omega(G)) - \log\min_{v\in K}(k\cdot\omega(v)) + g_3(n)
     & \text{by \eqref{eq-gamma-v-omega-v},~\eqref{gamma_br} and since $g_3$ is non-decreasing}\notag                                                                                  \\
     & = \log \omega(G) - \log\min_{v\in K}(\omega(v)) + g_3(n)
    \label{mu_RK_ii}                                                                                                                                                                   \\
     & \in\Oh(\log n)                                                                                 & \textrm{by~\eqref{eq:short-td-nice-weights} and since $g_3(n)\in\Oh(\log n)$.}
    \label{eq:mu_R(K)-in-logn}
   \end{align}
  Moreover, for each clique $K$ in $G^\star_R$,
  \begin{align}
    |\kappa_R(K,v)|
     & \leq g_2(|B_R|) \leq g_2(n) \in \Oh(\log n) \enspace ,
    \label{eq:kappa_R-logn}
  \end{align}
  since $g_2$ is non-decreasing.%

  \paragraph{\boldmath The subforest labelling $\mu_K$:}
  Suppose that $h\geq1$. Thus, $F-R$ is a non-empty rooted forest. We now describe the subproblems on which we apply induction. For each $z\in V(F-R)$, let $C_z:=B_z\setminus B_R$ and
  note that $\bigcup_{z\in V(F_K)}C_z=V(G_K)$.

  Let $K\in\mathcal{K}_R$. Define
  \begin{align*}
    G^\star_K
     & :=G^\star[\textstyle\bigcup\set{C_z\mid z\in V(F_K)}]\enspace, \\
    \mathcal{F}_K
     & :=(F_K, (C_z \mid z\in V(F_K)))\enspace.
  \end{align*}
  Thus, $\mathcal{F}_K$ is a rooted tidy forest-decomposition of $G^\star_K$ of adhesion-width at most $k$ and the height of $F_K$ is at most $h-1$.
  Since $\mathcal{G}$ is hereditary and by~\eqref{eq:torso-in-the-class-G},
  \begin{align*}
    \torso{G^\star_K}{\mathcal{F}_K,C_z}
     & =G^\star[B_z]-B_R\subseteq G^\star[B_z]\in \mathcal{G}\enspace ,
  \end{align*}
  for every $z \in V(F_K)$.

  Altogether, $\mathcal{F}_K$ is a rooted tidy forest-decomposition of $G^\star_K$ of adhesion-width at most $k$ and with $\height(F_K)\le h-1$ and with each torso in $\mathcal{G}$, $G_K$ is a spanning subgraph of $G^\star_K$, and the restriction
  $\omega_{|V(G_K)}$ of $\omega$ to the vertices of $G_K$ is a weight function that satisfies \eqref{eq:short-td-nice-weights}.  Thus, $(n,h-1,k,G^\star_K,G_K,\omega_{|V(G_K)},\mathcal{F}_K)$ is an instance on which we can apply induction.
  Fix a weighted $(g_1'(n),g_2'(n),g_3'(n))$ mixed labelling $(\mu_K,\kappa_K)$ of $(G^\star_K,G_K,\omega_{|V(G_K)})$ for $(A',I')$ that is obtained from the inductive hypothesis.
  In addition to being a weighted $(g_1'(n),g_2'(n),g_3'(n))$ mixed labelling, $(\mu_K,\kappa_K)$ also satisfies \eqref{strong_strong}.  That is, for each $z\in V(F_K)$ and each $w\in C_z\setminus A_z$,
  \begin{align}
    |\mu_K(w)|
     & \le \log \omega(G_K) - \log\omega(w) + g_1(n) \notag                                \\
     & \quad {} + |A_z\setminus B_R|\cdot (g_2(n) + \Oh(\log\log n)) \label{mu_Kw}         \\
     & \quad {} + (h-1)\cdot (g_3(n) + \Oh(\log\log n)) +\Oh(\log\log n) \enspace . \notag
  \end{align}

  \paragraph{\boldmath The vertex labelling $\mu:V(G)\to\{0,1\}^*$:}
  The format of a vertex label depends on whether the vertex is in $B_R$ or not.
  For each $v\in B_R$, define
  \[
    \mu(v) := \langle\mu_R(v)\rangle\enspace.
  \]
  Thus, $\mu(v)$ has just one part and
  by~\eqref{eq:mu-R-in-logn} this part is of length in $\Oh(\log n)$.
  Therefore,
  for each $v\in B_R$
  \begin{align*}
    |\mu(v)|
     & \leq |\mu_R(v)| + \Oh(1) + \Oh(\lgg|\mu_R(v)|)                &  & \text{by~\eqref{encoding_cost}}                                           \\
     & \leq \log\omega(G) - \log\omega(v) + g_1(n) + \Oh(\log\log n) &  & \textrm{by~\eqref{eq:mu-R-bounded-in-omega} and~\eqref{eq:mu-R-in-logn}.}
  \end{align*}
  Thus, the length of $\mu(v)$ satisfies \eqref{strong_strong} when $v\in B_R$.

  Now, let $w\in V(G)\setminus B_R$.
  Then, $w\in B_z\setminus A_z$ for some $z\in V(F_K)$ and some $K\in\mathcal{K}_R$.
  Note that the existence of $w$ implies that $h\geq 1$.
  Define
  \begin{align*}
    \alpha(w)
     & :=\langle\kappa_R(K,v)\mid v\in A_z\cap K, vw\in E(G)\rangle\enspace, \\
    \mu(w)
     & := \langle\mu_R(K)\,,\, \mu_K(w)\,,\, \alpha(w) \rangle \enspace .
  \end{align*}
  Thus,
  \begin{align*}
    |\alpha(w)|
     & \leq \textstyle\sum_{v\in A_z\cap K} |\kappa_R(K,v)| +
    \Oh(\lgg|A_z\cap K|) + \Oh(\sum_{v\in A_z\cap K} \lgg |\kappa_R(K,v)|)
     & \textrm{by~\eqref{encoding_cost}}                                            \\
     & \leq |A_z\cap K|\cdot g_2(n) +
    \Oh(\lgg k) + |A_z\cap K|\cdot\Oh(\lgg g_2(n))
     & \textrm{by~\eqref{eq:kappa_R-logn}}                                          \\
     & \leq |A_z\cap B_R|\cdot (g_2(n) +\Oh(\log\log n)) + \Oh(\log\log n)\enspace,
  \end{align*}
  where the last line follows
  from $A_z\cap K= A_z\cap B_R$ and $k(n),g_2(n)\in\Oh(\log n)$.
  Also,
  \begin{align*}
    |\mu_R(K)|+|\mu_K(w)|
     & \le \log \omega(G) - \log\omega(G_K) + g_3(n) + \log k                            & \text{by~\eqref{mu_RK_i}} \\
     & \qquad {} + \log\omega(G_K) - \log\omega(w) + g_1(n)                              & \text{by~\eqref{mu_Kw}}   \\
     & \qquad {} + |A_z\setminus B_R|\cdot (g_2(n) + \Oh(\log\log n))                                                \\
     & \qquad {} + (h-1)\cdot g_3(n) + h\cdot\Oh(\log\log n)                                                         \\
     & \leq \log\omega(G) - \log\omega(w) + g_1(n)                                                                   \\
     & \qquad {} + |A_z\setminus B_R|\cdot (g_2(n) + \Oh(\log\log n))                                                \\
     & \qquad {} + h\cdot (g_3(n) + \Oh(\log\log n)) + \Oh(\log\log n) \enspace . \notag
  \end{align*}
  Altogether
  \begin{align*}
    |\mu_R(K)|+|\mu_K(w)|+|\alpha(w)|
     & \le \log \omega(G) - \log\omega(w) + g_1(n)                              \\
     & \qquad {} + |A_z|\cdot (g_2(n)+\Oh(\log\log n))                          \\
     & \qquad {} + h\cdot (g_3(n) + \Oh(\log\log n)) + \Oh(\log\log n)\enspace.
  \end{align*}
  Therefore,
  \begin{align*}
    |\mu(w)|
     & =|\mu_R(K)|+|\mu_K(w)|+|\alpha(w)|+ \Oh(\lgg|\mu_R(K)|+\lgg|\mu_K(w)|+\lgg|\alpha(w)|)
     & \text{by \eqref{encoding_cost}}                                                        \\
     & \le \log \omega(G) - \log\omega(w) + g_1(n)                                            \\
     & \qquad {} + |A_z|\cdot (g_2(n)+\Oh(\log\log n))                                        \\
     & \qquad {} + h\cdot (g_3(n) + \Oh(\log\log n)) + \Oh(\log\log n)\enspace.
  \end{align*}
  Thus, the length of $\mu(w)$ satisfies \eqref{strong_strong} for each $w\in V(G)\setminus B_R$.

  \paragraph{Adjacency Testing:}  At this point, $\mu(v)$ is defined for each $v\in V(G)$, so we can describe the operation of the adjacency tester $A'$.
  Given the labels $\mu(v)$ and $\mu(w)$ for two vertices $v$ and $w$ of $G$, $A'$ examines how many parts each label has.
  For an arbitrary vertex $u\in V(G)$, if $\mu(u)$ has one part, then $u\in B_R$,
  otherwise $u\in V(G_K)$ for some $K\in\mathcal{K}_R$.
  \begin{itemize}
    \item If $v$ and $w$ are both in $B_R$ then $\mu(v)=\langle\mu_R(v)\rangle$ and $\mu(w)=\langle\mu_R(w)\rangle$.  In this case, $vw\in E(G)$ if and only if $vw\in E(G_R)$, so the adjacency tester returns $A'(\mu(v),\mu(w)):=A(\mu_R(v),\mu_R(w))$.
    \item If exactly one of $v$ or $w$ is in $B_R$ then, assume without loss of generality that $v\in B_R$ and $w\in V(G_K)$ for some $K\in\mathcal{K}_R$.
          So $\mu(v)=\langle\mu_R(v)\rangle$ and $\mu(w)=\langle\mu_R(K),\mu_K(w),\alpha(w)\rangle$.
          In this case, $vw\in E(G)$ if and only if $v\in K\cap N_G(w)$.  Thus, $vw\in E(G)$ if and only if $v\in K$ and $\alpha(w)$ contains $\kappa_R(K,v)$.  In this case, the adjacency tester iterates through each $\kappa$ in $\alpha(w)$ checking if $I(\mu_R(K),\kappa,\mu_R(v))=1$.  If so, then $\kappa=\kappa_R(K,v)$ and $vw\in E(G)$ so $A'(\mu(v),\mu(w)):=1$. If $I(\mu_R(K),\kappa,\mu_R(v))=0$ for every $\kappa$ in $\alpha(w)$ then $v\not\in K\cap N_G(w)$, so  $vw\not\in E(G)$, and $A'(\mu(v),\mu(w)):=0$.
    \item Otherwise, $v\in V(G_K)$ and $w\in V(G_{L})$ for some $K,L\in\mathcal{K}_R$.
          In this case, $\mu(v)=\linebreak \langle\mu_R(K),\mu_K(v),\alpha(v)\rangle$ and $\mu(w)=\langle\mu_R(L),\mu_{L}(w),\alpha(w)\rangle$.
          If $\mu_R(K)\neq\mu_R(L)$ then $K\neq L$. In this case, $K$ (and $L$) separates $v$ and $w$ in $G^+$, so $vw\not\in E(G)$ and $A'(\mu(v),\mu(w)):=0$.
          If $\mu_R(K)=\mu_R(L)$ then, since $\mu$ is injective, $K=L$ and $v,w\in V(G_K)$.  Therefore, $vw\in E(G)$ if and only if $vw\in E(G_K)$, so the adjacency tester returns $A'(\mu(v),\mu(w)):=A'(\mu_{K}(v),\mu_K(w))$.
  \end{itemize}

  \paragraph{\boldmath The clique label $\mu(K)$:}
  What remains is to define $\mu(K)$ and $\kappa(K,u)$ for each clique $K$ of $G^\star$ and each $u\in K$.  For each clique $K$ of $G^\star$ such that $K\subseteq B_R$,
  \[
    \mu(K) := \langle \mu_R(K)\rangle \enspace .
  \]
  Then    for each clique $K$ in $G^\star$ such that $K\subseteq B_R$,
  \begin{align*}
    |\mu(K)|
     & \leq |\mu_R(K)| + \Oh(1) + \Oh(\lgg|\mu_R(K)|)                               &  & \text{by~\eqref{encoding_cost}}                              \\
     & \leq \log\omega(G) - \log\min_{v\in K}(\omega(v)) + g_3(n) + \Oh(\log\log n) &  & \textrm{by~\eqref{mu_RK_ii} and~\eqref{eq:mu_R(K)-in-logn}.}
  \end{align*}  Now, let $K$ be a clique of $G^\star$ that contains at least one vertex not in $B_R$. Then there exists exactly one $L\in\mathcal{K}_R$ such that all the vertices of $K\setminus B_R$ are contained in $G_L$. Then $K\setminus B_R$ is a clique in $G^\star_L$, since $G^\star_L[K\setminus B_R]=G^\star[K\setminus B_R]$. Thus, $\mu_L(K\setminus B_R)$ is defined. Since $K$ is a clique in $G^\star$, $K\cap B_R \subseteq L$. Since $L$ is the parent adhesion of some node in $L_2(F)$, $|L| \le k$. Let $c(K)$ be a $k$-bit string whose $i$-th bit indicates whether the $i$-th vertex of $L$ (ordered canonically, for example by their $\mu_R(v)$ labels) belongs to $K$. Define
  \[
    \mu(K):= \langle\;\mu_R(L),\;\mu_{L}(K\setminus B_R),\;c(K)\;\rangle
    \enspace .
  \]
  We now argue that $\mu$ is injective on the set of cliques of $G^\star$.  Consider two distinct cliques $K$ and $P$ of $G^\star$.   If $K,P \subseteq B_R$, then $\mu(K)=\langle\mu_R(K)\rangle\neq\langle\mu_R(P)\rangle=\mu(P)$.  If $K\subseteq B_R$ and $P\not\subseteq B_R$, then $\mu(K)$ has one part and $\mu(P)$ has three parts, so $\mu(K)\neq\mu(P)$. Now suppose that $K\not\subseteq B_R$ and $P\not\subseteq B_R$. Thus $K$ is contained in $G^\star_L$ and $P$ is contained in $G^\star_M$ for some cliques $L,M$ in $G^\star_R$.  If $L\neq M$ then $\mu_R(L) \neq \mu_R(M)$, so $\mu(K)$ differ in their first parts. Now assume that $L=M$. If $K\cap L\neq P\cap L$ then $c(K)\neq c(P)$, so $\mu(K)$ and $\mu(P)$ differ in their third parts. Otherwise, $K\setminus L\neq P\setminus L$, so $\mu_L(K\setminus B_R)$ and $\mu_L(P\setminus B_R)$ differ by induction, so $\mu(K)$ and $\mu(P)$ differ in their second parts.

  Then,
  \begin{align*}
    |\mu(K)|
     & = |\mu_R(L)| + |\mu_{L}(K\setminus B_R)| + |c(K)| + \Oh(1)+\Oh(\lgg|\mu_R(L)|+\lgg|\mu_L(K\setminus B_R)|+\lgg|c(K)|)
     & \text{by \eqref{encoding_cost}}                                                                                       \\
     & \le |\mu_R(L)| + |\mu_{L}(K\setminus B_R)| + k + \Oh(\log\log n)\enspace,
  \end{align*}
  where the last line follows by~\eqref{eq:mu_R(K)-in-logn} and by induction as $|\mu_L(K\setminus B_R)|\leq \log\omega(G_L)-\linebreak \log\min_{v\in K\setminus B_R}\omega(v) + h\cdot(g_3(n)+\Oh(k+\log\log n)) \in\Oh(\log^2 n)$, since $h\le h(n)\in\Oh(\log n)$, $k\in\Oh(\log n)$, and $g_3(n)\in\Oh(\log n)$. Continuing,
  \begin{align*}
    |\mu(K)|
     & \le |\mu_R(L)| + |\mu_{L}(K\setminus B_R)| + k + \Oh(\log\log n)                                                             \\
     & \le \log\omega(G)-\log \omega(G_{L}) + g_3(n) + |\mu_{L}(K\setminus B_R)| + k + \Oh(\log\log n)
     & \text{by~\eqref{mu_RK_ii}}                                                                                                   \\
     & \le \log\omega(G) - \log\min_{v\in K\setminus B_R}(\omega(v)) + (h+1)\cdot g_3(n) + (h+1)\cdot k + (h+1)\cdot\Oh(\log\log n) \\
     & \le \log\omega(G) - \log\min_{v\in K}(\omega(v)) + (h+1)\cdot (g_3(n) + \Oh(k+\log\log n)) \enspace .
  \end{align*}
  Thus, $\mu$ satisfies \eqref{strong_clique} and thus the requirement for $g'_3(n)$.

  \paragraph{\boldmath The local identifier $\kappa(K,u)$:}
  Let $K$ be a clique in $G^\star$, let $u\in K$, and let $z$ be the vertex of $F$ such that $u\in B_z\setminus A_z$.   The local identifier $\kappa(K,u)$ for $(A',I')$ has two parts,
  \[
    \kappa(K,u) := \langle\depth_F(z),\zeta(K,u)\rangle \enspace .
  \]
  The first part is (the binary encoding of) $\depth_F(z)$. We now describe the second part, $\zeta(K,u)$.  If $K\subseteq B_R$, then $\zeta(K,u):=\kappa_R(K,u)$.
  If $K\not\subseteq B_R$ then there is exactly one clique $L\in\mathcal{K}_R$ such that $K\setminus B_R$ is contained in $G_L$.  If $u\in B_R$ then $\zeta(K,u):=\kappa_R(L,u)$.  If $u\in K\setminus B_R$ then $K\setminus B_R$ is non-empty, so $\mu_L(K\setminus B_R)$ and  $\kappa_L(K\setminus B_R, u)$ are defined as part of the mixed labelling $(\mu_L,\kappa_L)$ of $(G^\star_L,G_L,\omega_{|V(G_L)})$ for $(A',I')$.
  Since this is a labelling for $(A',I')$ obtained by induction,
  \begin{equation}
    \kappa_L(K\setminus B_R,u)=\langle\depth_{F_L}(z),\zeta_{L}(K\setminus B_R,u)\rangle=\langle\depth_F(z)-1,\zeta_L(K\setminus B_R,u)\rangle\enspace.
    \label{off_by_one}
  \end{equation}
  In this case, define $\zeta(K,u):=\zeta_L(K\setminus B_R,u)$.

  It follows from this definition that $\zeta(K,u)$ is a local identifier $\kappa'(K',u)$, with $K'\subseteq K$ for some mixed labelling $(\mu',\kappa')$ for $(A,I)$ on a graph with at most $n$ vertices.  Therefore, $|\zeta(K,u)|\le g_2(n)$.  Therefore,
  \begin{align*}
    |\kappa(K,u)|
     & \le \ceil{\log(\depth_F(z)+1)} + |\zeta(K,u)|
    + \Oh(\lgg\ceil{\log(\depth_F(z)+1)}) + \Oh(\lgg|\zeta(K,u)|)
     & \text{by \eqref{encoding_cost}}               \\
     & \leq |\zeta(K,u)| + \Oh(\log\log n)           \\
     & \leq g_2(n) +  \Oh(\log\log n)\enspace,
  \end{align*}
  since $\depth_F(z)\le h\le h(n)\in \Oh(\log n)$ and $|\zeta(K,u)|\le g_2(n)\in\Oh(\log n)$.  Thus, the length of $\kappa(K,u)$ satisfies the requirement for $g'_2(n)$.

  \paragraph{Identity Testing:}
  We now describe the identity tester $I'$ that takes $\mu(K)$, $\kappa(K,u)$, and $\mu(v)$ as input and outputs $1$ if $u=v$ or $0$ if $u\neq v$.  To do this, $I'$ first examines $\mu(v)$ to determine if $v\in B_R$.  If $\mu(v)$ has exactly one part, then $v\in B_R$. Otherwise, $\mu(v)$ has three parts and $v\not\in B_R$.
  In both cases, $I'$ will make use of $\kappa(K,u)=\langle\depth_F(z),\zeta(K,u)\rangle$, where $z$ is the unique node of $F$ such that $u\in B_z\setminus A_z$.
  \begin{itemize}
    \item If $v\in B_R$, then $\mu(v)=\langle\mu_R(v)\rangle$ and $I'$ examines  $\depth_F(z)$ as follows.
    \begin{itemize}
      \item If $\depth_F(z)> 0$, then $z\not\in R$ and $u\in B_z\setminus A_z\subseteq B_z\setminus B_R$.  Since $v\in B_R$ and $u\not\in B_R$, $u\neq v$.  Thus, $I'(\mu(K),\kappa(K,u),\mu(v)):=0$.

      \item If $\depth_F(z)=0$ then $z\in R$ and $u\in B_R$.   There are now two possibilities that $I'$ can distinguish by looking at the number of parts of $\mu(K)$.
      \begin{itemize}
        \item If $\mu(K)$ has one part then $K\subseteq B_R$,  $\mu(K)=\langle \mu_R(K)\rangle$, and $\zeta(K,u)=\kappa_R(K,u)$.  Therefore, $I'(\mu(K),\kappa(K,u),\mu(v)):=I(\mu_R(K),\kappa_R(K,u),\mu_R(v))$.
        \item If $\mu(K)$ has three parts, then $K\not\subseteq B_R$, $\mu(K)=\langle \mu_R(L),\mu_L(K\setminus B_R),c(K)\rangle$, and  $\zeta(K,u)=\kappa_R(L,u)$, where $L$ is the unique clique in $\mathcal{K}_R$ such that $G_L$ contains $K\setminus B_R$.  Then $K\cap B_R=K\cap L$, so $u\in L$. Therefore, $I'(\mu(K),\kappa(K,u),\mu(v)):=I(\mu_R(L),\kappa_R(L,u),\mu_R(v))$.
      \end{itemize}
    \end{itemize}
    \item If $v\not\in B_R$ then $\mu(v)=\langle \mu_R(L'),\mu_{L'}(v),\alpha(v)\rangle$, where $L'$ is the unique clique in $\mathcal{K}_R$ such that $v\in V(G_{L'})$.
    \begin{itemize}
      \item If $\depth_F(z)=0$ then $z\in R$, $u\in B_R$ and $v\not\in B_R$, so $I'(\mu(K),\kappa(K,u),\mu(v)):=0$.

      \item If $\depth_F(z)>0$ then $u\not\in B_R$.  Since $u\in K$ and $u\not\in B_R$, $K\not\subseteq B_R$, so $\mu(K)=\langle \mu_R(L),\mu_{L}(K\setminus B_R),c(K)\rangle$ where $L$ is the unique clique in $\mathcal{K}_R$ such that $G_L$ contains $K\setminus B_R$.  If $\mu_R(L')\neq \mu_R(L)$ then $L\neq L'$ and $u\neq v$ since $G_{L}$ and $G_{L'}$ are vertex-disjoint, so $I'(\mu(K),\kappa(K,u),\mu(v)):=0$.

      If $\mu_R(L)=\mu_R(L')$ then $L=L'$ since $\mu_R$ is injective.
      Since $u\in K$ and $u\not\in B_R$, we have that $u\in K\setminus B_R$, so $K\setminus B_R$ is non-empty. Thus, $\zeta(K,u)=\zeta_L(K\setminus B_R,u)$
      and $I'$ can compute the local identifier $\kappa_L(K\setminus B_R,u)$ since, by \eqref{off_by_one},
      \[
        \kappa_L(K\setminus B_R,u) = \langle \depth_F(z)-1, \zeta_L(K\setminus B_R,u)\rangle\enspace.
      \]
      In this case $I'(\mu(K),\kappa(K,u),\mu(v)):=I'(\mu_L(K\setminus B_R),\kappa_L(K\setminus B_R,u),\mu_L(v))$. \qedhere
    \end{itemize}
  \end{itemize}
\end{proof}

\subsection{Finalizing the Proof of \texorpdfstring{\cref{PlainRealMainTechnical}}{Theorem \ref*{PlainRealMainTechnical}}}
\label{sec:finalizing}

The following lemma proves \cref{PlainRealMainTechnical}, with quantitative bounds.

\begin{lem}\label{RealMainTechnical}
  Let $\mathcal{G}$ be a hereditary class of graphs closed under taking disjoint union.
  Let $k\in\N$.
  Let $\mathcal{G}'$ be the class of graphs consisting of, for each $n\in\N$,  every $n$-vertex graph $G$ that has a tree-decomposition of adhesion-width at most $k$ in which each torso is a member of $\mathcal{G}$.
  Let $g_i:\N\to\R^+$ be a non-decreasing function with $g_i\in\Oh(\log n)$ for each $i\in[3]$.
  Suppose that $\mathcal{G}$ admits a weighted $(g_1,g_2,g_3)$ mixed labelling scheme.
  Then $\mathcal{G}'$ admits a weighted $(g_1'',g_2'',g_3'')$ mixed labelling scheme, where
  \begin{align*}
    g''_1(n)
     & = g_1(n) + \Oh(g_2(n)) + \Oh\left(\sqrt{g_3(n)\log n}\right) + \Oh\left(\sqrt{\frac{\log n}{g_3(n)}}\log\log n\right) \enspace , \\
    g''_2(n)
     & = g_2(n) + \Oh\left(\sqrt{g_3(n)\log n}\right) + \Oh(\log\log n)  \enspace ,                                                     \\
    g''_3(n)
     & = \Oh\left(\sqrt{g_3(n)\log n}\right) + \Oh\left(\sqrt{\frac{\log n}{g_3(n)}}\log\log n\right) \enspace .
  \end{align*}
\end{lem}

\begin{proof}
  Let $b,h:\N\to\R^+$ be functions defined by
  \begin{align*}
    b(n)
     & :=2^{\sqrt{g_3(n)\log n}}\enspace,                        \\
    h(n)
     & := \log_{b(n)}n = \sqrt{\frac{\log n}{g_3(n)}} \enspace .
  \end{align*}
  Since $g''_1(n)$, $g''_2(n)$, and $g''_3(n)$ all include asymptotic terms that increase without bound as a function of $n$, we may assume that $n\ge 2$ and $g_3(2)\ge 1$.  Therefore, $b(n)\ge 2$.

  Let $\mathcal{B}$ be a class of graphs that contains, for each $n\in\N^+$, every $n$-vertex graph that has a rooted $b(n)$-skinny forest-decomposition of adhesion-width at most $k$ in which each torso is a member of $\mathcal{G}$.
  By \cref{skinny_labelling}, each graph in $\mathcal{B}$ admits a
  weighted $(g'_1(n),g'_2(n),g'_3(n))$ mixed labelling scheme, where
  \begin{align*}
    g'_1(n)
     & = g_1(n) + \Oh(k \log b(n) + k \log\log n)                                 \\
     & = g_1(n) + \Oh\left(\sqrt{g_3(n) \log n}\right) + \Oh(\log\log n)\enspace, \\
    g'_2(n)
     & = g_2(n)+ \Oh(\log b(n) + \log\log n)                                      \\
     & = g_2(n) + \Oh\left(\sqrt{g_3(n) \log n}\right) + \Oh(\log\log n)\enspace, \\
    g'_3(n)
     & = g_3(n) + \Oh(k + \log\log n)                                             \\
     & = g_3(n) + \Oh(\log\log n)\enspace.
  \end{align*}

  Let $\mathcal{C}$ be a class of graphs that contains, for each $n\in\N^+$,
  every $n$-vertex graph that has a rooted forest-decomposition $(F,(C_x)_{x\in V(F)})$
  of adhesion-width at most $k$, with $\height(F)\leq h(n)$, and
  in which each torso is a member of $\mathcal{B}$.
  By \cref{small_height}, $\mathcal{C}$ has a weighted \linebreak $(g''_1(n),g''_2(n),g''_3(n))$ mixed labelling scheme, where
  \begin{align*}
    g''_1(n)
     & = g'_1(n) + k\cdot(g'_2(n) + \Oh(\log\log n)) + h(n)\cdot (g'_3(n)+\Oh(\log\log n)) + \Oh(\log\log n) \\
     & = g_1(n) +  \Oh\left(\sqrt{g_3(n) \log n}\right) + \Oh(g_2(n)) + \Oh\left((1+h(n))\log\log n\right)   \\
    g''_2(n)
     & = g'_2(n) + \Oh(\log\log n)                                                                           \\
     & = g_2(n) + \Oh\left(\sqrt{{g_3(n)\log n}}\right) + \Oh(\log\log n)                                    \\
    g''_3(n)
     & = (h(n)+1)\cdot (g'_3(n)+\Oh(k+\log\log n))                                                           \\
     & = \Oh\left(\sqrt{g_3(n)\log n}\right) + g_3(n) + \Oh\left((1+h(n))\log\log n\right)                   \\
     & = \Oh\left(\sqrt{g_3(n) \log n}\right) + \Oh\left((1+h(n))\log\log n\right)\enspace.
  \end{align*}
  Substituting $h(n) = \sqrt{\log n\, / g_3(n)}$ into the error terms yields the bounds in the lemma statement.

  Let $G$ be an $n$-vertex graph in $\mathcal{G}'$.
  Then $G$ has a tree-decomposition $\mathcal{T}:=(T,(B_x\mid x\in V(T)))$ of adhesion-width at most $k$ and such that $\torso{G}{\mathcal{T},B_x}\in\mathcal{G}$ for each $x\in V(T)$.
  By~\cref{b_skinny_decomp}, there is a tree-decomposition $\mathcal{Q}:=(Q,(D_x \mid x\in V(Q)))$ of $G$ such that $\height(Q)\le\log_{b(n)} n = h(n)$, $\mathcal{Q}$ has adhesion-width at most $k$, and such that each torso of $\mathcal{Q}$ has a $b(n)$-skinny tree-decomposition whose torsos are in $\mathcal{G}$.
  Therefore, all the torsos of $\mathcal{Q}$ are in the class $\mathcal{B}$ and therefore $G$ is in the class $\mathcal{C}$.  This completes the proof of the lemma.
\end{proof}

\subsection{Proof of \texorpdfstring{\cref{main_result}}{Theorem \ref*{main_result}} with Explicit Bounds}

We now prove our main result, \cref{main_result}, with an explicit lower order term.

\begin{lem}
  \label{main_result_detailed}
  Every proper minor-closed graph class $\mathcal{C}$  admits a $\log n + \Oh\big((\log n)^{3/4}\big)$-bit adjacency labelling scheme.
\end{lem}

\begin{proof}
  We show that $\mathcal{C}$ admits a weighted $(g_1,g_2,g_3)$ mixed labelling scheme with $g_1(n) \in \Oh((\log n)^{3/4})$. By the definition of mixed labelling schemes, this implies that $\mathcal{C}$ admits an $f(n)$-bit adjacency labelling scheme with $f(n)\le\log n + \Oh((\log n)^{3/4})$.

  Let $k,a\geq 0$ be integers given by \cref{gmpst} for the class $\mathcal{C}$. So $\mathcal{C}\subseteq \mathcal{R}_k^{+a}$. By \cref{add_apexes} and \cref{gk_labelling_precise}, $\mathcal{R}_k^{+a}$ admits a weighted $(g'_1,g'_2,g'_3)$ mixed labelling scheme, where $g'_1(n),g'_3(n)\in \Oh(\sqrt{\log n})$ and $g'_2(n)\in \Oh(\log\log n)$.

  Let $\mathcal{Q}$ consist of every graph that can be formed by the disjoint union of a finite number of graphs in $\mathcal{R}_k^{+a}$.  By \cref{disjoint_union_ii}, $\mathcal{Q}$ admits a weighted $(g''_1,g_2'',g_3'')$ mixed labelling scheme with $g''_1(n),g''_3(n)\in \Oh(\sqrt{\log n})$ and $g''_2(n)\in \Oh(\log\log n)$.   By definition, $\mathcal{Q}$ is closed under disjoint union.  Also, by definition  $\mathcal{R}_k^{+a}\subseteq\mathcal{Q}$. Since $\mathcal{R}_k^{+a}$ is hereditary, $\mathcal{Q}$ is hereditary.

  Let $\mathcal{Q}'$ be the class of graphs that have a tree-decomposition whose torsos are in $\mathcal{Q}$.  The graphs in $\mathcal{R}_k^{+a}$ have maximum clique size at most $2k+a+2$.  So the graphs in $\mathcal{Q}$ also have maximum clique size at most $2k+a+2$. Since each adhesion of a tree-decomposition is a clique in some torso,   each graph in $\mathcal{Q}'$ has a tree-decomposition 
  of adhesion-width at most $2k+a+2$ whose torsos are in $\mathcal{Q}$.

  By \cref{RealMainTechnical}, each graph in $\mathcal{Q}'$ admits a weighted $(g_1,g_2,g_3)$ mixed labelling scheme with
  \[
    g_1(n) \in g''_1(n) + \Oh\left(g''_2(n) + \sqrt{g''_3(n)\log n}
    + \sqrt{\frac{\log n}{g''_3(n)}}\log\log n\right) = \Oh((\log n)^{3/4}) \enspace.
  \]
  Since $\mathcal{R}_k^{+a}\subseteq\mathcal{Q}$, \cref{gmpst} implies that $\mathcal{C}\subseteq\mathcal{Q}'$.  Therefore, $\mathcal{C}$ admits a weighted $(g_1,g_2,g_3)$ mixed labelling scheme with $g_1(n)\in \Oh((\log n)^{3/4})$.
\end{proof}

\section{Discussion}\label{conclusion}

We conclude with a brief discussion about constants and computational complexity.  The constant $a$ in \cref{gmpst} is surprisingly small. Suppose that some graph $X$ is not in $\mathcal{G}$, where $X-A$ is planar for some non-empty set $A\subseteq V(X)$. Then Theorem~42 in \citep{dujmovic.joret.ea:planar} (which employs the structure theorem of \citet{DvoTho}) shows that \cref{gmpst} holds with $a=|A|-1$. For example, if $K_t$ is not in $\mathcal{G}$, then \cref{gmpst} holds with $a=\max\{t-5,0\}$. Moreover, using the recent polynomial bounds in the Graph Minor Structure Theorem by \citet{gorsky2025polynomial}, in \cref{gmpst} one can obtain polynomial bounds on $k$ and $a$ (as a function of the excluded minor). However, it is open whether one can simultaneously get the above optimal bound on $a$ and a polynomial bound on $k$. \citet[Conjecture~18.9]{gorsky2025polynomial} conjecture this is possible (also see the discussion in \citep{GSW26}).

We now explain why the proof of \cref{main_result} is constructive and that there is a polynomial time encoder.  The only difficult part of this process is finding the decomposition described by the Graph Minor Product Structure Theorem (\cref{gmpst}).  \citet{dujmovic.joret.ea:planar} show how to obtain this decomposition in polynomial time, given the structural decomposition in the original Robertson-Seymour Graph Minor Structure Theorem. Finding the Robertson-Seymour decomposition is complicated, but polynomial time algorithms exist \cite{kawarabayashi2018flat,kawarabayashi2020new,gorsky2025polynomial,demaine2005algorithmic}.  Together, these give the parameters $k$ and $a$, the tree-decomposition $\mathcal{T}:=(T, (B_x:x\in V(T)))$ of $G$ and the embedding of each torso of $\mathcal{T}$ into a graph in $\mathcal{R}_k^{+a}$. The adjacency labelling scheme in \cite{dujmovic.esperet.ea:adjacency} and its modification to weighted mixed labellings in \cref{weighted_product_proof} are easily implemented in polynomial time just from their definitions.  Thus, for any torso of $\mathcal{T}$, the mixed labelling of any subgraph of the torso and any set of cliques in the torso can be computed in polynomial time.  The partition of $T$ into skinny subtrees and the corresponding contracted tree $T'$ and tree-decomposition $\mathcal{T}'$ are easily obtained (even in linear time).  The labelling of each torso of $\mathcal{T}'$ given in \cref{skinny_labelling} requires only computing  alphabetic codes (\cref{shannon_fano_elias_tree,shannon_fano_elias_code}) as well as weighted mixed labellings on subgraphs of torsos of $\mathcal{T}$.  The final vertex label $\mu(v)$ for a vertex of $G$ is then obtained by following the path from the root of $T'$ to the home of $v$ in $\mathcal{T}'$ and concatenating a sequence of clique labels, the vertex label for $v$ in its home torso, and at most $k$ local identifiers.

Finally, we note that \cref{PlainRealMainTechnical} applies beyond proper
minor-closed classes. Recall that, for fixed $k$, $\mathcal{R}_k$ consists of
the graphs isomorphic to subgraphs of $H\boxtimes P$ for some graph $H$ of
treewidth at most $k$ and some path $P$. This class is hereditary, closed
under taking disjoint union, and admits an efficient weighted mixed labelling
scheme (\cref{gk_labelling_precise}), so \cref{PlainRealMainTechnical} applies
to it. Moreover, for suitable $k=k(p)$, $\mathcal{R}_k$ contains every
$p$-planar graph (graphs that can be drawn in the plane with at most $p$
crossings per edge)~\citep{dujmovic.morin.ea:graph}, and more generally every
$p$-matching-planar graph~\citep{DBLP:journals/corr/abs-2507-22395} (graphs that have a drawing in the plane such that, for every edge $e$, every
matching among the edges crossing $e$ has size at most $p$). Note that these classes are not minor-closed. Consequently, for all fixed $p$ and $q$, the class of graphs admitting a tree-decomposition of adhesion-width at most $q$ in which every
torso is $p$-matching-planar (again a class that is not minor-closed)
 admits a $(1+o(1))\log n$-bit adjacency labelling scheme.

\section*{Acknowledgement} 
This work was initiated during the Ninth Annual Workshop on Geometry and Graphs (v2), held January 21--28, 2022, at the Bellairs Research Institute of McGill University.  The authors are thankful to the other organizers and participants for providing a stimulating research environment.

\setlength{\bibsep}{0.3ex plus 0.1ex minus 0.1ex}
\bibliographystyle{plainurlnat}
\bibliography{mflabelling}

\appendix
\crefalias{section}{appendix}

\section{Mixed Labelling Scheme for \texorpdfstring{$\mathcal{R}_k$}{Rk}}
\label{weighted_product_proof}

In this appendix, we modify the labelling scheme of \citet{dujmovic.esperet.ea:adjacency} and its improved version by \citet{gawrychowski.janczewski:simpler} in order to show the existence of the efficient mixed labelling scheme for $\mathcal{R}_k$ mentioned in \cref{KeyTools}.  We prove the following result.

\begin{lem}\label{gk_labelling_precise}
  For each fixed integer $k\ge 0$, the graph class $\mathcal{R}_{k}$ admits an efficient weighted $(g_1,g_2,g_3)$ mixed labelling scheme, where $g_1(n),g_3(n)\in \Oh(\sqrt{\log n})$, and $g_2(n)\in \Oh(1)$.
\end{lem}

\begin{proof}
  We begin by reviewing the proofs in \cite{dujmovic.esperet.ea:adjacency} and \cite{gawrychowski.janczewski:simpler} and then describe the changes needed to obtain the statement of \cref{gk_labelling_precise}.

  Let $G$ be an $n$-vertex graph in $\mathcal{R}_k$, so $G$ is a subgraph of $H\boxtimes P$, where $H$ is a graph of treewidth $k$ and $P$ is a path. Without loss of generality, assume that $P=(1,2,\dots,h)$, with $h\le n$, and $H$ is an edge-maximal graph of treewidth $k$.  This implies that $H$ has an acyclic orientation in which the out-neighbourhood \defin{$N^+_H(v)$} is a clique in $H$, for each $v\in V(H)$. It follows that the out-degree $|N^+_H(v)|$ is at most $k$, for each vertex $v\in V(H)$.

  \paragraph{\boldmath Mapping $V(H)$ onto intervals:}
  The first step in \cite{dujmovic.esperet.ea:adjacency} is an idea used by \citet{gavoille.labourel:shorter} in their labelling scheme for bounded-treewidth graphs.
  Each vertex $v$ of $H$ is mapped to a real interval \defin{$f(v)$}.
  This mapping has the following properties (for some fixed $c>0$):
  \begin{enumerate}[nosep,label=(I\arabic*)]
    \item\label{nested_or_disjoint} For any two vertices $v,w\in V(H)$, $f(v)\subseteq f(w)$, $f(w)\subseteq f(v)$ or $f(v)\cap f(w)=\emptyset$.
    \item\label{intersection_subgraph} For each $vw\in E(H)$, $f(v)\cap f(w)\neq\emptyset$.
    \item\label{thin} For each $x\in \R$, $|\{v\in V(H):x\in f(v)\}|\le ck\log n$.
  \end{enumerate}

  The existence of such a mapping $f$ is easily established by a recursive procedure that takes an open interval $(a,b)$ and a subgraph $H'$ of $H$.
  The procedure identifies a set $S$ of at most $k+1$ vertices in $H'$ such that the components of $H'-S$ can be partitioned into two graphs $H'_1$ and $H'_2$, each with at most $\tfrac{2}{3}|V(H')|$ vertices.
  For each $v\in S$, $f(v):=(a,b)$.  The procedure is then applied inductively to map the vertices of $H'_1$ onto subintervals of $(a,(a+b)/2)$ and then again to map the vertices of $H'_2$ onto subintervals of $((a+b)/2,b)$.

  Let $I$ be the interval intersection graph with vertex-set $V(I):=V(H)$ and edge-set $E(I):=\{vw\in\binom{V(H)}{2}:f(v)\cap f(w)\neq\emptyset\}$.
  Property~\ref{intersection_subgraph} implies that $H$ is a subgraph of $I$. Property~\ref{thin} implies that $H$ has no clique larger than $ck\log n$, for some fixed constant $c$.
  The graph $I$ has a proper colouring $\mathdefin{\varphi}:V(H)\to[\lfloor ck\log n\rfloor]$ \cite{gilmore.hoffman:characterization}.
  The label for each vertex $(v,i)$ of $G$ will include the $\Oh(\log (k\log n))$-bit integer $\varphi(v)$.

  \paragraph{\boldmath Mapping $V(H)$ onto points:}
  For each vertex $v$ of $H$, let $x_f(v)$ be the midpoint of the interval $f(v)$.\footnote{Any real number $x_f(v)$ in the interior of the interval $f(v)$ would also work.}  Property~\ref{nested_or_disjoint} then implies:
  \begin{enumerate}[nosep,label=(I\arabic*)]\setcounter{enumi}{3}
    \item\label{clique_property} For each clique $K$ in $I$, there exists $v\in K$ such that $x_f(v)\in f(w)$ for each $w\in K$.
  \end{enumerate}
  Indeed, a vertex $w$ of $K$ that minimizes the length of $f(w)$ satisfies this condition.

  \paragraph{\boldmath Treating $G$ as a sequence of graphs:}
  Each vertex of $G$ is a pair $(v,i)$ with an \defin{$H$-coordinate} $v\in V(H)$ and a \defin{$P$-coordinate} $i\in V(P)$.
  We may assume, without loss of generality, that $G$ contains at least one vertex with $P$-coordinate $i$, for each $i\in[h]$ (and therefore $h\le n$).
  The next idea in \cite{dujmovic.esperet.ea:adjacency} is to treat $V(G)$ as a sequence of vertex subsets of $H$.
  For each $i\in[h]$, let $V_i:=\{v\in V(H): \{(v,i),(v,i-1)\}\cap V(G)\neq\emptyset \}$.    Each vertex $(v,i)$ of $G$ contributes a vertex to the two sets $V_{i}$ and $V_{i+1}$, so $\sum_{i=1}^h |V_i|\le 2n$.  For each vertex $(v,i)$ of $G$, define the set $X_{v,i}:= \{v\}\cup (V_i\cap N^+_H(v))$.

  \begin{obs}\label{x_v_i_sufficient}
    For any edge $(v,i)(w,j)$ of $G$ with $w\in N^+_H(v)$, we have that $w\in X_{v,i}$ (if $j\in\{i-1,i\}$) or $w\in X_{v,i+1}$ (if $j\in\{i,i+1\}$).
  \end{obs}

  The label of any vertex $(v,i)$ of $G$ has two main parts, whose lengths sum to $\log n+o(\log n)$.
  The first part, $\lambda_1((v,i))$ is determined by the $P$-coordinate, $i$.
  The second part, $\lambda_2((v,i))$, is obtained from a labelling of $H[V_i]$, so it represents the $H$-coordinate, $v$.

  Let $S_0:=S_{h+1}:=\emptyset$ and,
  for each $i\in[h]$, let $\mathdefin{S_i}:=\{x_f(v):v\in V_i\}$.
  Note that $S_i$ is a set of real numbers, for each $i\in[h]$.
  It is helpful to impose a relationship between consecutive sequences $S_i$ and $S_{i+1}$, for each $i\in[h-1]$. Therefore, \cite{dujmovic.esperet.ea:adjacency} defines a sequence of supersets \defin{$S^+_1,\ldots,S^+_h$}. The exact definition of $S^+_1,\ldots,S^+_h$ is not necessary for the current discussion.
  The relevant properties here are:
  \begin{enumerate}[nosep,label=(\alph*)]
    \item\label{plus_superset} $S^+_i\supseteq S_i$ for each $i\in[h]$;
    \item\label{plus_size} there is a constant $c>0$ such that $\sum_{i=1}^h |S^+_i|\le cn$.
    \item\label{plus_ratio} there is a constant $c>0$ such that
          $1/c\le |S^+_i|/|S^+_{i+1}|\le c$, for each $i\in[h-1]$.
  \end{enumerate}
  We remark that the contents of the sets $S^+_1,\ldots,S^+_h$ are entirely determined by the contents of the sets $S_1,\ldots,S_h$.
  (Formally, the sequence $S^+_1,\ldots,S^+_h$ is the result of applying a function $q:(2^\R)^*\to (2^\R)^*$ whose input is a sequence of sets of real numbers and whose output is a sequence of sets of real numbers. The function $q$ is universal; it does not depend on $H$, $P$, or $G$.)

  \paragraph{\boldmath The row label $\lambda_1$:}
  With the sets $S^+_1,\ldots,S^+_h$ described, we can now discuss $\lambda_1$.
  The first part of $\lambda_1(v,i)$ is obtained by applying \cref{shannon_fano_elias_code} to the set $[h]$ using the weight function defined by $i \mapsto |S^+_i|$. This gives a prefix-free code $\mathdefin{\rho}:[h]\to\{0,1\}^*$ where, by Property~\ref{plus_size} of \cref{shannon_fano_elias_code}, $|\rho(i)|\le \log\rho([h]) - \log|S^+_i| + \Oh(1)\le \log n - \log|S^+_i| + \Oh(1)$ for each $i\in[h]$.
  For each $i\in\{2,\ldots,h\}$, a string $\mathdefin{\pi(i)}$ of length $\Oh(\log\log n)$ is computed so that $\rho(i-1)$ can be derived from $\rho(i)$ and $\pi(i)$.  Formally, there exists a universal function $P:(\{0,1\}^*)^2\to\{0,1\}^*$ such that $P(\rho(i),\pi(i))=\rho(i-1)$, for each $i\in\{2,\ldots,h\}$.
  Then
  \[
    \mathdefin{\lambda_1(v,i)}:=\langle \rho(i),\pi(i)\rangle
  \]
  and
  \begin{equation}
    |\lambda_1(v,i)|\le \log n - \log |S^+_i| + \Oh(\log\log n) \enspace .
    \label{lambda_one_length}
  \end{equation}
  for each vertex $(v,i)$ of $G$.

  \paragraph{Bulk trees:}
  In \cite{dujmovic.esperet.ea:adjacency}, each set $S^+_i$ is stored in a binary search tree $T_i$ of height $\log|S^+_i|+o(\log n)$.
  As in \cref{shannon_fano_elias_code}, each node $x$ of $T_i$ is assigned a \defin{signature} \defin{$\sigma_{T_i}(x)$} that describes the root-to-$x$ path $P_{T_i}(x)$ as a binary string of length $\depth_{T_i}(x)$.
  Note that, for any $x\in V(T_i)$:
  \[
    |\sigma_{T_i}(x)| \le \height(T_i) \le \log|S^+_i|+o(\log n)
  \]
  For each $i\in[h]$, each vertex $v\in V_i$ is mapped to the node $\mathdefin{x_{T_i}(v)}\in V(T_i)$ of minimum $T_i$-depth such that (the real number) $x_{T_i}(v)$ is contained in the (real interval) $f(v)$.
  For each $v\in V_i$, the node $x_{T_i}(v)$ always exists because $x_f(v)\in S_i\subseteq S_i^+$ and $x_f(v)\in f(v)$.  The same reasoning implies that $x_{T_i}(w)$ is defined, for each $w\in X_{v,i}$.

  For each $v\in V_i$, define the \defin{signature} $\sigma(v,i):=\sigma_{T_i}(x_{T_i}(v))$.  Note that, for two vertices $(v,i)$ and $(v',i')$ of $G$, $(v,i)=(v',i')$ if and only if
  \begin{equation}
    (i,\,\sigma(v,i),\,\varphi(v)) = (i',\,\sigma(v',i'),\,\varphi(v')) \enspace .
    \label{ps_identity}
  \end{equation}

  For each $(v,i)\in V(G)$, the \defin{extended signature} \defin{$\sigma^+(v,i)$} is the signature in $\mathdefin{\Sigma_{v,i}}:=\{\sigma(w,i):w\in X_{v,i}\}$ of maximum length.  Since $X_{v,i}$ is a clique in $H$, property~\ref{clique_property} implies that every signature in $\Sigma_{v,i}$ is a prefix of $\sigma^+(v,i)$ for each vertex $(v,i)$ of $G$.  In particular, for each $w\in X_{v,i}$, the signature $\sigma(w,i)$ can be derived easily from $\sigma^+(v,i)$ and $|\sigma(w,i)|$.

  For each $(v,i)\in V(G)$, define
  \begin{align*}
    \mathdefin{\alpha^{-}(v,i)}
     & = \langle\; \langle |\sigma(w,i)|, \varphi(w)\rangle \mid w\in \{v\}\cup N^+_H(v),\; (w,i-1)\in N_G((v,i)) \;\rangle   \\
    \mathdefin{\alpha^{0}(v,i)}
     & = \langle\; \langle |\sigma(w,i)|, \varphi(w)\rangle \mid w\in N^+_H(v),\; (w,i)\in N_G((v,i)) \;\rangle               \\
    \mathdefin{\alpha^+(v,i)}
     & = \langle\; \langle |\sigma(w,i+1)|, \varphi(w)\rangle \mid w\in \{v\}\cup N^+_H(v),\; (w,i+1)\in N_G((v,i)) \;\rangle \\
    \mathdefin{\alpha(v,i)}
     & = \langle\; \alpha^-(v,i),\alpha^0(v,i),\alpha^+(v,i) \;\rangle        .
  \end{align*}
  Since $|N^+_G(v)|\le k$ for each $v\in V(H)$, $|\alpha(v,i)|\in\Oh(k(\log k + \log\log n))$.

  The last (and critical) piece of the labelling scheme in \cite{dujmovic.esperet.ea:adjacency} is the notion of a \defin{transition code}. The authors of \cite{dujmovic.esperet.ea:adjacency} choose the sets $S^+_1,\ldots,S^+_h$ and design the binary search trees $T_1,\ldots,T_h$ so that, for any vertex $(v,i)$ of $G$, the extended signature $\sigma^+(v,i+1)$ can be derived from $\sigma^+(v,i)$ and a transition code \defin{$\tau((v,i))$} of length $o(\log n)$. Formally, there is a universal function $J:(\{0,1\}^*)^2\to\{0,1\}^*$ such that $J(\sigma^+(v,i),\tau(v,i))=\sigma^+(v,i+1)$ for each $(v,i)\in V(G)$.

  This gives all the pieces needed for the second part of the label for a vertex $(v,i)$ of $G$:
  \[
    \mathdefin{\lambda_2(v,i)} := \langle \sigma^+(v,i),\; |\sigma(v,i)|,\; \varphi(v),\; \alpha(v,i),\; \tau(v,i),\; |\sigma(v,i+1)| \rangle  \enspace .
  \]
  The complete vertex label for a vertex $(v,i)$ of $G$ is
  \[
    \mathdefin{\mu((v,i))} := \langle \lambda_1(v,i),\; \lambda_2(v,i) \rangle
  \]
  which has length $\log n + O(\sqrt{\log n\log\log n})$, for fixed $k$.  (For $k$ that grows with $n$, $|\mu(v,i))|=\log n + \Oh(\sqrt{\log n\log\log n} + k(\log k + \log\log n))$.)

  \paragraph{Adjacency Testing}

  We now describe a procedure that defines the output of the adjacency tester $A$.
  The operation of this procedure makes use of the fact that $\lambda_2(v,i)$ contains all the information in $\sigma^+(v,i)$, $\sigma(v,i)$, $\sigma^+(v,i+1)=J(\sigma^+(v,i),\tau(v,i))$, and $\sigma(v,i+1)$.

  Given the labels $\mu((v,i))$ and $\mu((w,j))$ for two vertices $(v,i)$ and $(w,j)$ of $G$, the procedure first uses $\lambda_1(v,i)$ and $\lambda_1(w,j)$ (along with the function $P$) to test if $i=j$, $i=j+1$ or $i=j-1$.  If none of these is the case, then the procedure can immediately conclude that $(v,i)(w,j)$ is not an edge of $H$, since $i\neq j$ and $ij$ is not an edge of $P$.  Since $G$ is a subgraph of $H\boxtimes P$, this implies that $A(\mu((v,i)),\mu((w,j))):=0$.  Otherwise, the procedure deals with one of the preceding cases:
  \begin{itemize}
    \item If $i=j$, then the procedure computes $\sigma(w,i)$ using $\sigma^+(w,i)$ and $|\sigma(w,i)|$.  Then the procedure uses $\sigma^+(v,i)$ and $\alpha^0(v,i)$ to compute $\sigma(w',i)$ and $\varphi(w')$ for each $w'\in N^+_H(v)$ such that $(w',i)\in N_G((v,i))$.  For each of these the procedure checks if $(\sigma(w',i),\varphi(w'))=(\sigma(w,i),\varphi(w))$. If so, then $(w,i)=(w',i)\in N_G((v,i))$ so $A(\mu((v,i)),\mu((w,j))):=1$.

          If this test fails for each of the entries in the second part of $\alpha(v,i)$, then the procedure concludes that $w\not\in X_{v,i}$ or $(w,i)\not\in N_G((v,i))$.  This still leaves the possibility that $v\in X_{w,i}$ and $(v,i)\in N_G((w,j))$.  The procedure checks this in exactly the same manner, but with the roles of $\mu(v,i)$ and $\mu(w,j)$ reversed.  If this test also fails to return a result, then $A(\mu((v,i)),\mu((w,j))):=0$.

    \item If $i=j+1$ then the procedure computes $\sigma^+(w,i)=J(\sigma^+(w,j),\tau(w,j))$. Then the procedure computes $\sigma(w,i)$ using $\sigma^+(w,i)$ and $|\sigma(w,i)|$. Then the procedure continues as in the previous case, except now using entries from $\alpha^-(v,i)$ to find $(w',i)$ such that $(\sigma(w',i),\varphi(w'))=(\sigma(w,i),\varphi(w))$.

          If this search fails, then the procedure continues as in the previous case, except now using entries from $\alpha^+(w,j)$ to find a $(v',i)$ such that $(\sigma(v',i),\varphi(v'))=(\sigma(v,i),\varphi(v))$.  If this search also fails, then $A(\mu((v,i)),\mu((w,j))):=0$.

    \item If $i=j-1$ then the procedure computes $\sigma^+(v,j)=J(\sigma^+(v,i),\tau(v,i))$. Then the procedure computes $\sigma(w,j)$ using $\sigma^+(w,j)$ and $|\sigma(w,j)|$. Then the procedure continues as in the previous cases, except now using entries from $\alpha^+(v,i)$ to find $(w',j)$ such that $(\sigma(w',j),\varphi(w'))=(\sigma(w,j),\varphi(w))$.

          If this search fails, then the procedure continues by using entries in $\alpha^+(w,j)$ to find $(v',j)$ such that $(\sigma(v',j),\varphi(v'))=(\sigma(v,j),\varphi(v))$.  If this search also fails, then the procedure concludes that $A(\mu((v,i)),\mu((w,j))):=0$.
  \end{itemize}

  \paragraph{B-trees:}
  \citet{gawrychowski.janczewski:simpler} use
  weight-balanced B-trees to replace the (binary) bulk trees used in \cite{dujmovic.esperet.ea:adjacency}.
  A \defin{B-tree} stores a set of real numbers in the leaves of a tree $B$ in which all leaves have the same depth.
  \defin{Weight-balanced B-trees} are parameterized by a parameter $a\ge 6$.  A weight-balanced B-tree $B$ is \defin{semi-balanced} if, for each node $x$ of $B$, the number of leaves in the subtree $B_x$ is at most $6a^{\height(B_x)}$ and for each \emph{non-root} node $x$ of $B$, the number of leaves in $B_x$ is at least $\tfrac{1}{2}a^{\height(B_x)}$.
  Let $B$ be a weight-balanced B-tree with $r\in [n]$ leaves.  The weight-balance conditions guarantee that $\height(B)\le \log_a r + 1$.
  The weight-balance conditions also ensure that the number of children of any node $x$ of $B$ is at most $12a$.
  This implies that any root-to-leaf path in a B-tree can be encoded by a sequence of $\log_a r+1$ integers, each requiring $\lceil\log(12a)\rceil\le 5+\log a$ bits to represent, for a total of at most $\lceil\log(12a)\rceil\cdot(\log_a r+1)=\log r + \Oh(\log_a r + \log a)$ bits.
  The value of $a$ used in \cite{gawrychowski.janczewski:simpler} is approximately $2^{\sqrt{\log n}}$, so that any root-to-leaf path in $B$ can be represented using $\log r + \Oh(\sqrt{\log n})$ bits, for any $r\in\Oh(n)$.

  The sequence of binary search trees $T_1,\ldots,T_h$ used in \cite{gawrychowski.janczewski:simpler} is replaced with a sequence of B-trees $B_1,\ldots,B_h$.
  For each $i\in[h]$, the leaves of $B_i$ store the values in $S^+_i$.
  For each $i\in[h]$, each vertex $(v,j)\in V_i\times\{i-1,i\}$ is mapped to the lowest-common-$B_i$-ancestor $x_{B_i}(v)$ of $f(v)\cap S^+_i$.
  Then, for a vertex $(v,i)$ of $G$, the signature \defin{$\sigma(v,i)$} is the encoding of $P_{B_i}(x_{B_i}(v))$.
  Then
  \[
    |\sigma(v,i)| \le \lceil\log (12a)\rceil\cdot(\depth_{B_i}(x_{B_i}(v))) = \lceil\log (12a)\rceil(\depth_{B_i}(\lca_{B_i}(f(v)\cap S_i^+))) \enspace .
  \]
  An upper bound on $|\sigma((v,i))|$ can be obtained by using the $\log_a |S^+_i|+1$ upper bound on $\height(B_i)$:
  \[
    |\sigma(v,i)|\le \lceil\log (12a)\rceil\height(B_i) \le \lceil\log (12a)\rceil(\log_a |S^+_i| + 1)= \log |S^+_i| + \Oh(\sqrt{\log n})
  \]

  The main advantage of using B-trees over the bulk trees in \cite{dujmovic.esperet.ea:adjacency} is that in B-trees the rebalancing operations needed to maintain the height bound result in a much simpler transition code $\tau(w,i-1)$.
  This simplicity ultimately comes from the fact that, for any $p\in S_{i-1}\cap S_i$, the root-to-$p$ path $P_{B_{i-1}}(p)$ in $B_{i-1}$ and the root-to-$p$ path $P_{B_i}(p)$ in $B_i$ differ by at most a single vertex.
  More precisely, a subpath $xyz$ in $P_{B_{i-1}}(p)$ can be replaced by a subpath $xy'z$ to give $P_{B_i}(x)$.%
  \footnote{There are two exceptions. Occasionally, $P_{B_i}(x)$ is obtained by removing the first vertex of $P_{B_{i-1}}(x)$ or by adding a vertex at the beginning of $P_{B_{i-1}}(x)$. Since these are distracting and easily handled, we ignore them here.}
  In terms of the signature function $\sigma$, this means that the difference between these two paths can be described by an $\Oh(\log\log_a n)$-bit integer that gives the index of the vertex $y$ being replaced and two $\lceil\log (12a)\rceil$ bit integers that give the index of $y'$ in the list of children of $x$ and the index of $z$ in the list of children of $y'$.\%
  \footnote{Determining $\sigma^+((w,i))$ from $\sigma^+((w,i-1))$ requires a bit more work, since it is possible that $x_{B_i}(w)\neq x_{B_{i-1}}(w)$. This is common to both labelling schemes and does not affect the upcoming discussion. The important thing for what comes next is the definition of $x_{B_i}(w)=\lca_{B_i}(f(w)\cap S^+_i)$.}
    As a result of this, the transition code $\tau((w,i-1))$ used for the B-trees in \cite{gawrychowski.janczewski:simpler} is shorter (it has length $\Oh(\sqrt{\log n})$ rather than the $\Oh(\sqrt{\log n\log\log n})$ required by bulk trees) and, on a computer with $\Oh(\log n)$-bit words, the translation from $\sigma^+((w,i-1))$ to $\sigma^+((w,i))$ can be done in constant time, resulting in a constant-time adjacency testing procedure. (All of the other adjacency-testing operations for both schemes can be done in constant-time.)

    Except for the new definition of the basic signature $\sigma$ and the transition function $\tau$, every other aspect of the resulting labelling scheme and adjacency testing procedure is as described in \cite{dujmovic.esperet.ea:adjacency}.

    \paragraph{A weighted mixed labelling scheme:}
    We now describe the changes to \cite{gawrychowski.janczewski:simpler,dujmovic.esperet.ea:adjacency} needed to prove \cref{gk_labelling_precise}.  For this we will define an adjacency tester $A$ and an identity tester $I$. As usual, we first describe the format of the vertex and clique labels before describing the operation of $A$ and $I$.

    Let $H$ be an edge-maximal graph of treewidth at most $k$, let $P=(1,\ldots,h)$ be a path, let $G^+$ be a subgraph of $H\boxtimes P$, let $G$ be a spanning subgraph of $G^+$ and let $\omega:V(G)\to\R^+$ be a weight function.  As in other proofs, we may assume that $\omega$ is nice, so that $\omega(v)$ is a positive integer, for each $v\in V(G)$, and that $\log\omega(G)\in\Oh(\log n)$.  We may also assume that $G^+$ is an induced subgraph of $H\boxtimes P$.

    To help with clique labelling later, we will introduce some \defin{dummy} vertices to $G^+$ and $G$.  Since $G^+$ is an induced subgraph of $H\boxtimes P$, this also introduces edges and cliques in $G^+$.  For each vertex $(v,i)$ of $G^+$ with $i\in[h-1]$, if $G^+$ does not contain the vertex $(v,i+1)$ then we add the dummy vertex $(v,i+1)$ to $G^+$ and $G$ and define $\omega(v,i+1):=\omega(v,i)$.  The addition of dummy vertices does not increase the number of vertices or the total weight $\omega(G)$ by more than a factor of $2$.  Any vertex of $G$ that is not a dummy vertex is called a \defin{real} vertex.

    To satisfy the requirement on $g_1(n)$, the label $\mu((v,i))$ of a vertex $(v,i)$ of $G$ should have length $\log \omega(G)-\log\omega((v,i)) + \Oh(\sqrt{\log n})$.  To achieve this, we require that the extended signature $\sigma^+(v,i)$ have length at most $\log|S_i^+|-\log\omega((v,i)) + \Oh(\sqrt{\log n})$.  To accomplish this, we define an intermediate weight function $\delta$, as in the proof of \cref{small_height}.  For each $(w,j)\in V(H\boxtimes P)$, define $X^{-1}_{w,j}:=\{(v,i)\in V(G)\mid i\in\{j,j+1\},\; w\in X_{v,i}\}$ and define
    \[
      \delta(w,j):=\sum_{(v,i)\in X^{-1}_{w,j}}\omega(v,i) + \omega((w,j-1)) \enspace ,
    \]
    where we use the convention that $\omega((w,j-1))=0$ if $(w,j-1)\not\in V(G)$.
    The intuition behind the sum in the definition $\delta(w,j)$ is that, for each vertex $(v,i)$ of $G$, the extended signature $\sigma^+(v,i)$ must have length at most $\log|S^+_i|-\log\omega((v,i))+\Oh(\sqrt{\log n})$.  However, $\sigma^+(v,i)$ should be long enough to contain $\sigma(w,i)$ for each $w\in X_{v,i}$.  Thus, $\delta$ increases the weight of $(w,j)$ to ensure this.
    Importantly, for each $(v,i)\in V(G)$ and each $w\in X_{v,i}$,
    \begin{equation}
      \delta(w,i) \ge \omega((v,i)) \enspace . \label{delta_v_omega}
    \end{equation}
    (The final term, $\omega(w,j-1)$ in the definition of $\delta(w,j)$ will be useful for clique labelling later.)

    The total weight assigned by $\delta$ is
    \begin{equation}
      \sum_{(w,j)\in V(H\boxtimes P)} \delta(w,j)
      \le \sum_{(v,i)\in V(G))}|X_{v,i}|\cdot\omega((v,i)) + \omega(G)
      \le (2k+3)\omega(G) \enspace .
      \label{total_delta}
    \end{equation}

    At this point, we proceed almost exactly as in the proof of \cref{shannon_fano_elias_tree}.
    For each $i\in[h]$ and each $v\in V_i$, define
    \[
      S_{\delta,i}(v):=\{(x_f(v),1),\,\ldots,\, (x_f(v),\delta(v,i)\}
    \]
    and define $S_{\delta,i}:=\bigcup_{v\in V_i}S_{\delta,i}(v)$.  For each $i\in[h]$, treat $S_{\delta,i}$ as a totally ordered set where order is determined by lexicographic comparison.  Note that $|S_{\delta,i}|=\delta(V(H)\times\{i\})$ for each $i\in[h]$ and $\sum_{i=1}^h |S_{\delta,i}|=\sum_{(w,j)\in H\boxtimes P}\delta(w,j)\le (2k+3)\omega(G)$, by \cref{total_delta}.

    In the following, we use the convention that a real interval $(a,b)$ contains $(x,j)\in S^+_i$ if (the real interval) $(a,b)$ contains (the real number) $x$. In particular, for a vertex $v\in V(H)$ and $i\in[h]$, $\mathdefin{f(v)\cap S^+_i}:=\{(x,j)\in S^+_i:x\in f(v)\}$. Then, just as in the unweighted case $x_{B_i}(v)=\lca_{B_i}(f(v)\cap S^+_i)$.

    At this point, we apply exactly the same  machinery used in \cite{gawrychowski.janczewski:simpler,dujmovic.esperet.ea:adjacency} beginning with the sets $S_{\delta,1},\ldots,S_{\delta,h}$ instead of the sets $S_1,\ldots,S_h$.
    This results in sets $S^+_{\delta,1},\ldots,S^+_{\delta,h}$ of total size $\sum_{i=1}^h |S^+_{\delta,i}|\in \Oh(k\cdot\omega(G))\in \Oh(kn)$, by \eqref{total_delta}, and since $\omega$ is nice.
    This has the desired effect because a B-tree of height $h$ has at most $6a^h$ leaves.
    Therefore,
    \begin{align*}
      \depth_{B_i}(x_{B_i}(v))
       & =
      \depth_{B_i}(\lca_{B_i}(f(v)\cap S_{\delta,i}))            \\
       & \le
      \depth_{B_i}(\lca_{B_i}(S_{\delta,i}(v)))                  \\
       & \le\height(B_i)-\log_a(\delta(v,i)/6)                   \\
       & \le \log_a|S^+_{\delta,i}| - \log_a\delta(v,i) + \Oh(1) \\
      \enspace .
    \end{align*}
    Furthermore, for any $(w,j)\in X_{v,i}$, by \eqref{delta_v_omega},
    \begin{align*}
      \depth_{B_i}(x_{B_i}(w))
       & \le \log_a|S^+_{\delta,i}| - \log_a\delta(w,i) + \Oh(1)
      \le \log_a|S^+_{\delta,i}| - \log_a\omega(v,i) + \Oh(1) \enspace .
    \end{align*}
    Therefore,
    \begin{align*}
      |\sigma^+(v,i)|
       & = \max\{|\sigma((w,i))| \mid w\in X_{v,i}\}                                              \\
       & \le \lceil \log (12a)\rceil\cdot\max\{\depth_{B_i}(x_{B_i}(w)) \mid w\in X_{v,i}\}       \\
       & \le \lceil \log (12a)\rceil\cdot (\log_a|S^+_{\delta,i}| - \log_a\omega((v,i)) + \Oh(1)) \\
       & \le (5+ \log a)\cdot (\log_a|S^+_{\delta,i}| - \log_a\omega((v,i)) + \Oh(1))             \\
       & \le \log |S^+_{\delta,i}| - \log\omega((v,i)) + \Oh(1+\log_a |S^+_{\delta_i}| + \log a)  \\
       & \le \log |S^+_{\delta,i}| - \log\omega((v,i)) + \Oh(\sqrt{\log n}) \enspace .
    \end{align*}
    Then
    \begin{align*}
      |\lambda_2(v,i)|
       & \le \log |S^+_{\delta,i}|-\log \omega((v,i)) + \Oh(\sqrt{\log n} + k(\log k+\log\log n))            \\
       & \le \log |S^+_{\delta,i}|-\log \omega((v,i)) + \Oh(\sqrt{\log n} + k(\log k+\log\log n)) \enspace .
    \end{align*}
    By \eqref{total_delta}, we have that $|\lambda_1((v,i))|\le \log\omega(G) - \log |S^+_{\delta,i}| + \Oh(\log\log n)$, so
    \begin{align*}
      |\mu((v,i))|
       & = |\lambda_1((v,i))|+|\lambda_2((v,i))| + \Oh(\log\log n)                                                                        \\
       & \le \log\omega(G) - \log |S^+_{\delta,i}| + \log |S^+_{\delta,i}|-\log \omega((v,i)) + \Oh(\sqrt{\log n} + k(\log k+\log\log n)) \\
       & = \log\omega(G) - \log\omega((v,i)) + \Oh(\sqrt{\log n} + k(\log k+\log\log n)) \enspace .
    \end{align*}
    For any fixed $k$, $|\mu((v,i))|=\log \omega(G)-\log \omega((v,i)) + \Oh(\sqrt{\log n})$, satisfying the requirement $g_1(n)\in\Oh(\sqrt{\log n})$. Since the vertex labels assigned by $\mu$ have the same format as those in \cite{alstrup.dahlgaard.ea:optimal,dujmovic.esperet.ea:adjacency}, the adjacency tester $A$ proceeds exactly as described in \cite{alstrup.rauhe:improved,dujmovic.esperet.ea:adjacency}.

    \paragraph{Clique labelling and identity testing:}
    Thus far, we have shown that the lengths of the vertex labels promised in \cref{gk_labelling_precise} are achievable. 
    What remains is to describe the clique labels, the local identifiers, and the identity testing procedure.

    Let $K$ be a clique in $G^+$ that contains no dummy vertices and let $V_K:=\{v\mid (v,j)\in K\}$.%
    \footnote{The requirement that $K$ contains no dummy vertices is justified by the fact that we are only required to provide labels for cliques that exist in the original input graph $G^+$.}
    Then $K\subseteq V(H)\times\{i-1,i\}$  and $V_i\supseteq V_K$ for some $i\in\{2,\ldots,h\}$.  Since $K$ is a clique in $G^+\subseteq H\boxtimes P$, $V_K$ is a clique in $H$.  Therefore, there exists $v^*\in V_K$ such that $V_K\subseteq \{v^*\}\cup N^+_H(v^*)$.

    If $(v^*,i)\in K$ then $\delta(v^*,i)\ge \omega((v^*,i))$. If $(v^*,i)\not\in K$ then $(v^*,i-1)\in K$ and, by the definition of $\delta$, $\delta(v^*,i)\ge \omega((v^*,i-1))$.%
    \footnote{This is precisely the reason for the $\omega((v,i-1))$ term in the definition of $\delta(v,i)$.}
    In either case, $\delta(v^*,i)\ge\min_{(v,j)\in K}\omega((v,j))$.  Since $K$ contains no dummy vertices, $(v^*,i)$ is a (possibly dummy) vertex of $G^+$.  Therefore, $\mu((v^*,i))=\langle \lambda_1(v^*,i),\lambda_2(v^*,i)\rangle$ is defined.  Furthermore, $X_{v^*,i}$ contains $V_K$.  To summarize, $\mu((v^*,i))$ is defined, and has length at most $\log n - \log\min_{(v,j)\in K}(\omega((v,j))) + \Oh(\sqrt{\log n})$.  Also, $\mu((v^*,i)))$ contains $\sigma^+(v^*,i)$, which contains $\sigma(w,i)$ as a prefix, for each $w\in V_K$.

    We can now mimic the adjacency labelling.  Define
    \begin{align*}
      \alpha^-(K)
       & := \langle\; \langle |\sigma(w,i)|, \varphi(w)\rangle \mid (w,i-1)\in K \;\rangle \\
      \alpha^0(K)
       & := \langle\; \langle |\sigma(w,i)|, \varphi(w)\rangle \mid (w,i)\in K \;\rangle   \\
      \alpha(K)
       & := \langle\; \alpha^-(K),\; \alpha^0(K)  \;\rangle \enspace .
    \end{align*}
    Since $|K|\le 2k+2$, $|\alpha(K)|\in\Oh(k\log\log n)$.  Define
    \[
      \mu(K) := \langle\; \lambda_1(v^*,i),\; \sigma^+(v^*,i),\; \alpha(K) \;\rangle \enspace .
    \]
    To see that $\mu$ is injective when restricted to cliques in $G^+$, suppose $\mu(K)=\mu(L)$ for two distinct cliques $K$ and $L$.
    Then their first parts match, which implies $i_K=i_L=i$ since the prefix-free code $\rho(i)$ inside $\lambda_1$ uniquely determines the $P$-coordinate $i$.
    They also have the same extended signature $\sigma^+:=\sigma^+(v^*_K,i)=\sigma^+(v^*_L,i)$ and the same third part $\alpha(K)=\alpha(L)$.
    The set of vertices in $K$ is uniquely determined by $i$, $\sigma^+$, and $\alpha(K)$: for each entry $\langle \ell, \varphi \rangle$ in $\alpha^-(K)$ (respectively, $\alpha^0(K)$), there is a vertex $(w,i-1)$ (respectively, $(w,i)$) in $K$ where $\varphi(w)=\varphi$ and $\sigma(w,i)$ is the prefix of $\sigma^+$ of length $\ell$.
    By \eqref{ps_identity}, the $H$-coordinate $w$ is uniquely identified by $i$, $\sigma(w,i)$, and $\varphi(w)$, which in turn uniquely determines each vertex $(w,j)$ of $K$.
    Since $\alpha(K)=\alpha(L)$, we conclude $K=L$. Thus $\mu$ is injective on cliques of $G^+$.

    Then,
    \begin{align*}
      |\mu(K)|
       & \le |\lambda_1(v^*,i)| + |\sigma^+(v^*,i)| + \Oh(k \log\log n)                                       \\
       & \le \log\omega(G) - \log\delta(v^*,i) + \Oh(\sqrt{\log n} + k(\log k+\log\log n))                    \\
       & \le \log\omega(G) - \log\min_{(v,j)\in K}\omega((v,j)) + \Oh(\sqrt{\log n} + k(\log k+\log\log n)) .
    \end{align*}
    For fixed $k$, $|\mu(K)|\le \log\omega(G) - \log\min_{(v,j)\in K}\omega((v,j)) + \Oh(\sqrt{\log n})$, satisfying the requirement $g_3(n)\in\Oh(\sqrt{\log n})$.

    \paragraph{Local identifiers}

    For each vertex $(u,j)\in K$, the local identifier for $(u,j)$ is a single integer $\kappa(K,(u,j))\in[2k+2]$ that indicates the position of the entry for $(u,j)$ in $\alpha(K)$.  This integer also implicitly encodes whether $j=i-1$  or $j=i$.  Thus $|\kappa(K,(u,j))| = \Oh(\log k)$.  For fixed $k$, this satisfies the requirement $g_2(n)\in\Oh(1)$.

    \paragraph{Identity testing}

    We now describe the identity tester $I$.  Let $K$ be a clique in $G^+$ containing only real vertices, let $(u,j)\in K$, and let $(v,i')$ be a real vertex of $G$.  Let $(v^*,i)$ be as in the description of $\mu(K)$, so that $K\subseteq V(H)\times\{i-1,i\}$ and therefore $j\in\{i-1,i\}$.   Given $\mu(K)=\langle\lambda_1(v^*,i),\sigma^+(v^*,i),\alpha(K)\rangle$, $\kappa(K,(u,j))$, and $\mu((v,i'))=\langle \lambda_1(v,i'),\lambda_2(v,i')\rangle$, the identity tester proceeds as follows. First, the tester examines $\lambda_1((v,i'))$ and $\lambda_1((v^*,i))$ to determine if $i'=i$ or $i'=i-1$.  If neither of these is the case, then $(u,j)\neq (v,i')$ so $I(\mu(K),\kappa(K,(u,j)),\mu((v,i'))):=0$ since $j\in\{i-1,i\}$.
    If $i'\in\{i-1,i\}$, then the tester locates the entry $\langle |\sigma(u,i)|,\varphi(u)\rangle$ in $\alpha(K)$ using $\kappa(K,(u,j))$. The location of this entry also determines whether $j=i-1$ or $j=i$.  Since the tester now knows the values of $j$ and $i'$, it can immediately determine that $I(\mu(K),\kappa(K,(u,j)),\mu((v,i'))):=0$ if $j\neq i'$.
    If $j=i'$, then the tester computes $\sigma(u,i)$ using $\sigma^+(v^*,i)$ and $|\sigma(u,i)|$. Recall that $\lambda_2(v,i')=\langle \sigma^+(v,i'), |\sigma(v,i')|, \varphi(v),\alpha(v,i'),\tau(v,i'),|\sigma(v,i'+1)|\rangle$.
  There are now two cases to consider:
  \begin{itemize}
    \item If $i'=i$ then the tester uses $\sigma^+(v,i')$ and $|\sigma(v,i')|$ to compute $\sigma(v,i')=\sigma(v,i)$.  If $(\sigma(v,i),\varphi(v))=(\sigma(u,i),\varphi(u))$ then $(v,i)=(u,i)$ and  $I(\mu(K),\kappa(K,(u,j)),\mu((v,i'))):=1$.  Otherwise, $I(\mu(K),\kappa(K,(u,j)),\mu((v,i'))):=0$.
    \item If $i'=i-1$ then the tester computes $\sigma^+(v,i)=J(\sigma^+(v,i'),\tau(v,i'))$ and uses this with $|\sigma(v,i'+1)|=|\sigma(v,i)|$ to compute $\sigma(v,i)$.  Again, the tester concludes by checking if $(\sigma(v,i),\varphi(v))=(\sigma(u,i),\varphi(u))$.
  \end{itemize}
  This concludes the proof.
\end{proof}

\end{document}